\documentclass[11pt,a4paper,oneside]{article}
\usepackage[top=3cm, bottom=3cm, left=2.5cm, right=2.5cm]{geometry}
\usepackage[english]{babel}
\usepackage[utf8x]{inputenc}
\usepackage[T1]{fontenc}

\usepackage{amsthm,amsmath,amssymb,mathrsfs,dsfont,mathtools}
\usepackage{bm}
\usepackage{graphicx}

\usepackage[colorinlistoftodos]{todonotes}
\usepackage{booktabs,subcaption}
\usepackage[colorlinks=true, allcolors=blue, hypertexnames=false]{hyperref}
\usepackage[all]{nowidow}
\usepackage{setspace}
\usepackage{caption}
\usepackage{lscape}

\usepackage{algorithm}
\usepackage{algpseudocode}

\usepackage{float}
\usepackage{tikz} 
\usepackage{lipsum}
\usepackage{adjustbox}
\usepackage{lineno} 
\usepackage{macros}
\usepackage{verbatim,amssymb,epsfig}  
\usepackage{amsmath, amsfonts, amstext,amsthm}
\usepackage{mathtools}
\usepackage{enumerate}
\usepackage{enumitem}
\usepackage{natbib}
\usepackage{graphicx}
\usepackage[title]{appendix}
\usepackage{subcaption}
\usepackage{booktabs} 
\usepackage{hyperref}
\usepackage{dsfont}
\hypersetup{
colorlinks = true,
allcolors = blue
}
\usepackage[font=small]{caption}
\usepackage{url} 
\usepackage{authblk}

\usepackage{multirow, makecell}
\definecolor{revision_color}{HTML}{008080}

\usepackage{authblk}
\title{Scalable and robust regression models for continuous proportional data}
\date{October 2025}
\author[1]{Changwoo J. Lee\thanks{changwoo.lee@duke.edu}}
\author[1]{Benjamin K. Dahl}
\author[2]{Otso Ovaskainen}
\author[1]{David B. Dunson} 
\affil[1]{Department of Statistical Science, Duke University}
\affil[2]{Department of Biological and Environmental Science, University of Jyväskylä}
\newtheorem{theorem}{Theorem}

\newtheorem{lemma}{Lemma}
\newtheorem{proposition}{Proposition}

\theoremstyle{definition}
\newtheorem{definition}{Definition}
\newtheorem{remark}{Remark}

\def\sinhc{\mathrm{sinhc}}
\def\pkg{\normalfont{p_\textsc{kg}}}

\begin{document}

\maketitle
\begin{abstract}
Beta regression is used routinely for continuous proportional data, but it often encounters practical issues such as a lack of robustness to misspecification of the beta distribution and sensitivity to outliers. We develop an improved class of generalized linear models starting with the continuous binomial (cobin) distribution and further extending to dispersion mixtures of cobin distributions (micobin). The proposed cobin regression and micobin regression models have attractive robustness, computation, and flexibility properties. A key innovation is the Kolmogorov-Gamma data augmentation scheme, which facilitates Gibbs sampling for Bayesian computation, including in hierarchical cases involving nested, longitudinal, or spatial data.
We demonstrate robustness, ability to handle responses exactly at the boundary (0 or 1), and computational efficiency relative to beta regression
in simulation experiments and through
analysis of the benthic macroinvertebrate multimetric index of US lakes using lake watershed covariates.
\end{abstract}

\noindent%
{\it Keywords: Bayesian; Data augmentation; Generalized linear model; Latent Gaussian model; Markov chain Monte Carlo.} 
\vfill

\newpage 
\section{Introduction}

Regression analysis of proportional data, corresponding to bounded continuous data in the unit interval $[0,1]$, is a common focus in many fields. Such data arise in diverse contexts, including rates or indices in economics \citep{Brazier2002-dz}, percentage of tissue area in medical imaging \citep{Peplonska2012-lb}, and rating scales for Alzheimer's disease \citep{Rosen1984-pn}. 
We are particularly motivated by ecological applications; \citet{Warton2011-jf} estimates $\sim 14\%$ of ecology articles involve data based on proportions that are not derived from counts, with examples including measurements of percent coverage \citep{Korhonen2024-ka} and ecological indices \citep{Lindholm2021-ir}. 

Generalized linear models (GLMs) \citep{Nelder1972-hn} and their variants are fundamental tools for regression analysis. GLMs have multiple appealing theoretical and computational properties \citep{McCullagh1989-oa}. For continuous proportional data, one of two strategies is typically taken. Firstly, one may transform the data to fall on the real line and then apply a Gaussian linear model. However, such transformation-based approaches complicate the interpretation of the results in the original scale and have problems when some observations are concentrated near the boundary of the support. Alternatively, beta regression is widely applied \citep{Ferrari2004-ni, Cribari-Neto2010-aa}, assuming beta-distributed response variables supported on the open interval (0,1). We refer to \citet{Douma2019-xo} for a review of beta regression and applications in ecological contexts. 

However, beta regression has several limitations.
First, the beta distribution is neither in the natural exponential family \citep{Morris1982-qz}, nor part of the broader family of exponential dispersion models \citep{Jorgensen1987-bq}, and thus beta regression does not strictly belong to the GLM class. This implies that the mean and dispersion parameters are not orthogonal to each other \citep{Ferrari2004-ni}, and the well-established properties of GLMs may not hold.
In particular, beta regression is known to be highly sensitive to outliers \citep{Bayes2012-rr}.
Beta regression also faces computational difficulties when considering extensions to accommodate complex dependence structures. This limitation is particularly relevant in Bayesian frameworks, which support flexible hierarchical extensions \citep{Bolker2009-tg}. 
Finally, the presence of exact 0s or 1s prevents the direct application of beta regression models. This issue is often bypassed by manipulating the data to be within the open interval (0,1) \citep{Smithson2006-ia}, but the results are often highly sensitive to such preprocessing \citep{Kosmidis2025-sh}.

The motivation of this article is to introduce a proper GLM approach to continuous proportional response data, without the need for data preprocessing and facilitating computation, including in complex settings involving random effects. With these goals in mind, we propose continuous binomial (cobin) regression, a name inspired by models based on the continuous Bernoulli distribution \citep{Loaiza-Ganem2019-dl,Quijano-Xacur2019-eg}. The cobin distribution is an exponential dispersion model \citep{Jorgensen1987-bq} with orthogonal mean and dispersion parameters, implying that the corresponding GLM has appealing properties, including consistency of maximum likelihood estimates (MLE) under distributional misspecification. 
Importantly, the exponential dispersion structure also leads to bounded score functions, which provides improved robustness against boundary-proximate outliers. 
For posterior sampling using the canonical link function, we propose a novel data augmentation strategy based on Kolmogorov-Gamma random variables, leading to a conditionally normal likelihood, similar to P\'olya-Gamma augmentation for logistic models \citep{Polson2013-gb}. We develop an efficient sampler for Kolmogorov-Gamma random variables. The proposed Gibbs sampler can easily accommodate extensions to Gaussian latent variables and random effects, including in spatial contexts. Furthermore, we show uniform ergodicity of the proposed MCMC algorithms, providing theoretical guarantees for efficient inference procedures.

We also propose micobin regression, an extension of cobin regression based on dispersion mixtures of continuous binomial (micobin) distributions. 
This improves flexibility and robustness by localizing the dispersion parameter \citep{Wang2018-pd}, and can handle exact 0s or 1s without modifying the data. 
Micobin regression handles non-structural boundary values, such as those introduced by measurement precision limits or rounding during preprocessing, and is distinct from models that explicitly model boundary values by assigning positive probability mass to 0s or 1s, such as zero/one inflated or censored models \citep{Ospina2012-bh, Kubinec2023-it, Kosmidis2025-sh}. 
We describe how micobin regression can be further extended to allow the dispersion parameter to systematically vary according to covariates. Through simulations, we show that micobin regression achieves better predictive performance under distributional misspecification and greater robustness to outliers compared to beta regression and regression models based on mixtures of beta distributions \citep{Hahn2008-ch, Bayes2012-rr}.

As an illustration, we analyze the benthic macroinvertebrate multimetric index (MMI) of US lakes, an index ranging from 0 to 100 (scaled by 0.01 throughout the article) that reflects the condition of lake macroinvertebrate communities \citep{Stoddard2008-fb}. We examine association between MMI and lake watershed covariates, and use them to predict MMI at unsampled lakes; see Figure~\ref{fig:mmidata}. Since accounting for spatial dependence is highly important in such spatial ecological data \citep{Guelat2018-us}, we fit spatial cobin and micobin regression models with latent Gaussian process random effects. We compare the results with those from spatial beta regression, highlighting the robustness and computational advantages provided by Kolmogorov-Gamma augmentation.

\begin{figure}
    \centering
    \includegraphics[width=\linewidth]{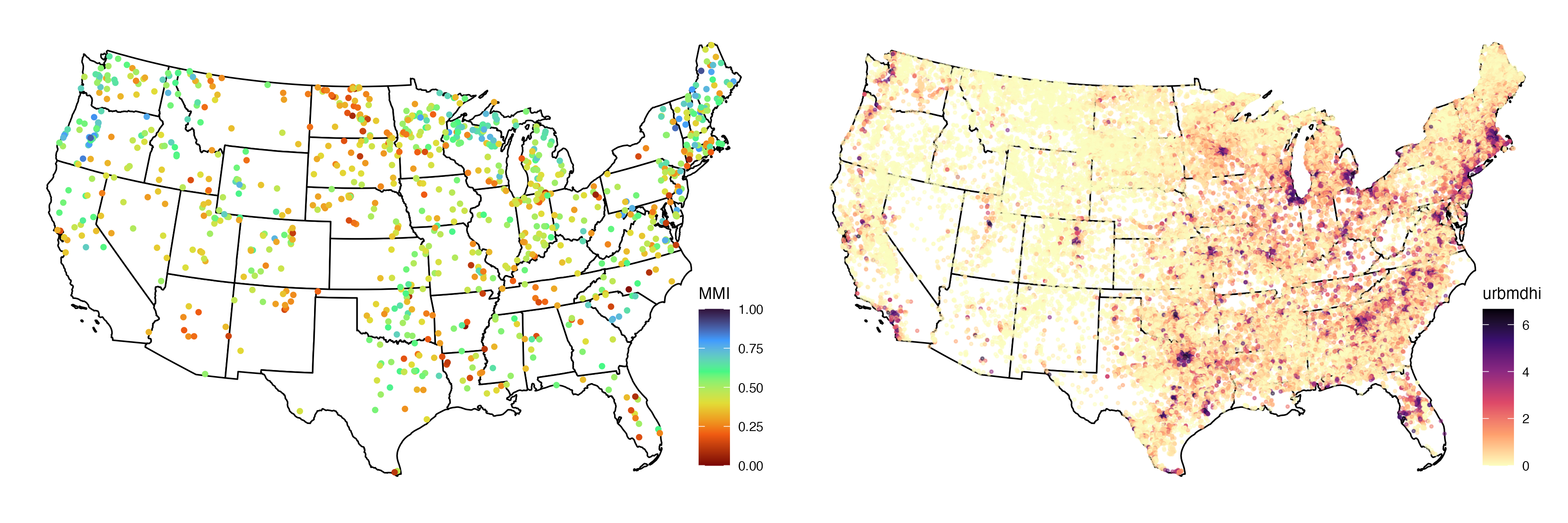}
    \caption{(Left) Benthic macroinvertebrate multimetric index of 949 lakes in conterminous US \citep{US-Environmental-Protection-Agency2022-dp} (Right) Urban land cover in the watersheds of 55,215 lakes, the log-transformed percentage of watershed area classified as developed, medium and high intensity land use (2016 National Land Cover Database classes 23, 24) \citep{Hill2018-eq}. }
    \label{fig:mmidata}
\end{figure}

The paper is structured as follows. In Section~\ref{sec:2_cobin}, we introduce cobin and micobin regressions, study their properties, and present hierarchical extensions. In Section~\ref{sec:3_kg}, we establish the Kolmogorov-Gamma integral identity in Theorem~\ref{thm:kgintidentity}, which could be of independent interest for other models, and describe our Kolmogorov-Gamma augmentation strategy for posterior computation. In Section~\ref{sec:5_simul} and \ref{sec:6_mmi}, we perform simulation studies to support our claims and present an application example. All proofs of theorems are in Supplementary Section~\ref{appendix:proofs}.

\section{Cobin and micobin regression models}
\label{sec:2_cobin}
\subsection{Continuous binomial distribution}

An exponential dispersion model \citep{Jorgensen1987-bq} is an extension of a one-parameter natural exponential family by adding a dispersion parameter that controls variance in a systematic way. Notable examples of such extensions include normal with fixed variance to normal distributions with unknown variance, exponential to gamma distributions, and geometric to negative binomial distributions. We refer to \citet{Jorgensen1992-ui} for a comprehensive review of exponential dispersion models and their properties.

Our goal is to find a family of distributions supported on the unit interval that serve as the random component of a GLM having appealing properties. Unlike binary regression where the conditional mean completely determines the distribution, it is desirable to have an additional dispersion parameter instead of just the natural parameter. We introduce the continuous binomial (cobin) distribution, which is in the exponential dispersion family and contains the uniform distribution as a special case. We first define the density of the cobin distribution. 
To avoid overloading notation, we adopt the convention that $(e^0 - 1)/0\coloneqq \lim_{x\to 0}(e^x - 1)/x = 1$ and $\sinh(0)/0\coloneqq \lim_{x\to 0}\sinh(x)/x = 1$, and we similarly fill in removable singularities in all functions that are otherwise infinitely differentiable on all of $\bbR$.

\begin{definition}[cobin] The continuous binomial (cobin) distribution with natural parameter $\theta\in\bbR$ and dispersion parameter $\lambda^{-1}\in \{1, 1/2, 1/3, \dots\}$, denoted as $Y\sim \mathrm{cobin}(\theta, \lambda^{-1})$, is a exponential dispersion model with density function
\begin{equation}
\label{eq:cobin}
    p_{\mathrm{cobin}}(y; \theta, \lambda^{-1}) = h(y, \lambda)\exp\left[\lambda \{\theta y- B(\theta)\}\right] = h(y, \lambda) \frac{e^{\lambda \theta y}}{\{(e^\theta-1)/\theta\}^\lambda}, \quad 0\le y \le 1,
\end{equation}
where $B(\theta) = \log\left\{(e^\theta-1)/\theta\right\}$ and $h(y, \lambda) = \frac {\lambda}{(\lambda-1)!}\sum _{k=0}^{\lambda}(-1)^{k}{\lambda \choose k}\{\max(\lambda y-k,0)\} ^{\lambda-1}$. 
The support is $[0,1]$ when $\lambda=1$ and $(0,1)$ if $\lambda \ge 2$.
\end{definition}

When $\theta=0$ and $\lambda=1$, the cobin distribution reduces to $\mathrm{Unif}(0,1)$.  
In Supplementary Section~\ref{appendix:cobinmicobindetail}, we provide a formal derivation of the cobin as an exponential dispersion model, including a proof that $\lambda$ must be an integer, along with comparable derivations for the normal, gamma, and inverse Gaussian families. 
To the best of our knowledge, the earliest reference to the continuous binomial distribution appears in \citet{Bates1955-mg} in the context of modeling time intervals between accidents, but has received little attention in the literature. 
The special case of $\lambda = 1$, with reparametrization via $\varphi = 1/(1+e^{-\theta})$, has been rediscovered under the name continuous Bernoulli distribution \citep{Loaiza-Ganem2019-dl}, since the density is proportional to $\varphi^y(1-\varphi)^{1-y}$. The name ``continuous binomial'' comes from the fact that the cobin distribution arises from the convolution of i.i.d. continuous Bernoulli random variables, $Y_1,\dots,Y_\lambda\iidsim  \mathrm{cobin}(\theta, 1)$ then $\lambda^{-1}\sum_{l=1}^\lambda Y_l\sim \mathrm{cobin}(\theta, \lambda^{-1})$, which follows from the convolution property of the exponential dispersion model \citep{Jorgensen1987-bq}. 

By the properties of the exponential family, the mean and variance can be derived directly by differentiating $B(\theta) = \log\left\{(e^\theta-1)/\theta\right\}$, 
\[
E(Y) = B'(\theta) = e^\theta/(e^\theta-1) - \theta^{-1}, \quad \var(Y) = \lambda^{-1}B''(\theta) = \lambda^{-1}\{(2-e^\theta-e^{-\theta})^{-1} + \theta^{-2}\},
\]
which illustrates that $E(Y)$ is only controlled by the natural parameter $\theta$ and $\var(Y)$ is proportional to the dispersion parameter $\lambda^{-1}$; if $\theta=0$ we have $E(Y)=1/2$ and $\var(Y)=1/(12\lambda)$. See Figure~\ref{fig:beta_cobin_micobin} for examples and a comparison with the beta distribution having the same mean and variance.

\begin{figure}
    \centering
    \includegraphics[width=\linewidth]{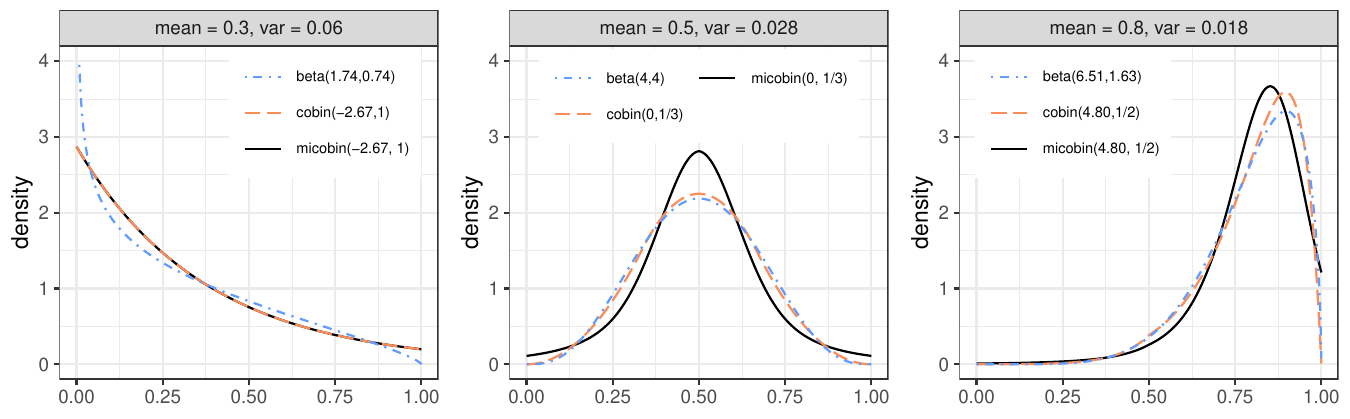}
    \caption{Comparison of beta, cobin, and mixture of cobin (micobin) with common mean and variance.}
    \label{fig:beta_cobin_micobin}
\end{figure}

\subsection{Cobin regression}

With $\bmx_i=(x_{ij})_{j=1}^p\in\bbR^p$ and a link function $g:(0,1)\to \bbR$ that is strictly monotone and differentiable, the GLM with continuous binomial response can be expressed as:  
\begin{equation}
\label{eq:cobinreg}
 Y_i \mid \theta_i, \lambda, \bmx_i \indsim \mathrm{cobin}(\theta_i, \lambda^{-1}), \quad  \theta_i = (B')^{-1}\{g^{-1}(\bm{x}_i^\T\bm\beta)\}, \quad i=1,\dots,n,
\end{equation}
which implies $g\{E(y_i\mid \bmx_i)\} = \bmx_i^\T\bm\beta$ for $\bm\beta\in\bbR^p$. We refer to the model \eqref{eq:cobinreg} as cobin regression. Denoting the linear predictor $\eta_i = \bm{x}_i^\T\bm\beta$ and the conditional mean $\mu_i = g^{-1}(\eta_i)$, the likelihood equations for $\bm\beta$ are 
\begin{equation}
\label{eq:cobinlikeq}
\frac{\partial}{\partial \beta_j}\sum_{i=1}^n\log p_{\mathrm{cobin}}(y_i; \theta_i, \lambda^{-1}) = \sum_{i=1}^n \frac{(y_i-\mu_i)x_{ij}}{\lambda^{-1}B''(\theta_i)}\frac{\partial \mu_i}{\partial \eta_i} = 0, \quad j=1,\dots,p
\end{equation}
where the solution $\hat{\bm\beta}$ does not depend on the dispersion parameter $\lambda^{-1}$. Standard procedures, such as Newton--Raphson or iteratively reweighted least squares, can be used to find the MLE $\hat{\bm\beta}$.  

Unlike the beta distribution, where inference on the mean parameter is based on the sufficient statistics $\{\sum_{i=1}^n\log(y_i), \sum_{i=1}^n\log(1-y_i)\}$ rather than $\sum_{i=1}^n y_i$, 
the cobin distribution belongs to an exponential dispersion model, and cobin regression is therefore a proper GLM \citep[][\S 2.2.2]{McCullagh1989-oa}. Thus, it inherits many attractive properties of GLMs, such as a concave log-likelihood function under the canonical link. We highlight the robustness property of its MLE against model misspecification, which implies that cobin regression under the correct mean structure produces an asymptotically valid point estimate even if $Y$ does not follow the cobin distribution. 
\begin{proposition}
\label{prop:qmle_consistency}
Under mild regularity conditions \citep{Gourieroux1984-ld}, $\hat{\bm\beta}$ of the cobin likelihood equations \eqref{eq:cobinlikeq} is consistent if $E(y_i\mid \bmx_i) = g^{-1}(\bm{x}_i^\T\bm\beta)$ is correctly specified. 
\end{proposition}

As shown in \citet{Gourieroux1984-ld}, obtaining a consistent estimator of $\bm\beta$ under potential distributional misspecification requires that $\hat{\bm\beta}$ is the solution to estimating equations associated with a natural (linear) exponential family. This condition is both necessary and sufficient; see also \S 4.2.6 and \S 8.3 of \citet[][]{Agresti2015-xu}. Thus, Proposition~\ref{prop:qmle_consistency} is a direct consequence of the cobin distribution belonging to an exponential dispersion model, where the dispersion parameter disappears in the estimating equations, satisfying the sufficient condition for consistency. Furthermore, the necessity of the condition implies that the MLE of beta regression is generally not consistent, unless the full distribution is correctly specified.

\begin{remark}
Quasi-likelihood approaches, which specify only the conditional mean and variance of $Y$ without requiring a full distributional form, also yield consistency results analogous to Proposition~\ref{prop:qmle_consistency}. A well-known example is the fractional logit or fractional probit models based on quasi-MLE \citep{Papke1996-gy}, whose quasi-score equation replaces the $\lambda^{-1}B''(\theta_i)$ term in equation \eqref{eq:cobinlikeq} with $\sigma^2\mu_i(1-\mu_i)$ for some $\sigma^2>0$. However, due to a lack of proper probabilistic formulation, quasi-likelihood approaches face challenges for uncertainty quantification and hierarchical extensions analogous to mixed-effects model settings, and therefore will not be pursued further here.
\end{remark}

Another important difference between the beta and cobin distributions is the form of the score function, the derivative of the log-likelihood with respect to $\bm\beta$. In contrast to the cobin regression score function \eqref{eq:cobinlikeq} which is bounded in terms of $y$, the beta regression score function for a single observation $y_i$ and the $j$th component of $\bm\beta$ is
\begin{equation}
\label{eq:betascore}
\frac{\partial}{\partial \beta_j}\log p_{\mathrm{beta}}\left(y_i; g^{-1}(\bmx_i^\T\bm\beta), \phi\right) = \phi\left\{\log\left(\frac{y_i}{1-y_i}\right)- \frac{\Gamma'(\mu_i \phi)}{\Gamma(\mu_i\phi)} + \frac{\Gamma'(\phi-\mu_i\phi)}{\Gamma(\phi-\mu_i\phi)}\right\}x_{ij} \frac{\partial \mu_i}{\partial \eta_i}, 
\end{equation}
where $p_{\mathrm{beta}}\left(y_i; g^{-1}(\bmx_i^\T\bm\beta), \phi\right)$ corresponds to the beta density evaluated at $y_i$ with mean-precision parametrization \citep{Ferrari2004-ni}, and $\Gamma$ is a gamma function. From \eqref{eq:betascore}, the beta regression score function depends directly on $\log\{y/(1 - y)\}$, which is unbounded for $y \in (0,1)$. As a result, $y$ near its boundary values can lead to an arbitrarily large change in the beta regression log-likelihood landscape. Given the fact that boundedness of the score function often serves as a necessary condition for robustness properties (e.g., Prop. 2 of \citet{Cantoni2001-xu}), cobin regression has greater robustness to observations near the boundaries than beta regression.

The canonical link function of cobin regression, which we call the cobit link, has inverse $g_{\mathrm{cobit}}^{-1}(\eta) = B'(\eta) = e^\eta/(e^\eta-1) - \eta^{-1}$. Although $g_{\mathrm{cobit}}(\mu)$ does not have a closed-form expression, it can be easily calculated in practice by solving $g_{\mathrm{cobit}}^{-1}(x)=\mu$ numerically. 
Compared to logistic function $1/(1+e^{-\eta})$, $g_{\mathrm{cobit}}^{-1}(\eta)$ slowly approaches 0 and 1 as $|\eta|\to\infty$, asymptotically at the same speed as the Cauchy c.d.f. See Supplementary Section~\ref{appendix:cobinmicobindetail} for details. 
Similarly to heavy-tailed link functions in binary response regression, this behavior is appealing in allowing the mean parameter to approach zero or one more slowly for extreme values of the predictors, reducing sensitivity to certain types of outliers. 

Under the canonical link function, the cobin regression model can be expressed as $Y_i \mid \bmx_i \indsim \mathrm{cobin}(\bm{x}_i^\T\bm\beta, \lambda^{-1})$ for $i=1,\dots,n$. In the following sections, we focus on cobin regression with canonical link functions, which greatly simplifies model fitting. Regardless of the choice of the link function, relationships between the conditional mean and covariates can be reported in terms of average slopes, also known as average marginal effects \citep{Williams2012-jz}. For example, the average slope of a continuous covariate $X_j$ under the link function $g$ is computed as
$
\frac{1}{n}\sum_{i=1}^n\frac{\partial E(y\mid \bm{x})}{\partial x_j}\big|_{\bm{x} = \bm{x}_i} = \frac{1}{n}\sum_{i=1}^n (g^{-1})'(\bmx_i^\T\bm\beta)\beta_j. 
$
For a categorical covariate, the partial derivative is replaced by a unit difference.

\subsection{Micobin regression and varying dispersion model}

Cobin regression comes with some limitations. First, the parameter $\lambda$ is only allowed to be an integer. This not only limits flexibility, but also introduces significant difficulties when we want to allow the dispersion parameter to vary systematically based on covariates. Next, the support of the cobin distribution is an open interval $(0,1)$ when $\lambda \ge 2$ and therefore cannot handle data at exact boundaries. Beta regression suffers from the same problem, and while ``nudging'' data into $(0,1)$ before analysis \citep{Smithson2006-ia} is a standard practice in the literature, doing so is certainly not desirable. 

To overcome these limitations, we introduce an extension of cobin regression, named micobin regression. Specifically, we propose using the cobin distribution dispersion mixture as a response distribution, defined as follows.

\begin{definition}[micobin]
    $Y$ follows a mixture of cobin distributions (micobin) with natural parameter $\theta \in \bbR$ and dispersion parameter $\psi \in (0,1)$, written as $Y\sim \mathrm{micobin}(\theta, \psi)$, if 
    \begin{equation}
    \label{eq:micobin}
       Y \mid \lambda \sim \mathrm{cobin}(\theta, \lambda^{-1}), \quad (\lambda-1)\sim \mathrm{negbin}(2, \psi), 
    \end{equation}
    where $\mathrm{negbin}(r,\psi)$ denotes a negative binomial with mean $r(1-\psi)/\psi$ and variance $r(1-\psi)/\psi^2$. The mean and variance are $E(Y) = B'(\theta)$ and $\var(Y) = \psi B''(\theta)$. 
\end{definition}

The micobin distribution is obtained by mixing the cobin distributions over the dispersion parameter $\lambda$, analogous to
obtaining Student's $t$ distribution as a scale mixture of normal distributions. See Figure~\ref{fig:beta_cobin_micobin} for comparisons with the cobin and beta distributions ($\psi=1$ corresponds to the limiting case). The mixture over the dispersion parameter preserves the mean structure and improves the robustness against outliers \citep{Wang2018-pd}. The choice of $(\lambda-1) \sim \mathrm{negbin}(2, \psi)$ leads to a property that $E(\lambda^{-1}) = \psi$, which implies $\var(Y) = \psi B''(\theta)$. Thus, $\psi$ plays a similar role as $\lambda^{-1}$ as a dispersion parameter, with the additional flexibility that $\psi$ can take any value between 0 and 1. 
The micobin distribution is supported on the closed interval $[0,1]$ for any choice of $\psi\in(0,1)$, with positive densities at the boundaries $p_{\mathrm{micobin}}(0; \theta, \psi) = \psi^2\theta/(e^\theta-1)>0$ and $p_{\mathrm{micobin}}(1; \theta, \psi) = \psi^2\theta e^\theta/(e^\theta-1)>0$. Thus, it accommodates boundary data without needing to arbitrarily manipulate boundary values to lie in the open interval $(0,1)$. 

A micobin regression model with link function $g$ is defined as $Y_i\mid \theta_i, \psi \indsim \mathrm{micobin}(\theta_i, \psi)$ with $\theta_i = (B')^{-1}\{g^{-1}(\bmx_i^\T\bm\beta)\}$ for $i=1,\dots,n$. In particular, using the canonical cobit link function $g_{\mathrm{cobit}}$, the micobin regression model can be compactly written as 
\begin{equation}
\label{eq:micobinreg}
    Y_i \mid \bmx_i \indsim \mathrm{micobin}(\bm{x}_i^\T\bm\beta,\psi), \quad i=1,\dots,n,
\end{equation}
which implies $g_{\mathrm{cobit}}\{E(y_i\mid \bmx_i)\} = \bm{x}_i^\T\bm\beta$.

\begin{remark}
Micobin regression does not assign positive probability mass to the boundary values, and is therefore distinct from zero/one-inflated or censored models \citep{Ospina2012-bh, Kubinec2023-it, Kosmidis2025-sh}. Unlike these two-part models, which treat boundary values as qualitatively different by introducing a latent mechanism to determine whether $y$ falls on the boundary, micobin regression employs a one-part modeling framework that directly models $E(y_i\mid \bm{x}_i)$ without conditioning on $y_i \in (0,1)$. 
Although the micobin model assigns zero probability to boundary values due to its continuous formulation, it is suitable for settings where boundary values arise from measurement precision limits or rounding of data during preprocessing. In such cases, including when there is no substantial reason to assume that boundary values arise from a distinct mechanism, two-part models may be unnecessarily complex or even misleading; see \citet{Ramalho2011-ro} for a related discussion between one-part and two-part models for continuous proportional data in the econometrics context.
\end{remark}

Micobin regression can be extended to allow the dispersion to vary systematically with covariates, which may differ from covariates in the mean regression component of the model
\citep{Smyth1989-qy}. 
The varying dispersion micobin regression, under the canonical link for mean and logit link for dispersion, can be written as
\begin{equation}
\label{eq:vardispmicobin}
    Y_i\mid x\bm_i, \psi_i \indsim \mathrm{micobin}(\bm{x}_i^\T\bm\beta, \psi_i), \quad  \mathrm{logit}(\psi_i) =  \bmd_i^\T\bm\gamma, \quad i=1,\dots,n,
\end{equation}
where $\bmd_i\in\bbR^d$ is a dispersion covariate that may overlap with $\bmx_i$, and $\bm\gamma\in\bbR^d$ is the coefficient. The choice of the logit link for dispersion is particularly attractive in terms of posterior computation, which enables P\'olya-Gamma data augmentation with the latent negative binomial model \citep{Pillow2012-qq}, see Supplementary Section~\ref{subsec:vardispmicobin} for details. 

An important distinction from the varying dispersion beta regression \citep{Smithson2006-ia} is that the formulation in \eqref{eq:vardispmicobin} does not allow bimodal densities of $Y$ with spikes at both zero and one. While cobin and micobin have a more restrictive range of $\var(Y)$ compared to the beta by excluding U-shaped densities, this restriction can be advantageous in many applications where the true conditional density of $Y$ given covariates is unlikely to be U-shaped. In contrast, U-shaped fitted distributions can arise in beta regression, especially in the varying dispersion model, due to lack of fit and sparse data.

\subsection{Hierarchical extensions with latent Gaussian structure}

When data come with multilevel, longitudinal, or spatial structure, cobin and micobin regression models can be naturally extended to mixed-effect models, and more generally to latent Gaussian models. For spatially indexed proportional data $y(s)$ and covariate $\bmx(s)$, we focus on spatial generalized linear mixed models \citep{Diggle1998-xw} with the cobit link,
\begin{equation}
\label{eq:spreg}
g_{\mathrm{cobit}}(E\{y(s)\mid \bmx(s), u(s)\}) = \bmx(s)^\T\bm\beta + u(s), \quad u(\cdot)\sim \mathrm{mean\,\,zero\,\,Gaussian\,\,process}.
\end{equation}
With a choice of response distribution such as cobin or micobin, this allows us to perform both inference for $\bm\beta$ and make predictions at new locations in a single modeling framework. 

Inference for non-Gaussian spatial models \eqref{eq:spreg} is often performed using Bayesian approaches implemented with MCMC, due to the challenges in inferring spatial random effects at many locations. However, generic sampling algorithms such as Metropolis-Hastings require careful tuning and face significant challenges when sampling high-dimensional latent parameters. 
In the following section, we propose a novel data augmentation scheme called Kolmogorov-Gamma augmentation, which converts cobin or micobin likelihoods into a conditionally normal likelihood and yields a simple Gibbs sampler that does not require tuning and with theoretical guarantees on uniform ergodicity.

\section{Inference with Kolmogorov-Gamma augmentation}
\label{sec:3_kg}
\subsection{Kolmogorov-Gamma distribution and integral identity}

First, we define Kolmogorov-Gamma random variables. 
\begin{definition}
\label{def:kgdef}
We say a positive random variable $\kappa$ follows a Kolmogorov-Gamma (KG) distribution with parameters $b>0$ and $c \in \bbR$, denoted as $\kappa \sim \mathrm{KG}(b,c)$, if 
\begin{equation}
\label{eq:kgdef}
    \kappa \stackrel{\rmd}{=} \frac{1}{2\pi^2}\sum_{k=1}^\infty \frac{\epsilon_k}{k^2 + c^2/(4\pi^2)}, \quad \epsilon_k \iidsim \mathrm{Gamma}(b, 1), \quad k=1,2,\dots,
 \end{equation}
where $\stackrel{\rmd}{=}$ denotes equality in distribution.
\end{definition}
For the case where $b=1$ and $c=0$, $\mathrm{KG}(1,0)$ scaled by $\pi^2$ corresponds to the squared Kolmogorov (or squared Kolmogorov-Smirnov) distribution, the infinite convolution of independent exponential distributions \citep[][\S 4]{Andrews1974-ms}. Since $\epsilon_k$ are gamma distributed, following a similar naming convention as P\'olya-Gamma \citep{Polson2013-gb}, we call $\kappa$ a Kolmogorov-Gamma random variable. The difference with the P\'olya-Gamma is the term $k^2$ in the denominator in \eqref{eq:kgdef}, which is $(k-0.5)^2$ for P\'olya-Gamma. 
From relationships between exponential and gamma distributions, $\mathrm{KG}(b,c)$ with integer $b$ is equal in distribution to the sum of $b$ independent $\mathrm{KG}(1,c)$ variables. 

Now we describe the key result. Under the canonical link, the cobin likelihood \eqref{eq:cobin} is proportional to $(e^{\eta})^{\lambda y}/\{(e^{\eta}-1)/\eta\}^\lambda$ in terms of the linear predictor $\eta = \theta = \bm{x}^\T\bm\beta$, which is not a well-recognized form in terms of $\eta$. We present the following Kolmogorov-Gamma integral identity, which establishes a direct connection between the cobin likelihood and Kolmogorov-Gamma random variables.

\begin{theorem}
\label{thm:kgintidentity}
    For any $a\in\bbR$ and $b>0$, the following integral identity holds for $\eta\in\bbR$: 
    \begin{equation}
    \label{eq:kgintidentity}
    \frac{(e^\eta)^a}{\{(e^{\eta}-1)/\eta\}^b} = e^{(a-b/2)\eta}\int_0^\infty e^{-\kappa \eta^2/2}\pkg(\kappa; b, 0) \rmd \kappa,        
    \end{equation}
    where $\pkg(\kappa; b,0)$ is the density of a $\mathrm{KG}(b,0)$ random variable. 
\end{theorem}

Theorem~\ref{thm:kgexptilting} shows that the conditional distribution with density $p(\kappa\mid \eta)\propto e^{-\kappa \eta^2/2}\pkg(\kappa; b, 0)$ is
$\mathrm{KG}(b,\eta)$, which arises from treating integrand of \eqref{eq:kgintidentity} as an unnormalized density of $\kappa$.
\begin{theorem} 
\label{thm:kgexptilting}
The Kolmogorov-Gamma random variable $\mathrm{KG}(b,c)$ in Definition~\ref{def:kgdef} coincides with the distribution arising from exponential tilting of the $\mathrm{KG}(b,0)$ random variable, with density equal to $\pkg(x; b,c) = \left\{\sinh(c/2)/(c/2)\right\}^{b}\exp(-c^2x/2)\pkg(x; b, 0)$, $x>0$. 
\end{theorem}

\subsection{Kolmogorov-Gamma augmentation and blocked Gibbs sampler}

Based on Theorem~\ref{thm:kgintidentity}, consider the augmented model by introducing variables $\bm\kappa=\{\kappa_i\}_{i=1}^n$,
\begin{equation}
\label{eq:auglik}
p(y_i, \kappa_i\mid \bmx_i, \bm\beta, \lambda) = h(y_i, \lambda) \exp\{{\lambda(y_i - 0.5)\bm{x}_i^\T\bm\beta}- \kappa_i (\bm{x}_i^\T\bm\beta)^2/2\}\pkg (\kappa_i; \lambda, 0),    
\end{equation}
for $i=1,\dots,n$, which reduces to original cobin regression upon marginalization $\int_0^\infty p(y_i, \kappa_i\mid x_i, \bm\beta, \lambda) \rmd \kappa_i = p_{\mathrm{cobin}}(y_i; \bm{x}_i^\T\bm\beta, \lambda^{-1})$. 
Then, the log-likelihood from \eqref{eq:auglik} conditional on $\bm\kappa$ and $\lambda$ becomes a quadratic form in terms of $\bm\beta$.  
In addition, the conditional distribution of $\kappa_i$ given $\bm\beta$ and $\lambda$ is $\mathrm{KG}(\lambda,\bm{x}_i^\T\bm\beta)$ according to Theorem~\ref{thm:kgexptilting}. 
Under a Bayesian framework with a normal prior on $\bm\beta$, this leads to straightforward Gibbs samplers. For micobin regression, the same augmentation strategy can be adopted by conditioning on the mixing variable $\bm\lambda = (\lambda_1,\dots,\lambda_n)$. 

Algorithms~\ref{alg:cobin} and \ref{alg:micobin} describe blocked Gibbs samplers for cobin and micobin regression, where $X\in\bbR^{n\times p}$ is the design matrix. We assume a normal prior $\bm\beta\sim N_p(0, \Sigma_\beta)$, some proper prior $\lambda \sim p_\lambda$ for cobin regression, and a beta prior $\psi\sim \mathrm{Beta}(a_\psi,b_\psi)$ for micobin regression. We set a large upper bound $L$ on $\lambda$, as the posterior probability of very large $\lambda$ is negligible in practice. The update steps for $\lambda$ and $\bm\kappa$ are blocked together, which improves mixing and avoids the evaluation of the density of KG in \eqref{eq:auglik} when updating $\lambda$. The beta prior for $\psi$ in micobin regression leads to a conditionally conjugate update. In Supplementary Section~\ref{appendix:algdetails}, we describe the detailed derivations and the application of the Kolmogorov-Gamma augmentation to an EM algorithm for estimating the posterior mode, and discuss sampling strategies for varying dispersion micobin regression models \eqref{eq:vardispmicobin} and spatial models \eqref{eq:spreg}.

\begin{algorithm}[t]
\footnotesize
\captionsetup{font=footnotesize} 
\caption{One cycle of a blocked Gibbs sampler for cobin regression \eqref{eq:cobinreg} with cobit link}
\label{alg:cobin}
\begin{algorithmic}[1]
\State Sample $\lambda$ from $\pr(\lambda = l \mid \bm\beta) \propto p_{\lambda}(l)\prod_{i=1}^n p_{\mathrm{cobin}}(y_i; \bm{x}_i^\T\bm\beta, l^{-1})$ among $\{1,\dots,L\}$ 
\State Sample $\kappa_i$ from $(\kappa_i \mid \lambda, \bm\beta) \indsim \mathrm{KG}(\lambda, \bm{x}_i^\T\bm\beta)$, $i=1,\dots,n$ \Comment{steps 1,2 jointly updates $(\lambda,\bm\kappa)$}
\State Sample $\bm\beta$ from $(\bm\beta \mid \lambda, \bm\kappa)\sim N_p(\bm{m}_\beta, V_\beta)$, where \vspace{-1mm}
\[
V_\beta^{-1} = X^\T \diag(\kappa_1,\dots,\kappa_n) X + \Sigma_\beta^{-1}, \quad \bm{m}_\beta =  V_\beta X^\T(y_1\lambda-0.5\lambda,\dots,y_n\lambda-0.5\lambda)^\T 
\]
\end{algorithmic}
\end{algorithm}

\begin{algorithm}[t]
\footnotesize
\captionsetup{font=footnotesize} 
\caption{One cycle of a blocked Gibbs sampler for micobin regression \eqref{eq:micobinreg} with cobit link}
\label{alg:micobin}
\begin{algorithmic}[1]
\State Sample $\lambda_i$ from $\pr(\lambda_i = l \mid \bm\beta,\psi) \propto l(1-\psi)^{l-1}p_{\mathrm{cobin}}(y_i; \bm{x}_i^\T\bm\beta, l^{-1})$ among $\{1,\dots,L\}$, $i=1,\dots,n$ 
\State Sample $\kappa_i$ from $(\kappa_i \mid \lambda_i, \bm\beta) \indsim \mathrm{KG}(\lambda_i, \bm{x}_i^\T\bm\beta)$, $i=1,\dots,n$ \Comment{steps 1,2 jointly updates $(\bm\lambda,\bm\kappa)$}
\State Sample $\bm\beta$ from $(\bm\beta \mid \bm\lambda, \bm\kappa)\sim N_p(\bm{m}_\beta, V_\beta)$, where \vspace{-1mm}
\[
V_\beta^{-1} = X^\T \diag(\kappa_1,\dots,\kappa_n) X + \Sigma_\beta^{-1}, \quad \bm{m}_\beta =  V_\beta X^\T(y_1\lambda_1-0.5\lambda_1,\dots,y_n\lambda_n-0.5\lambda_n)^\T 
\]
\State Sample $\psi$ from $(\psi \mid \bm\lambda) \sim \mathrm{Beta}(a_\psi + 2n, b_\psi - n + \sum_{i=1}^n\lambda_i)$ \Comment{steps 3,4 jointly updates $(\bm\beta,\psi)$}
\end{algorithmic}
\end{algorithm}

The proposed Kolmogorov-Gamma augmentation scheme offers several advantages. Due to the conditionally normal likelihood, the algorithms can be trivially extended to mixed-effects models and more complex hierarchical models involving latent Gaussian structures.  Moreover, by exploiting normal-normal conjugacy, the random and fixed effects can be updated jointly, which is especially important for spatial models \eqref{eq:spreg} where the spatial random effect is often highly correlated with the intercept. In addition, a normal prior for $\bm\beta$ can be easily replaced by normal mixture priors, such as Laplace, Cauchy, or more broadly global-local shrinkage priors \citep{Bhadra2016-ri}, simply by combining existing sampling methods developed for normal models. 

Furthermore, we show that the proposed Gibbs sampler for cobin and micobin regression based on Kolmogorov-Gamma augmentation is uniformly ergodic, meaning that the Markov chain converges to the posterior geometrically fast in terms of total variation distance, uniformly at any initialization. This guarantees the existence of central limit theorems for Monte Carlo averages of functions of $\bm\beta$. 

\begin{theorem}
\label{thm:uniformergodic}
The blocked Gibbs samplers presented in Algorithms~\ref{alg:cobin} and \ref{alg:micobin} for cobin and micobin regressions are uniformly ergodic. 
\end{theorem}

The proof follows a related approach to \citet{Choi2013-jv} based on a uniform minorization argument. Considering that popular approaches, such as Metropolis-adjusted Langevin or Hamiltonian Monte Carlo, rely on the assumption of log-concavity for their theoretical guarantees on fast mixing \citep{Dwivedi2019-ve, Chen2020-pe},
Theorem~\ref{thm:uniformergodic} is a strong result; for micobin regression the likelihood is not log-concave.

\subsection{Sampling Kolmogorov-Gamma random variables}

Developing a reliable and efficient sampling method for Kolmogorov-Gamma random variables is essential for our methodology. A naive approximation approach, based on truncating the sum of gamma random variables in Definition~\ref{def:kgdef}, is computationally inefficient and prone to truncation errors. We develop an efficient method for exact sampling of $\mathrm{KG}(1,c)$ random variates based on the alternating series method of \citet{Devroye1986-gj}. For integer $\lambda$, one can then sample $\mathrm{KG}(\lambda,c)$ variables by summing $\lambda$ independent $\mathrm{KG}(1,c)$ variates. 

We outline key steps involved in developing the exact $\mathrm{KG}(1,c)$ sampler; full details are provided in Supplementary Section~\ref{appendix:kgsampler}. 
It begins by identifying two alternating series representations of the $\mathrm{KG}(1,c)$ density. We then find an optimal proposal distribution that minimizes the expected number of loop iterations, which is a mixture of a truncated generalized inverse Gaussian (GIG) distribution and a truncated exponential distribution. Compared to the P\'olya-Gamma sampler implemented in the \texttt{BayesLogit} \texttt{R} package \citep{Polson2013-gb}, which took $\sim 0.3$ seconds to generate 1 million samples from P\'olya-Gamma(1,2), an \texttt{Rcpp} implementation of our Kolmogorov-Gamma sampler took $\sim 0.5$ seconds to generate 1 million samples from $\mathrm{KG}(1,2)$ in an Apple M1 CPU environment. 

The blocked Gibbs samplers described in Algorithms~\ref{alg:cobin} and~\ref{alg:micobin} require sampling $n$ KG random variables at each iteration, and exact sampling of a $\mathrm{KG}(\lambda, c)$ variate via convolution involves generating $\lambda$ i.i.d. $\mathrm{KG}(1, c)$ samples. Therefore, one of the key determinants of MCMC runtime other than dataset dimensions $(n,p)$ is the signal-to-noise ratio, which governs the parameter $\lambda$ (or latent $\{\lambda_i\}_{i=1}^n$ in micobin regression).

\section{Simulation studies}
\label{sec:5_simul}

\subsection{Robustness of MLE to model misspecification}

The first simulation study aims to empirically validate Proposition~\ref{prop:qmle_consistency}, focusing on the consistency and robustness of the maximum likelihood estimate of cobin regression models. We considered four different scenarios for the response distribution: (A) beta, (B) cobin, (C) mixture of beta and uniform distributions, known as the beta rectangular distribution \citep{Hahn2008-ch}, and (D) mixture of three beta distributions. Specifically, the densities of (C) and (D) with mean $\mu$ correspond to   
\begin{align}
\label{eq:betarec}
    p_{\mathrm{brec}}(y; \mu, \alpha, \phi) &= w(\mu,\alpha)p_{\mathrm{beta}}\left(y; (\mu - 0.5)/w(\mu,\alpha) + 0.5 , \phi\right) + 1-w(\mu,\alpha), \\
    p_{\mathrm{bmix}}(y; \mu, \phi) &= 0.25p_{\mathrm{beta}}(y; \mu-\epsilon(\mu), \phi) + 0.5p_{\mathrm{beta}}(y; \mu, \phi) + 0.25p_{\mathrm{beta}}(y; \mu+\epsilon(\mu), \phi), \nonumber
\end{align}
for $y\in(0,1)$, $w(\mu,\alpha) = 1-\alpha(1-|2\mu-1|)$ for some $\alpha\in(0,1)$, and $\epsilon(\mu) = \min(\mu, 1-\mu)/2$. 
We also consider two different link functions (cobit, logit) with sample sizes $n\in\{100,400,1600\}$, resulting in $4\times 2\times 3 = 24$ different data generating scenarios. We consider two covariates including the intercept and set the true coefficients as $\bm\beta = (\beta_0,\beta_1)=(0,1)$. The non-intercept covariate is generated as $x_i\iidsim N(0, \sigma_x^2)$, where we set $\sigma_x^2 = 9$ for cobit and $\sigma_x^2 = 1$ for logit to account for the difference in scales between the link functions reflecting that $g'_{\mathrm{cobit}}(0.5)$ and $g'_{\mathrm{logit}}(0.5)$ differ by a factor of 3. 
For parameters not related to the mean, we set $\phi = 8$ for beta, $\lambda = 3$ for cobin, $(\alpha,\phi) = (0.2, 10)$ for beta rectangular and $\phi = 40$ for mixture of beta. For the beta rectangular, the choice of $\alpha = 0.2$ ensures that the weight $w(\mu,\alpha)$ in \eqref{eq:betarec} assigned to the beta distribution is greater than 0.8. We simulate 1000 replicated datasets for each data generation scenario. 

For each data set, we fit the cobin and beta regression models to find MLE $\hat{\bm\beta}$ with unknown dispersion parameters, using the correct link function but not necessarily the true data generating distribution. We used iteratively reweighted least squares to find $\hat{\bm\beta}$ for cobin regression based on the \texttt{stats} package in \texttt{R} \citep{R-Core-Team2024-qx}, and we utilized the \texttt{betareg} \texttt{R} package \citep{Cribari-Neto2010-aa} for the beta regression estimate.

\begin{table}
\scriptsize
\caption{Analysis of beta and cobin regression MLE under four different data generating distributions: Beta, Cobin, beta rectangular (Betarec), and mixture of beta (Beta3mix). Results are based on 1000 replicates under the correct mean structure but possibly misspecified distributions.}
    \label{table:consistencycobitlogit}
    \centering
    \begin{tabular}{c c c c c c c c c c c}
    \toprule
  & &  & \multicolumn{2}{c}{Beta data}  & \multicolumn{2}{c}{Cobin data} & \multicolumn{2}{c}{Betarec data}  & \multicolumn{2}{c}{Beta3mix data}\\
     \cmidrule(lr){4-5}\cmidrule(lr){6-7} \cmidrule(lr){8-9}\cmidrule(lr){10-11}
   Link & Method & $n$ & Bias & RMSE & Bias & RMSE & Bias & RMSE & Bias & RMSE\\ 
     \midrule
 \multirow{6}{*}{cobit}   & \multirow{3}{*}{\makecell{beta\\regression}} & $100$ & 0.004 & 0.106 & -0.024 & 0.099  & -0.061  & 0.153 & -0.030 & 0.094 \\
   & & $400$ & -0.003 & 0.054 & -0.030 & 0.057 & -0.069 & 0.099 & -0.037 & 0.058 \\
   & & $1600$ & 0.002 & 0.026 & -0.030 & 0.038 &-0.076 & 0.084 & -0.036 & 0.042 \\
    \cmidrule(lr){2-11}
   & \multirow{3}{*}{\makecell{cobin\\regression}} & $100$ & 0.005 & 0.113 & 0.003 & 0.099 & 0.013 & 0.134 & 0.006 & 0.092\\
   & & $400$ & -0.002 & 0.056 & -0.001 & 0.049 & 0.006 & 0.070 & -0.001 & 0.046 \\
   & & $1600$ & 0.002 & 0.027 & 0.000 & 0.023 & -0.001 & 0.035 & 0.001 & 0.022\\
    \midrule
\multirow{6}{*}{logit} &\multirow{3}{*}{\makecell{beta\\regression}} & $100$ & 0.003 & 0.084 & -0.043 & 0.080 & -0.054 & 0.117 & -0.041 & 0.074 \\
   & & $400$ & 0.000 & 0.042 & -0.047 & 0.059 & -0.059 & 0.080 & -0.046 & 0.055\\
   & & $1600$ & 0.000 & 0.021 & -0.045 & 0.048 & -0.062 & 0.068 & -0.046 & 0.048\\
    \cmidrule(lr){2-11}
   & \multirow{3}{*}{\makecell{cobin\\regression}} & $100$ & 0.015 & 0.101& 0.005 & 0.066 & 0.020 & 0.116 & 0.005 & 0.067\\
   & & $400$ & 0.004 & 0.051& 0.000 & 0.035 & 0.007 & 0.062 & 0.000 & 0.033\\
   & & $1600$ & 0.000 & 0.026 & 0.001 & 0.016 & 0.001 & 0.032 &0.001 & 0.016\\
    \bottomrule
    \end{tabular}
    \begin{flushleft} 
{\footnotesize Monte Carlo standard errors of bias are all less than 0.004, and mean squared error (MSE) are all less than 0.001$\approx 0.031^2$ for $n=100$, 0.0004$=0.02^2$ for $n=400$, and  0.0002$ \approx 0.014^2$ for $n=1600$.}  
\end{flushleft}
\end{table}

The result is summarized in Table~\ref{table:consistencycobitlogit} in terms of bias and root mean square error (RMSE) of $\hat{\beta}_1$.  
As $n$ increases, the estimates from the cobin regression model exhibit a decreasing bias for any data-generating scenario, supporting the consistency described in Proposition~\ref{prop:qmle_consistency}. This stands in contrast to the results from the beta regression model, which show persistent bias and inconsistency under data-generating scenarios other than beta. For cobin regression results, the RMSE decreases proportionally to $n^{-1/2}$ for any data-generating scenario, which aligns with the asymptotic normality of $\hat{\bm\beta}$ in correctly specified or misspecified cases \citep{Gourieroux1984-ld}. 

Supplementary Section~\ref{appendix:simul_consistency} presents additional results, including a visualization of the data-generating distributions, the length of 95\% confidence intervals and empirical coverage probabilities, and predictive performance under even misspecified mean structures. These results further demonstrate that cobin regression is more robust than beta regression in terms of both inference and prediction across a range of misspecified settings.

\subsection{Robustness and scalability under spatial regression models}

We consider a spatial regression simulation study to analyze the robustness of cobin and micobin regression models to distributional misspecification. This simulation also aims to assess the scalability achieved through the Kolmogorov-Gamma augmentation, particularly for latent Gaussian hierarchical models. For data generation, we choose spatial locations uniformly at random from $[0, 1]^2$ as training and test locations with sizes $(n_{\mathrm{train}},n_{\mathrm{test}})\in\{(200,50), (400,100)\}$. Then, spatial random effects are generated from a mean zero Gaussian process (GP) with exponential kernel $\cov\{u(s),u(s')\}=\sigma_u^2\exp(-\|s-s'\|_2/\rho)$, where we set $\sigma_u^2 = 1$ and $\rho \in \{0.1, 0.2\}$. For the fixed effect terms, we consider two covariates including the intercept, generate a non-intercept covariate as $x(s_i)\iidsim N(0, 3^2)$, and set true coefficients as $(\beta_0^{\mathrm{true}},\beta_1^{\mathrm{true}}) = (0,1)^\T$. The responses $y(s_i)$ are generated from a beta rectangular distribution \eqref{eq:betarec} with cobit link $\mu(s_i) = g_{\mathrm{cobit}}^{-1}\{\beta_0^{\mathrm{true}}+ x(s_i)\beta_1^{\mathrm{true}} + u(s_i)\}$, $\alpha=0.2$, and $\phi = 10$. This data generation process is repeated 100 times. 

Based on the training set with size $n_{\mathrm{train}}$, we fit three different spatial regression models (beta, cobin, micobin) with correct mean structure \eqref{eq:spreg} but none of the response distributions are correctly specified. We set a $N_2(\bm{0}_2, 100^2\bfI_2)$ prior on the regression coefficients $\bm\beta$ and a half-Cauchy prior on the spatial random effects standard deviation $p(\sigma_u)\propto (1+\sigma_u^2)^{-1}$, $\sigma_u>0$. The spatial range parameter $\rho\in\{0.1, 0.2\}$ is fixed at the true value. For spatial cobin and micobin regression models, we employ a blocked Gibbs sampler with Kolmogorov-Gamma augmentation. 
For the spatial beta regression model, we use \texttt{Stan} to carry out posterior inference using the No-U-Turn Sampler algorithm \citep{Carpenter2017-lo}. We run a total of 6,000 MCMC iterations and record wall-clock running time, with the first 1,000 samples discarded as burn-in. %

\begin{table}
\scriptsize
\caption{Bayesian spatial regression simulation results based on 100 replicates, under the correct mean structure but misspecified distributions.}
    \label{table:sim_spatial}
    \centering
    \begin{tabular}{c c c c c c c c c}
    \toprule
\multirow{2}{*}{\makecell{Spatial\\dependence}} & & & \multicolumn{2}{c}{Inference $(\hat\beta_1)$} & \multicolumn{2}{c}{Prediction} &  \multicolumn{2}{c}{Sampling $(\bm\beta)$}  \\
    \cmidrule(lr){4-5}\cmidrule(lr){6-7} \cmidrule(lr){8-9}
&  Method  & $(n_{\mathrm{train}},n_{\mathrm{test}})$ & Bias  & RMSE & NTLL & MSPE$\times 10^2$ & mESS & time (min) \\ 
     \midrule
    & \multirow{2}{*}{\makecell{beta\\ regression}} & $(200,50)$  & -0.048 & 0.118 & -0.325 & 0.427 & 919.8 & 44.5 \\
    & & $(400,100)$ & -0.052 & 0.089 & -0.354 & 0.345 & 978.7 & 437.7 \\
    \cmidrule(lr){2-9}
  \multirow{2}{*}{\makecell{$\rho=0.1$\\(moderate)}}  & \multirow{2}{*}{\makecell{cobin\\
    regression}} & $(200,50)$  & 0.005 & 0.093 & -0.340 & 0.388 & 2791.3 & 2.0 \\
    & & $(400,100)$ & 0.005 & 0.067 & -0.372 & 0.323 & 3220.9 & 11.2 \\
     \cmidrule(lr){2-9}
    & \multirow{2}{*}{\makecell{micobin\\
    regression}} & $(200,50)$  & 0.034 & 0.099 & -0.367 & 0.373 & 1908.4 & 2.4 \\
    & & $(400,100)$ & 0.037 & 0.074 & -0.394 & 0.312 & 2137.5 & 11.7 \\
\midrule
    & \multirow{2}{*}{\makecell{beta\\
    regression}} & $(200,50)$  & -0.065 & 0.120 & -0.320 & 0.329 & 1187.2 & 96.3 \\
    & & $(400,100)$  & -0.052 & 0.095 & -0.350 & 0.248 & 808.0 & 933.4 \\
     \cmidrule(lr){2-9}
 \multirow{2}{*}{\makecell{$\rho=0.2$\\(strong)}}  & \multirow{2}{*}{\makecell{cobin\\
    regression}} & $(200,50)$  & 0.000 & 0.088 & -0.346 & 0.306 & 3366.0 & 2.2 \\
    & & $(400,100)$  & 0.013 & 0.078 & -0.370 & 0.233 & 3663.9 & 12.1 \\
     \cmidrule(lr){2-9}
    & \multirow{2}{*}{\makecell{micobin\\
    regression}} & $(200,50)$  & 0.039 & 0.092 & -0.373 & 0.293 & 2265.3 & 2.2 \\
    & & $(400,100)$  & 0.050 & 0.091 & -0.395 & 0.226 & 2575.4 & 12.7 \\ 
    \bottomrule
    \end{tabular}
\begin{flushleft} 
{\footnotesize NTLL, negative test log-likelihood; MSPE, mean square prediction error; mESS, multivariate effective sample size. Monte Carlo standard errors are all less than 0.015 for NTLL, 0.013 for MSPE, 127.2 for mESS.}  
\end{flushleft}
\end{table}

Table~\ref{table:sim_spatial} summarizes the results. First, with respect to the posterior mean estimate $\hat{\beta}_1$ of the fixed effect coefficient, cobin regression consistently produces the lowest bias and RMSE in all scenarios, in accordance with the previous simulation results. In contrast, beta and micobin regression estimates exhibit bias, with the magnitude of these biases increasing as spatial dependence becomes stronger, which is partly attributable to beta and micobin not belonging to the exponential dispersion model. 
Second, to evaluate predictive performance under model misspecification, we compare the negative test log-likelihood conditional on random effects (negtestLL) and mean square prediction error (MSPE) based on test data. Micobin regression outperforms the others, while beta regression performs the worst on both metrics, highlighting the robustness of micobin regression in prediction. Finally, in terms of computational efficiency, cobin and micobin achieve a significantly higher multivariate effective sample size (mESS) \citep{Vats2019-wo} of $\bm\beta$ per unit time compared to beta regression. These results highlight the scalability of both cobin and micobin regressions in hierarchical settings, with cobin showing more reliable parameter estimation, consistent with the behavior suggested by Proposition~\ref{prop:qmle_consistency} in a simpler setting, and micobin offering greater robustness in predictive performance under potential distributional misspecification. Supplementary Section~\ref{appendix:simul_spatial} contains further details of the spatial regression simulation settings, including priors for dispersion parameters, as well as additional summaries including 95\% credible interval length and empirical coverage probabilities of $\beta_1$.

\subsection{Robustness to the presence of outliers}

We conduct a third simulation study to systematically assess the robustness of the cobin and micobin models to outliers, focusing on the stability of both parameter estimates and predictions. We considered two types of datasets: one without outliers, denoted as $\calD = \{(y_i,\bm{x}_i)\}_{i=1}^n$ with size $n=500$, and one with a single outlier, denoted as $\calD^\circ = \calD\cup \{(y^\circ,\bm{x}^\circ)\}$ with size $501$. We set $p=3$ predictors including an intercept, where covariates in $\calD$ are generated as $(x_{i1},x_{i2})\iidsim N_2(\bm{0}, 3^2\bfI_2)$, $i=1,\dots,n$.
We considered two different response distributions for $\calD$, (A) beta and (B) cobin, under the same mean structure $g_{\mathrm{cobit}}\{E(y_i \mid \bmx_i)\} = \beta_0 + \beta_1x_{i1} + \beta_2 x_{i2}$. The true regression coefficients were set as  $(\beta_0^{\mathrm{true}},\beta_1^{\mathrm{true}},\beta_2^{\mathrm{true}})^\T = (-6,1,0)^\T$, and for dispersion parameters we set $\phi = 17$ for beta and $\lambda = 6$ for cobin. We then introduced an outlier with covariate $\bmx^\circ=(6,6)^\T$ and a response value of either $y^\circ=10^{-2}$ (setting 1) or $y^\circ=10^{-3}$ (setting 2), resulting in four different data settings (A1, A2, B1, B2). Since $E(y \mid \bm{x}=\bmx^\circ) = 0.5$ under the true model without outlier, the response $y^\circ=10^{-2}$ or $y^\circ=10^{-3}$ are highly unlikely, with $y^\circ=10^{-3}$ being more extreme. This data generation process was repeated 100 times. 

For each setting with and without outliers, we fit four different fixed-effects models (beta, cobin, beta rectangular, micobin) with the mean structure $g_{\mathrm{cobit}}\{E(y_i \mid \bmx_i)\} = \beta_0 + \beta_1x_{i1} + \beta_2 x_{i2}$ for $i=1,\dots,n$. We denote $\hat\beta_j^\circ$ and $\hat\beta_j$ $(j=0,1,2)$ as posterior mean estimates of $\beta_j$ based on the dataset with outlier ($\calD^\circ$) and without outlier $(\calD)$. For all four models, we set a $N_3(\bm{0}_3, 100^2\bfI_3)$ prior on the regression coefficients $\bm\beta$. For the cobin and micobin regression models, posterior computation is performed with Kolmogorov-Gamma augmentation based on Algorithm~\ref{alg:cobin} and \ref{alg:micobin}. For beta regression and beta rectangular regression, we use \texttt{Stan} with the No-U-Turn Sampler. We ran a total of 6,000 MCMC iterations and discarded the first 1,000 samples as a burn-in.

\begin{table}
\scriptsize
\caption{Outlier robustness simulation results based on 100 replicates in terms of stability of parameter estimates $|\Delta\hat{\beta}_j|$, empirical coverage, and stability of linear predictor $\|\Delta\hat\eta_{1:n}\|_2$.}
\label{tab:outliersim}
\centering
\begin{tabular}{c c c c c c c c c c}
\toprule
\multirow{2}{*}{Method} &  
\multirow{2}{*}{Setting} &  
\multirow{2}{*}{$|\Delta\hat{\beta}_0|$}  & 
\multirow{2}{*}{$|\Delta\hat{\beta}_1|$} & 
\multicolumn{2}{c}{Coverage of $\mathrm{CI}_{.95}(\beta_1)$} & 
\multirow{2}{*}{$|\Delta\hat{\beta}_2|$} & 
\multicolumn{2}{c}{Coverage of $\mathrm{CI}_{.95}(\beta_2)$} & 
\multirow{2}{*}{$\|\Delta\hat\eta_{1:n}\|_2$} \\
\cmidrule(lr){5-6} \cmidrule(lr){8-9}
 & & & & with $y^\circ$ & without $y^\circ$ & & with $y^\circ$ & without $y^\circ$ & \\
\midrule
\multirow{2}{*}{\makecell{beta\\regression}}  
  & A1 & 0.095 & 0.061 & 74.0\% & \multirow{2}{*}{94.0\%}  & 0.056 & 69.0\% & \multirow{2}{*}{96.0\%} & 5.948 \\
  & A2 & 0.145 & 0.091 & 51.0\% &  & 0.083 & 39.0\% &  & 8.852 \\
\cmidrule(lr){1-10}
\multirow{2}{*}{\makecell{cobin\\regression}}  
  & A1 & 0.007 & 0.017 & 95.0\% & \multirow{2}{*}{97.0\%} & 0.020 & 95.0\% & \multirow{2}{*}{95.0\%} & 1.768 \\
  & A2 & 0.008 & 0.016 & 95.0\% & & 0.020 & 96.0\% & & 1.762 \\
\cmidrule(lr){1-10}
\multirow{2}{*}{\makecell{betarec\\regression}}  
  & A1 & 0.039 & 0.007 & 93.0\% & \multirow{2}{*}{92.0\%} & 0.001 & 95.0\% & \multirow{2}{*}{96.0\%} & 0.987 \\
  & A2& 0.038 & 0.006 & 93.0\% &  & 0.001 & 96.0\% & & 0.964 \\
\cmidrule(lr){1-10}
\multirow{2}{*}{\makecell{micobin\\regression}}  
  & A1 & 0.007 & 0.005 & 91.0\% & \multirow{2}{*}{90.0\%} & 0.006 & 93.0\% & \multirow{2}{*}{95.0\%} & 0.601 \\
  & A2 & 0.006 & 0.005 & 92.0\% & & 0.006 & 93.0\% &  & 0.562 \\
\midrule
\multirow{2}{*}{\makecell{beta\\regression}}  
  & B1 & 0.099 & 0.062 & 41.0\% &  \multirow{2}{*}{88.0\%}& 0.058 & 54.0\% &  \multirow{2}{*}{92.0\%} & 6.136 \\
  & B2 & 0.150 & 0.093 & 17.0\% & & 0.085 & 28.0\% & & 9.111 \\
\cmidrule(lr){1-10}
\multirow{2}{*}{\makecell{cobin\\regression}}  
  & B1 & 0.004 & 0.017 & 96.0\% &  \multirow{2}{*}{96.0\%} & 0.020 & 91.0\% &  \multirow{2}{*}{96.0\%} & 1.770 \\
  & B2 & 0.005 & 0.017 & 97.0\% & & 0.021 & 93.0\% &  & 1.807 \\
\cmidrule(lr){1-10}
\multirow{2}{*}{\makecell{betarec\\regression}}  
  & B1& 0.044 & 0.007 & 77.0\% &  \multirow{2}{*}{79.0\%} & 0.001 & 91.0\% &  \multirow{2}{*}{91.0\%} & 1.114 \\
  & B2 & 0.045 & 0.007 & 75.0\% & & 0.001 & 91.0\% & & 1.124 \\
\cmidrule(lr){1-10}
\multirow{2}{*}{\makecell{micobin\\regression}}  
  & B1 & 0.005 & 0.003 & 93.0\% &  \multirow{2}{*}{91.0\%} & 0.004 & 94.0\% &  \multirow{2}{*}{96.0\%} & 0.357 \\
  & B2 & 0.005 & 0.003 & 93.0\% & & 0.004 & 95.0\% & & 0.368 \\
\bottomrule
\end{tabular}
\begin{flushleft}
{\footnotesize Monte Carlo standard errors are all less than 0.003 for $|\Delta\hat{\beta}_j|$ ($j=0,1,2$), 5.0\% for coverage probabilities, and 0.098 for $\|\Delta\hat{\eta}_{1:n}\|_2$.}
\end{flushleft}
\end{table}

We summarize simulation results in three different ways. First, we compute $|\Delta \hat\beta_j|=|\hat{\beta}_j^\circ-\hat\beta_j|$ to quantify the stability of the parameter estimates in the presence of an outlier. 
Also, we compare the empirical coverage of the 95\% credible intervals for $\beta_1^\circ, \beta_1$ (with $\beta_1^{\mathrm{true}}=1$) and $\beta_2^\circ, \beta_2$ (with $\beta_2^{\mathrm{true}}=0$) to assess how outliers affect both point estimates and uncertainty quantification. 
Finally, we assess the stability of the linear predictor using  $\|\Delta \hat{\eta}_{1:n}\|_2 = \left[\sum_{i=1}^n(\hat\eta_{i}^\circ-\hat\eta_{i})^2\right]^{1/2}$, where $\hat\eta_i = \hat\beta_0 + \hat\beta_1x_{i1}+\hat\beta_2x_{i2}$ and $\hat\eta_i^\circ = \hat\beta_0^\circ+ \hat\beta_1^\circ x_{i1}+\hat\beta_2^\circ x_{i2}$. 

 Table~\ref{tab:outliersim} demonstrates that beta regression is the most sensitive to outliers in all aspects. In terms of stability of regression coefficients and linear predictors, beta regression exhibits poor performance, especially showing a stark difference between settings where $y^\circ = 10^{-2}$ and $y^\circ = 10^{-3}$. This sharp contrast is also reflected in the empirical coverage; all coverage values fall below 74\% in Setting 1 ($y^\circ = 10^{-2}$) and below 51\% in Setting 2 ($y^\circ = 10^{-3}$). In contrast, the cobin regression result is virtually unaffected by outliers. 
Both beta rectangular and micobin regressions, which are based on mixture formulations, offer generally better stability than cobin regression. Between the two, micobin regression shows greater robustness in terms of the stability of the linear predictor. This is partly because the intercept estimate in beta rectangular regression is relatively non-robust (see $|\Delta\hat{\beta}_0|$ column), since the beta rectangular model assumes outliers arise from a uniform distribution with fixed mean 1/2, which tends to highly influence the intercept estimate. 

Overall, cobin regression provides clear advantages over beta regression in terms of robustness to outliers, even without relying on mixture formulations. Micobin regression achieves even greater robustness, with similar improvements. However, the beta rectangular model suffers from instability in the intercept estimate, which leads to degraded predictive performance compared to the micobin regression model. Supplementary Section~\ref{appendix:simul_outlier} presents additional results, including visualizations of the data-generating distributions, prior distributions for parameters not related to the mean, and summaries of posterior mean estimates $\hat{\beta}_j^\circ$, $\hat{\beta}_j$ and sampling efficiency. These additional results further support the robustness of cobin and micobin regression methods in the presence of outliers.

\section{Benthic macroinvertebrate multimetric index of U.S. lakes}
\label{sec:6_mmi}

As an illustrative application of cobin and micobin regression models, we analyze the benthic macroinvertebrate multimetric index (MMI) of US lakes and the association with lake watershed covariates. MMI, also known as an index of biotic integrity, is a standard quantitative measure for the bioassessment of macroinvertebrate assemblages \citep{Karr1991-sq, Stoddard2008-fb} that integrates various attributes of the assemblage (e.g. taxonomic composition and richness). Higher MMI values indicate a healthier and more diverse benthic macroinvertebrate community. We consider MMI data from the 2017 National Lake Assessment Survey (NLA) \citep{US-Environmental-Protection-Agency2022-dp}, which covers about 1,000 lakes in the conterminous US; see the left panel of Figure~\ref{fig:mmidata}. 

We are interested in understanding how the biotic integrity of lake ecosystems measured by MMI is associated with natural and human-related lake watershed characteristics, as well as in predicting the MMI of unsurveyed lakes. We consider LakeCat data \citep{Hill2018-eq} covering more than 380,000 US lakes, containing various natural and anthropogenic watershed covariates. For illustrative purposes, we select 7 watershed covariates that are highly important in the analysis of lake eutrophication \citep[][Fig. 7]{Hill2018-eq}, as well as 2 additional covariates (manure application and urban land cover); see Table~\ref{table:mmicovariate} in Supplementary Section~\ref{appendix:mmi_data} for description. Since all covariates were heavily right-skewed and exhibited limited variation around the mean, we applied the transformation $x \mapsto \log_2(x+1)$ to reduce skewness and mitigate the influence of very large values.

\begin{table}[t]
    \scriptsize
    \caption{Lakes with the 3 lowest MMI values (multiplied by 0.01 from the original 0--100 scale)}
    \centering
    \begin{tabular}{c c c c c c c c c}
    \toprule
    COMID & Longitude & Latitude & Lake name, State & MMI & Mean MMI (10 nearest) & $\mathcal{D}^\star$ & $\mathcal{D}^\circ$ & $\mathcal{D}$ \\
    \midrule
    22721231 & -81.893 & 33.300 & Brierpatch Lake, GA & 0     & 0.534 & \checkmark &       &       \\
    9201925  & -80.263 & 35.116 & Jones Pond, NC      & 0.020 & 0.460 & \checkmark & \checkmark &     \\
    22845861 & -92.273 & 34.538 & Ferguson Lake, AR & 0.021 & 0.414 & \checkmark & \checkmark &  \\
    \bottomrule
    \end{tabular}
    \label{table:lakelowmmi}
\vspace{-3mm}
    \begin{flushleft}
    {\scriptsize
COMID: unique lake identifier; checkmark indicates inclusion in the respective dataset ($\mathcal{D}^\star$, $\mathcal{D}^\circ$, or $\mathcal{D}$).
}
    \end{flushleft}
\end{table}
After excluding lakes with missing MMI and/or covariates, a total of 950 lakes remained. Among 950 lakes, MMI values range from 0 to 0.904 (mean 0.452; median 0.450, Q1 0.340, Q3 0.564), with one lake at 0 and two near the lower boundary (0.020, 0.021). To assess robustness, we analyze three datasets: (1) the full dataset $(\calD^\star, n = 950)$; (2) one lake with zero MMI removed $(\calD^\circ, n = 949)$; and (3) three lakes removed $(\calD, n = 947)$, additionally excluding the two low MMI lakes. As shown in Table~\ref{table:lakelowmmi}, these lakes exhibit exceptionally low MMI compared to their surroundings, potentially due to localized environmental degradation such as waste discharge or nearby construction activity. 
In fact, beta and cobin regression fit results on $\calD^\circ$ reveal that the two low MMI lakes (0.020 and 0.021) exhibit poor goodness-of-fit based on residual diagnostics, suggesting they are acting as potential outliers under these models (see Supplementary Section~\ref{appendix:qresidual} for details). By comparing the analysis results with and without these potential outliers, we aim to evaluate the robustness of the proposed cobin and micobin regression models and compare them with beta regression.

We fit spatial beta and cobin regression models with the cobit link \eqref{eq:spreg} to $\mathcal{D}^\circ$ and $\mathcal{D}$, and spatial beta rectangular and micobin regression models to all three datasets, since these models can accommodate boundary values. 
We selected 55,215 lakes from the LakeCat dataset for prediction, focusing on those with surface areas greater than 40,000 m$^2$. 
Since making probabilistic predictions for these 55,215 lakes using a traditional GP prior is computationally prohibitive, we employ a nearest neighbor Gaussian process (NNGP) \citep{Datta2016-dx} on the spatial random effect. Since MMI values are unavailable for unsampled lakes, we assess predictive performance using WAIC  \citep{Gelman2014-ga} and leave-one-out cross-validation estimate using Pareto smoothed importance sampling (PSIS-LOO) \citep{Vehtari2024-qa}, both conditional on random effects.  We use a prior and algorithm similar to that in Section~\ref{sec:5_simul} for spatial cobin and micobin regression, and \texttt{Stan} for spatial beta and beta rectangular regression; see Supplementary Section~\ref{appendix:mmi_data} for details. We ran three chains for each model for a total of 6,000 iterations per chain, discarding the first 1,000 samples from each as burn-in. The convergence diagnosis with Gelman-Rubin statistic is less than 1.02 for all model parameters, indicating satisfactory convergence.

\begin{table}[t]
\scriptsize
\caption{Estimated fixed-effect coefficients ($\hat{\bm\beta}^\circ$) from four spatial regression models fitted on $\mathcal{D}^\circ$ ($n = 949$) and 95\% credible intervals. Bold entries indicate non-intercept coefficients whose 95\% credible intervals exclude zero in $\mathcal{D}^\circ$ ($n = 949$); asterisks indicate the same on fits to $\mathcal{D}$ ($n = 947$)}
\label{table:mmiresult_final}
\centering
\begin{tabular}{lcccc}
\toprule
  & Beta & Cobin & Beta rectangular & Micobin \\
\midrule
(Intercept) & -2.363 (-4.160, -0.553) & -2.106 (-3.859, -0.345) & -1.609 (-3.370, 0.149) & -1.797 (-3.551, -0.085) \\
agkffact    & -2.586$^*$(-5.584, 0.330) & \textbf{-2.888}$^*$ (-5.714, -0.003) & \textbf{-3.041}$^*$ (-5.902, -0.202) & \textbf{-3.457}$^*$ (-6.113, -0.800) \\
bfi         & \textbf{0.343} (0.016, 0.672) & 0.293 (-0.022, 0.614) & 0.206 (-0.110, 0.523) & 0.229 (-0.082, 0.548) \\
cbnf        & 0.165 (-0.081, 0.412) & 0.182 (-0.055, 0.420) & 0.201 (-0.034, 0.438) & 0.191 (-0.035, 0.425) \\
conif       & 0.081$^*$ (-0.002, 0.164) & \textbf{0.093}$^*$ (0.011, 0.176) & \textbf{0.106}$^*$ (0.025, 0.187) & \textbf{0.124}$^*$ (0.044, 0.203) \\
crophay     & -0.079 (-0.250, 0.091) & -0.063 (-0.231, 0.106) & -0.036 (-0.204, 0.128) & -0.054 (-0.213, 0.105) \\
fert        & -0.073 (-0.310, 0.158) & -0.092 (-0.323, 0.132) & -0.128 (-0.356, 0.097) & -0.082 (-0.300, 0.138) \\
manure      & -0.048 (-0.202, 0.102) & -0.036 (-0.182, 0.115) & -0.018 (-0.166, 0.132) & -0.029 (-0.173, 0.118) \\
pestic97    & -0.014 (-0.106, 0.075) & -0.021 (-0.108, 0.067) & -0.035 (-0.121, 0.053) & -0.025 (-0.108, 0.059) \\
urbmdhi     & \textbf{-0.181}$^*$ (-0.288, -0.076) & \textbf{-0.170}$^*$ (-0.273, -0.067) & \textbf{-0.166}$^*$ (-0.271, -0.062) & \textbf{-0.142}$^*$ (-0.243, -0.041)\\
\midrule
mESS$(\bm\beta)$/t & 3099.1/128 mins & 4719.2/4 mins & 3894.9/118 mins & 3095.6/5 mins \\
$\|\hat{\bm\beta}^\circ - \hat{\bm\beta}\|_2$ & 0.743 & 0.444 & 0.098 & 0.069 \\
PSIS-LOO & -1084.4$(\calD^\circ)$, -1125.2$(\calD)$ & -1095.5$(\calD^\circ)$, -1134.7$(\calD)$ & -1109.3$(\calD^\circ)$, -1124.8$(\calD)$ & -1115.4$(\calD^\circ)$, -1126.9$(\calD)$ \\
WAIC     & -1093.4$(\calD^\circ)$, -1134.7$(\calD)$ & -1103.5$(\calD^\circ)$, -1142.4$(\calD)$ & -1114.6$(\calD^\circ)$, -1131.3$(\calD)$ & -1119.3$(\calD^\circ)$, -1131.1$(\calD)$ \\
\bottomrule
\end{tabular}
\vspace{-2mm}
\begin{flushleft}
{\scriptsize
agkffact, soil erodibility factor; bfi, base flow index; cbnf, cultivated biological N fixation; conif, coniferous forest cover; crophay, crop/hay land cover; fert, synthetic N fertilizer use; manure, manure application; pestic97, 1997 pesticide use; urbmdhi, medium/high-density urban land cover; mESS$(\bm\beta)$, t, multivariate effective sample size and wall-clock fitting time (in mins), averaged over 3 chains; PSIS-LOO, leave-one-out cross-validation estimate using Pareto smoothed importance sampling; WAIC, widely applicable information criterion. All variables are $\log_2(x+1)$ transformed. 
}
\end{flushleft}
\end{table}

We summarize the posterior mean estimate of fixed-effect coefficients from all four models based on $\calD^\circ$ in Table~\ref{table:mmiresult_final}, where non-intercept coefficients with 95\% credible intervals excluding zero are emphasized in bold. 
Supplementary Section~\ref{appendix:mmi_results} provides corresponding estimates based on $\calD^\star$ and $\calD$, as well as the average slopes (based on $\calD^\circ$) for better interpretation of the results. 
The estimated signs of coefficients are generally sensible, where decreased biotic integrity (low MMI) in US lakes is associated with high soil erodibility, low base flow index, low coniferous forest cover, and high urban land cover in the lake watershed area. 

Table~\ref{table:mmiresult_final} also highlights coefficients with asterisks whose 95\% credible intervals exclude zero when the model is fitted to the dataset $\mathcal{D}$ without potential outliers. 
Notably, the beta regression yields different conclusions between $\mathcal{D}^\circ$ and $\mathcal{D}$, whereas cobin, beta rectangular, and micobin regression results remain consistent. This pattern aligns with the previous outlier robustness simulation, where beta regression exhibited poor stability and empirical coverage of $\bm\beta$, whereas micobin regression shows the highest stability. In terms of computational efficiency, both spatial cobin and micobin regressions based on Kolmogorov-Gamma augmentation outperform spatial beta and beta rectangular regression, achieving approximately a 20-fold improvement in mESS of $\bm\beta$ per unit time.

In environmental applications, it is also important to evaluate predictive performance in the presence of extreme observations rather than discarding them as mere anomalies. Lakes with low MMI likely reflect genuine, localized ecological degradation, and it is important that models accurately learn from such cases to make reliable predictions, particularly because similarly degraded lakes likely exist among the unsampled lakes. From this perspective, spatial micobin regression achieves the strongest predictive performance on $\mathcal{D}^\circ$ ($n = 949$) with the lowest PSIS-LOO and WAIC. When the two low MMI lakes are excluded, spatial cobin regression achieves the best predictive performance on $\mathcal{D}$ ($n = 947$), as its lighter tails compared to micobin yield sharper probabilistic predictions in the absence of outliers. 
Figure~\ref{fig:mmi_micobinpred} illustrates the predicted MMI at 55,215 lakes from micobin regression fit based on $\calD^\circ$, with the right panel reflecting increased uncertainty in southern California and western Texas due to sparse MMI data. Supplementary Section~\ref{appendix:mmi_results} presents corresponding predictions from the beta, cobin, and beta rectangular regression models, along with their differences from the micobin predictions. 
Overall, the results suggest that the cobin and micobin regressions are more robust to observations near the boundaries, exhibit significantly better scalability, and achieve better predictive performance.

\begin{figure}[t]
    \centering
    \includegraphics[width=\linewidth]{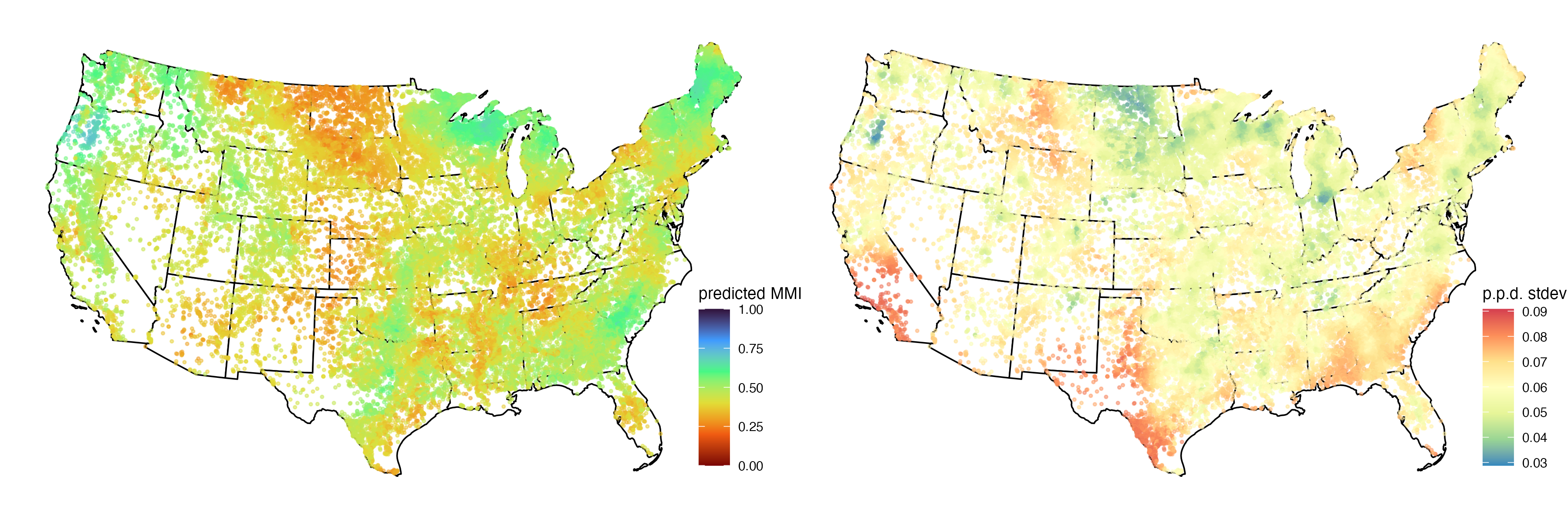}
    \caption{Predicted MMI at 55,215 lakes from the spatial micobin regression model based on $\calD^\circ$ in terms of $E\{Y(s^*)\mid X(s^*)\}$ at unsampled location $s^*$. (Left) Posterior predictive mean of $E\{Y(s^*)\mid X(s^*)\}$. (Right) Posterior predictive standard deviation of $E\{Y(s^*)\mid X(s^*)\}$.}
    \label{fig:mmi_micobinpred}
\end{figure}

\section{Discussion}
\label{sec:7_discussion}

While the proposed cobin and micobin models demonstrate strong robustness properties and introduce a new framework for modeling continuous proportional data, our further investigation also revealed challenges in posterior computation under certain scenarios. 
When data are highly concentrated near the boundaries of the unit interval, the Markov chain based on Kolmogorov-Gamma augmentation can exhibit poor mixing behavior, analogous to known issues with P\'olya-Gamma augmentation under highly imbalanced binary data \citep{Johndrow2019-os}. Detailed theoretical and empirical results illustrating this phenomenon are provided in Supplementary Section~\ref{appendix:imbalanced}. Developing improved strategies to address these challenges, potentially through calibrated data augmentation \citep{Duan2018-zl} or parameter expansion techniques \citep{Zens2024-nb}, represents an important direction for future work.

\color{black}
Beyond generalized linear (mixed) models, we anticipate that cobin and micobin distributions can be naturally incorporated into a more diverse family of models with continuous proportional data, such as tree ensembles \citep{Stoddard2008-fb} or deep generative models \citep{Loaiza-Ganem2019-dl}. Extensions to multivariate outcomes, such as analyzing compositional data that lie on the simplex, also represent a promising future research avenue.
Developments of scalable inference methods are key to enabling such extensions in large-scale settings. Similar to how P\'olya-Gamma augmentation is connected with variational inference for logistic models \citep{Durante2019-sd}, we hope that the proposed Kolmogorov-Gamma augmentation also provides insight for the future development of approximate Bayesian inference methods, in addition to facilitating computation using MCMC.

\section*{Data availability statement}
Code to reproduce the analyses: \url{https://github.com/changwoo-lee/cobin-reproduce}.

\section*{Acknowledgement}
This research was partially supported by the National Institutes of Health (grant ID R01ES035625), by the European Research Council under the European Union’s Horizon 2020 research and innovation programme (grant agreement No 856506), by the National Science Foundation (NSF IIS-2426762), and by the Office of Naval Research (N00014-24-1-2626).

\bibliography{ref.bib}
\bibliographystyle{apalike}

\clearpage

\begin{center}
\LARGE{Appendices}
\end{center}

Section~\ref{appendix:kgsampler} contains sampling details from the Kolomogorov-Gamma $(1,c)$ distribution and the pseudocode of the algorithm. 
Section~\ref{appendix:proofs} contains proofs of statements presented in the main article. 
Section~\ref{appendix:cobinmicobindetail} contains a detailed derivation of cobin as an exponential dispersion model, density and cumulative distribution functions of cobin and micobin, and a discussion on the link and variance functions. 
Section~\ref{appendix:algdetails} contains derivations of Gibbs samplers as well as an EM algorithm. 
Section~\ref{appendix:simul} contains details of the simulation settings and additional results. 
Finally, Section~\ref{appendix:mmi} gives further information about benthic macroinvertebrate multimetric index data, analysis settings, and additional results.

\setcounter{section}{0}
\setcounter{equation}{0}
\setcounter{table}{0}
\setcounter{proposition}{0}
\setcounter{lemma}{0}
\setcounter{algorithm}{0}
\renewcommand{\theequation}{S.\arabic{equation}}
\renewcommand{\thesection}{S.\arabic{section}} 
\renewcommand{\theproposition}{S.\arabic{proposition}} 
\renewcommand{\thetheorem}{S.\arabic{theorem}} 
\renewcommand{\thecorollary}{S.\arabic{corollary}} 
\renewcommand{\thelemma}{S.\arabic{lemma}} 
\renewcommand{\thetable}{S.\arabic{table}} 
\renewcommand{\thefigure}{S.\arabic{figure}}
\renewcommand{\thealgorithm}{S.\arabic{algorithm}}

\section{Kolmogorov-Gamma sampler}
\label{appendix:kgsampler}
\subsection{Kolmogorov-Gamma density}
We first describe the density of a $\mathrm{KG}(1,0)$ random variable, which can be easily derived from the two different density representations of the Kolmogorov distribution. 
\begin{proposition}
\label{prop:kg10density}
The $\mathrm{KG}(1,0)$ density has two different alternating series representations, 
\begin{align*}
\pkg(x; 1,0) = \sum_{n=0}^\infty(-1)^n a_n^L(x)  &= \sum_{n=0}^\infty (-1)^n a^R_n(x), \quad x>0,\\
a_n^L(x) = \begin{cases}
\frac{2}{\pi^{1/2}(2x)^{3/2}} \exp(-\frac{n^2}{8x}) & (n\,\,\mathrm{odd})\\
\frac{(n+1)^2}{\pi^{1/2}(2x)^{5/2}}\exp(-\frac{(n+1)^2}{8x}) &(n\,\,\mathrm{even}) 
\end{cases},&\quad a_n^R(x)=  4\pi^2 (n+1)^2\exp\{-2\pi^2(n+1)^2 x\}.
\end{align*}
\end{proposition} 
Having two different density representations is crucial for developing a sampling scheme with the alternating series method. From Theorem~\ref{thm:kgexptilting}, the density of $\mathrm{KG}(1,c)$ is $\pkg(x; 1,c) = \left\{\sinh(c/2)/(c/2)\right\}\exp(-c^2x/2)\pkg(x;1,0)$. Therefore, $\pkg(x; 1,c)$ can also be represented as an alternating series $\pkg(x;1,c) = \sum_{n=0}^{\infty} (-1)^n a_n(x;c,t)$ where we define $a_n(x;c,t)$ as
\begin{equation}
\label{eq:anseq}
a_n(x;c,t) = \begin{cases}\left\{\sinh(c/2)/(c/2)\right\}\exp(-c^2x/2)a_n^L(x),  & 0<x< t, \\
\left\{\sinh(c/2)/(c/2)\right\}\exp(-c^2x/2)a_n^R(x),  & t\le x,
\end{cases}
\end{equation}
for some suitable choice of cutoff point $t>0$, which will be discussed next.

\subsection{Sampling KG(1,c) using alternating series method}

The alternating series method \citep{Devroye1986-gj} is an effective sampling algorithm when the target density $p(x)$ is computationally expensive to evaluate but can be approximated from above and below by a sequence of envelope functions $\{S_m(x)\}_{m=0}^\infty$, satisfying $S_0(x) > S_2(x) >\dots > p(x) >\dots > S_3(x)>S_1(x)$. With such an envelope function in hand, exact sampling of $X\sim p$ can be achieved with the following steps: (1) draw $X\sim q$ from the proposal distribution $q$, (2) generate $U\sim \mathrm{Unif}(0,Mq(X))$ where $\|p/q\|_\infty \le M$, (3) repeat until $U\le S_m(X)$ for odd $m$ or $U> S_m(X)$ for even $m$, (4) accept $X$ if $m$ is odd, repeat from (1) again if $m$ is even.  

Similar to \citet{Devroye1986-gj} and \citet{Polson2013-gb}, our choice of envelope function for sampling $\mathrm{KG}(1,c)$ is $S_m(x) = \sum_{n=0}^m(-1)^n a_n(x;c,t)$. 
For this choice to be a valid envelope function, $a_n(x;c,t)$ must be monotonically decreasing in $n$ for any fixed $x>0$ and parameter $c$. The following Lemma~\ref{lemma:trange} shows the valid range of $t$, an intersection between the range of $x$ where $a_n^L(x)$ and $a_n^R(x)$ are both monotonically decreasing in $n$. 
\begin{lemma} 
\label{lemma:trange}
For any fixed $x>0$ and $c\in\bbR$, the sequence $a_n(x;c,t)$ defined in \eqref{eq:anseq} is monotonically decreasing in $n$ when $\log (2)/(3\pi^2) < t < 0.25$, where $\log (2)/(3\pi^2)\approx 0.0234$. 
\end{lemma}
In what follows, we assume that $t$ satisfies Lemma~\ref{lemma:trange}; the optimal choice of $t$ will be discussed soon. Since $S_0(x) = a_0(x;c,t) > \pkg(x; 1,c)$ for all $x$, a natural choice of proposal $q$ that ensures $\|\pkg(\cdot;1,c)/q(\cdot)\|_\infty\le M$ for some $M$ is $q(x;c,t) = M^{-1} a_0(x;c,t)$. Plugging in $a_0^L(x)$ and $a_0^R(x)$ in \eqref{eq:anseq}, the proposal distribution with density $q(x; c,t)$ is
\begin{equation}
\label{eq:kg1cproposal}
    X \sim \begin{cases}
        \mathrm{GIG}(-1.5, c^2, 1/4) 1(0<X< t)\quad &\text{with prob. }A^L(c,t)/\{A^L(c,t) + A^R(c,t)\}\\
       \mathrm{Exp}(c^2/2 + 2\pi^2)1(t\le X) &\text{with prob. }A^R(c,t)/\{A^L(c,t) + A^R(c,t)\} \\
    \end{cases} 
\end{equation}
where $\mathrm{GIG}(p,a,b)1(0<X\le t)$ is a generalized inverse Gaussian (GIG) distribution with density proportional to $x^{p-1}\exp\{-(ax+b/x)/2\}$ truncated to $(0,t)$, $\mathrm{Exp}(a)1(t\le X)$ is an exponential distribution with rate $a$ truncated to $[t, \infty)$, $A^L(c,t)=\int_0^t a_0(x;c,t) \rmd x$, and $A^R(c,t)=\int_t^\infty a_0(x;c,t) \rmd x$. Compared to the P\'olya-Gamma case \citep{Polson2013-gb}, which involves inverse Gaussian and exponential distributions in its proposal, our proposal is slightly more complicated, involving GIG and exponential distributions. \citet{Pena2025-ky} obtain results on the exact evaluation of the cdf and the sampling of GIG random variables that have a half-integer parameter $p=-1.5$. We use these results in calculating $A^L(c,t)$ and sampling truncated GIG random variables.  

Our $\mathrm{KG}(1,c)$ sampling algorithm is highly efficient. We show this by first investigating how often a proposal $X\sim q$ from \eqref{eq:kg1cproposal} is accepted, based on the expected number of outer loop iterations $M = \int_0^\infty a_0(x;c,t) \rmd x = A^L(c,t)+A^R(c,t)$, and then by illustrating that with very high probability, only few series terms $S_m(x)$ need to be computed in the inner loop in order to decide whether to accept or reject $X$. 

\begin{proposition}
\label{prop:optimal-t}
The following statements hold:
\begin{enumerate}
    \item The best cutoff point $t^*$ minimizing the expected number of outer loop iterations $M = A^L(c,t)+A^R(c,t)$ is independent of $c$; this value is $t^*\approx 0.050239$.
    \item Using the best cutoff point $t^*$, $M$ is bounded above by 
    $1.1456$ and the expected number of inner loop iterations is bounded above by $1.1275$ for any given $c$.
\end{enumerate}
\end{proposition}

\subsection{Pseudocode for KG(1,c) sampler}
\label{appendix:pseudocode}
\begin{algorithm}[t]
\caption{Sampling from $\text{KG}(1, c)$}
\begin{algorithmic}[1]
\State \textbf{Input:} Parameter $c$, cutoff value $t\in (0.0234,0.25)$ (optimal $t^*=0.050239$)
\State $\tilde{A}^L \leftarrow (|c|+2) \exp(-|c|/2) \texttt{pGIG}(t \mid p=-3/2, a = c^2, b=1/4)$ \Comment{proportional to $A^L(c,t)$} 
\State $\tilde{A}^R \leftarrow 4\pi^2 \exp\{- ( 2\pi^2 + c^2 / 2)t\}/(2\pi^2 + c^2 / 2)$ \Comment{proportional to $A^R(c,t)$} 
\Repeat
    \State Generate $U, V \sim U(0, 1)$
    \If{$U < \tilde{A}^R / (\tilde{A}^L + \tilde{A}^R)$} \Comment{$\sinh(c/2)/(c/2)$ are canceled out thus not calculated}
        \State $X \leftarrow t + E / (2\pi^2 + c^2 / 2)$ where $E \sim \text{Exp}(1)$ \Comment{Truncated exponential}
    \Else
            \Repeat \Comment{Truncated GIG}
                \State $X \sim \mathrm{GIG}(p=-3/2, a=c^2, b=1/4)$ \Comment{If $c=0$, $X \sim \mathrm{InvGamma}(3/2,1/8)$}
            \Until{$X < t$}
    \EndIf
    \State $S \leftarrow a_0(X)$, $Y \leftarrow VS$, $m \leftarrow 0$
    \Repeat
        \State $m \leftarrow m + 1$
        \If{$m$ is odd}
            \State $S \leftarrow S - a_m(X)$; if $Y < S$, then return $m$
        \Else
            \State $S \leftarrow S + a_m(X)$; if $Y > S$, then break
        \EndIf
    \Until{FALSE}
\Until{FALSE}
\end{algorithmic}
\label{alg:pseudocode}
\end{algorithm}

The Algorithm~\ref{alg:pseudocode} describes the pseudocode for the KG$(1,c)$ sampler using the alternating series method. For the part involving the sampling and evaluation of the c.d.f of the GIG distribution with half-integer parameter $p=-3/2$, we employ the result of \citet{Pena2025-ky}. When $c=0$, sampling from $\mathrm{GIG}(-3/2, c^2,1/4)$ in line 10 is replaced with $\mathrm{InvGamma}(3/2,1/8)$, and c.d.f. evaluation in line 2 is replaced with c.d.f. of $\mathrm{InvGamma}(3/2,1/8)$. For the choice of optimal $t^* = 0.050239$, we have $\texttt{pGIG}(t^* \mid p=-3/2, a = 0, b=1/4)= 0.1735472$.

 To sample from the truncated GIG distribution supported on $(0,t)$, we use a simple rejection method by drawing a GIG variate until it falls in $(0,t)$. Letting $t$ be fixed at the optimal value $t^* = 0.050239$, the expected number of draws depending on the choice of $c$ is  $1/\texttt{pGIG}(t^* \mid p=-3/2, a = c^2, b=1/4)$. By comparing the integrand, it can be verified that $\texttt{pGIG}(t^* \mid p=-3/2, a, b=1/4)$ is an increasing function of $a$ and achieves a minimum value of $0.1735472$ at $a=0$. 
 Thus, the expected number of GIG draws to obtain a single truncated GIG sample is not larger than $1/0.1735472 \approx 5.7622$. In practice, we found that this simple rejection method is much more efficient and numerically stable than transforming a uniform random variable through the inverse of the c.d.f. of a truncated GIG distribution.

\section{Proofs}
\label{appendix:proofs}

\subsection{Proof of Theorem~\ref{thm:kgintidentity}}
\begin{proof} It suffices to show the result when $a=0$, since $e^{a \eta}$ is always positive and cancels out. Using the fact that the Laplace transformation of $\epsilon_k\iidsim \mathrm{Gamma}(b,1)$ is $E\{\exp(-\epsilon_k t)\} = (1+t)^{-b}$, the Laplace transformation of $\kappa\sim \mathrm{KG}(b,0)$ is 
    \begin{equation}
    \label{eq:KGb0laplace}
        E\{\exp(-\kappa t)\} = \prod_{k=1}^\infty E\left\{\exp\left(-\frac{\epsilon_k t}{2\pi^2k^2}\right)\right\} = \prod_{k=1}^\infty \left(1+\frac{t}{2\pi^2k^2}\right)^{-b} = \left[\frac{(t/2)^{1/2}}{\sinh\{(t/2)^{1/2}\}}\right]^b, 
    \end{equation}
    where the last equation follows from the Weierstrass factorization theorem \citep[][\S 4.36.1]{Olver2010-jj}. Plugging in $t = \eta^2/2$, we have $E(e^{-\eta^2 \kappa /2}) = \{(\eta/2)/\sinh(\eta/2)\}^{b} = e^{b\eta/2}/\{(e^{\eta}-1)/\eta\}^b$, which completes the proof.
\end{proof}

\subsection{Proof of Theorem~\ref{thm:kgexptilting}}
\begin{proof} The case when $c=0$ is trivial, thus assume $c\neq 0$. From the proof of Theorem 1, we know that $E(e^{-c^2\kappa/2}) = \left\{\sinh(c/2)/(c/2)\right\}^{-b}$ for $\kappa\sim \mathrm{KG}(b,0)$, which corresponds to the normalizing constant of an exponential tilted density. Thus, the exponential tilted distribution has density 
\begin{equation}
\label{eq:exptilting}
    \left\{\sinh(c/2)/(c/2)\right\}^{b} \exp(-c^2\kappa/2)\pkg(\kappa; b, 0).
\end{equation}
To show that the Laplace transformation of \eqref{eq:exptilting} coincides with the Laplace transformation of a $\mathrm{KG}(b,c)$ random variable defined as the infinite convolution of gamma distributions, let 
\begin{align*}
\int_0^\infty e^{-t\kappa} \frac{\{\sinh(c/2)\}^b}{(c/2)^b} e^{-c^2\kappa/2} \pkg(\kappa; b, 0) \rmd\kappa &= \frac{\{\sinh(c/2)\}^b}{(c/2)^b}\int e^{-(-0.5 c^2+t)\kappa } \pkg(\kappa; b, 0) d\kappa  \\
&= \frac{[\sinh\{(c^2/4)^{1/2}\}]^b}{\{(c^2/4)^{1/2}\}^b}\frac{[\{(0.5c^2+t)/2\}^{1/2}]^b}{(\sinh{[\{(0.5c^2+t)/2\}^{1/2}}])^b}\\
&= \prod_{k=1}^\infty \left\{\frac{\frac{2k^2\pi^2+c^2/2}{2k^2\pi^2}}{\frac{2k^2\pi^2+(0.5c^2+t)}{2k^2\pi^2}}\right\}^b\\
&= \prod_{k=1}^\infty (1+d_k^{-1}t)^{-b}
\end{align*}
where $d_k = 2k^2\pi^2 + c^2/2$. This corresponds to the Laplace transformation of the infinite sum of independent $\mathrm{Gamma}(b,1)$ distributions scaled by $d_k^{-1}$ for $k=1,\dots$, which completes the proof. 
\end{proof}

\subsection{Proof of Theorem~\ref{thm:uniformergodic}}

First, define $\sinhc(x) \coloneqq \sinh(x)/x$ for $x\neq 0$ and $\sinhc(0) \coloneqq 1$. 
We first present two technical lemmas, where Lemma~\ref{lemma:eigen} is the same as Lemma 3.1 of \citet{Choi2013-jv}. 

\begin{lemma}
\label{lemma:eigen}
If $A$ is a symmetric nonnegative definite matrix, all eigenvalues of $(I+A)^{-1}$ are in $(0,1]$ and thus $\bfz^\T(I+A)^{-1}\bfz \le \bfz^\T \bfz$ for any vector $\bfz$. Also, $I-(I+A)^{-1}$ is nonnegative definite.
\end{lemma}
\begin{lemma}
\label{lemma:sinhcineq}
For $a,b \ge 0$, $\sinhc(a+b) \le 2 \sinhc(a) \cosh(b)$.
\end{lemma}
\begin{proof}
It trivially holds when $a=0$ or $b=0$. When $a,b>0$, by expanding $\sinh(a+b) = \sinh(a)\cosh(b)+\cosh(a)\sinh(b)$ and multiplying both side by $a+b$, it is equivalent to showing the inequality $\coth(a)\tanh(b) \le 1+2b/a$ for $a,b>0$. In other words, it suffices to show that for any given $a>0$, $f(x) = \coth(a)\tanh(x)-1-2x/a \le 0$ for all $x>0$. Consider two linear functions 
$g_1(x) = x-1$ and $g_2(x) = \coth(a) - 1-2x/a$, with $g_1$ increasing and $g_2$ decreasing. It can be easily checked that $g_1(x) \ge f(x)$ and $g_2(x) \ge f(x)$ for any $x>0$, thus $\min(g_1(x),g_2(x))\ge f(x)$ for any $x>0$. The proof is completed by confirming $g_1(x)$ and $g_2(x)$ intersects at $x^*=a\coth(a)/(a+2)$ with $g_1(x^*) = g_2(x^*) = a\coth(a)/(a+2) - 1 \le 0$ for any $a>0$. 
\end{proof}

Now we provide the formal statement of the theorem. We denote $P^t$ as a $t$-step transition kernel, $\|\nu_1 - \nu_2\|_{\mathrm{TV}}$ as a total variation distance between probability measures $\nu_1$ and $\nu_2$. Let $\bm\Theta$ be a set of parameters of the model and $\Pi(\cdot)$ be a posterior of $\bm\Theta$. We say the Markov chain $\{\bm\Theta^{(m)}\}_{m=0}^\infty$ is uniformly ergodic if there exist a constant $M > 0$ and $\rho\in[0,1)$, both independent of initial state $\bm\Theta^{(0)}$, such that $\| P^t(\bm\Theta^{(0)},\cdot) - \Pi(\cdot) \|_{\mathrm{TV}}\le M\rho^t$ for all $t\ge 1$.
 
When parameters are partitioned into two blocks $\bm\Theta = \bm\Theta_1 \cup \bm\Theta_2$, $\bm\Theta_1 \cap \bm\Theta_2 = \emptyset$ and a Gibbs sampler iteratively updates between $p(\bm\Theta_1\mid \bm\Theta_2)$ and $p(\bm\Theta_2\mid \bm\Theta_1)$, showing the uniform ergodicity of either $\{\bm\Theta_1^{(m)}\}_{m=0}^\infty$ or $\{\bm\Theta_2^{(m)}\}_{m=0}^\infty$ is sufficient for the uniform ergodicity of $\{\bm\Theta^{(m)}\}_{m=0}^\infty$; see \citet{Roberts2001-st} and also \citet{Bhattacharya2021-dl}. 
The conditioning on data $\bmy$ is always assumed and suppressed from the notation for simplicity.  
Recall that Algorithm~\ref{alg:cobin} for cobin regression consists of two blocks $\bm\Theta_1 = (\lambda, \bm\kappa)$ and $\bm\Theta_2 = \bm\beta$. Also, Algorithm~\ref{alg:micobin} for micobin regression consists of two blocks $\bm\Theta_1 = (\bm\lambda, \bm\kappa)$ and $\bm\Theta_2 = (\bm\beta, \psi)$. To end with, we establish uniform ergodicity by showing that the marginal chain $\{\bm\Theta_2^{(m)}\}_{m=0}^\infty$ is uniformly ergodic for both Algorithms, under a mean zero normal prior for coefficient $\bm\beta\sim N_p(\bm{0}, \Sigma_\beta)$ (and beta prior $\psi\sim \mathrm{Beta}(a_\psi, b_\psi)$ for micobin) and with some large upper bound $L$ of $\lambda$. 

The proof strategy is based on the establishment of a uniform minorization condition, also known as a Doeblin condition \citep{Rosenthal1995-uv,Jones2001-iv}. Our approach is structurally similar to \citet{Choi2013-jv}, but now involves Kolmogorov-Gamma variables instead of Polya-Gamma, and additionally involves $\bm\lambda = (\lambda_1,\dots,\lambda_n)$ as well as $\psi$ for micobin. 
Letting $k(\bm\Theta_2\mid \bm\Theta_2^*)$ be a Markov transition density, it is sufficient to show that there is a constant $\delta^\star > 0$ and a probability density function $q(\bm\Theta_2)$, which does not depend on $\bm\Theta_2^*$, such that $k(\bm\Theta_2 \mid \bm\Theta_2^*) \ge \delta^\star q(\bm\Theta_2)$ for any $\bm\Theta_2, \bm\Theta_2^*$.

We first focus on Algorithm~\ref{alg:micobin} for micobin regression, where $\bm\Theta_2 = (\bm\beta, \psi)$, and then specialize to cobin afterwards. Let $\bm\Theta_2^* = (\bm\beta^*, \psi^*)$ be a parameter of the previous step. The Markov transition density of $\bm\Theta_2$ (marginalizing out $\bm\Theta_1$) is
\begin{equation}
    k(\bm\Theta_2 \mid \bm\Theta^*_2) = \sum_{\bm\lambda\in \{1,\dots,L\}^n} p(\psi \mid \bm\lambda) k(\bm\beta \mid \bm\beta^*; \bm\lambda)  p(\bm\lambda\mid \bm\beta^*, \psi^*),
\end{equation}
where the intermediate transition density of $\bm\beta$, conditional on $\bm\lambda$, is
\begin{equation}
\label{eq:mtk}
k(\bm\beta \mid \bm\beta^*; \bm\lambda) = 
    \int_{(0,\infty)^n} p(\bm\beta \mid \bm\kappa, \bm\lambda) p(\bm\kappa \mid \bm\beta^*, \bm\lambda)  \rmd \bm\kappa
\end{equation}
We aim to establish a lower bound of $k(\bm\beta \mid \bm\beta^*; \bm\lambda)$, expression \eqref{eq:mtk}, that is independent of $\bm\beta^*$. This is achieved through two steps.

\begin{proposition}
\label{prop:s1}
Denoting $\bfs =  \bfs(\bm\lambda) = \Sigma_\beta^{1/2}X^\T\tilde{\bmy}$ where $\tilde{\bmy}\in\bbR^n$ is a vector with $i$th element $(y_i\lambda_i - \lambda_i/2)$,
\[
p(\bm\beta \mid \bm\kappa, \bm\lambda) \ge \frac{1}{(2\pi)^{p/2}|\Sigma_\beta|^{1/2}} \exp\left(-\frac{1}{2}\bm\beta^\T\Sigma_\beta^{-1}\bm\beta - \frac{1}{2}\bfs^\T \bfs + \bfs^\T\Sigma_{\beta}^{-1/2}\bm\beta\right)\prod_{i=1}^n\exp\left\{-\frac{(\bmx_i^\T\bm\beta)^2}{2}\kappa_i\right\}
\]
\end{proposition}
\begin{proof} 
   Letting $K = \diag(\bm\kappa)$ and $V_\beta = (X^\T K X + \Sigma_\beta^{-1})^{-1}$, we have $|V_\beta|\le |\Sigma_\beta|$ by Lemma~\ref{lemma:eigen} since 
   $
   |V_\beta| = |\Sigma_\beta| |(\tilde{X}^\T K \tilde{X} + I_p)^{-1}| 
   $
   where $\tilde{X} = X \Sigma_\beta^{1/2}$. Next, for $\bfm = V_\beta X^\T\tilde{\bmy}$ we have $\bfm^\T V_\beta^{-1}\bfm \le \bfs^\T\bfs$, which follows from  $\bfm^\T V_\beta^{-1}\bfm = (\tilde{X}\tilde{\bmy})^\T(\tilde{X}^\T K \tilde{X} + I_p)^{-1}(\tilde{X}\tilde{\bmy})\le \bfs^\T\bfs$ from Lemma~\ref{lemma:eigen}. Using two inequalities and $\bfm^\T V_\beta^{-1} = \bfs^\T\Sigma_\beta^{-1/2}$, we have 
   \begin{align*}
   p(\bm\beta \mid \bm\kappa, \bm\lambda) &= (2\pi)^{-p/2}|V_\beta|^{-1/2}\exp(-(\bm\beta - \bfm)^\T V_\beta^{-1}(\bm\beta-\bfm)/2)\\
   &\ge (2\pi)^{-p/2}|\Sigma_\beta|^{-1/2}\exp\left(-\frac{1}{2}\bm\beta^\T\Sigma_\beta^{-1}\bm\beta - \frac{1}{2}\bfs^\T \bfs + \bfs^\T\Sigma_{\beta}^{-1/2}\bm\beta\right)\prod_{i=1}^n\exp\left\{-\frac{(\bmx_i^\T\bm\beta)^2}{2}\kappa_i\right\}
   \end{align*}
\end{proof}

\begin{proposition} We have a lower bound of $k(\bm\beta \mid \bm\beta^*; \bm\lambda)$, expression \eqref{eq:mtk}:
\[
\int_{(0,\infty)^n} p(\bm\beta \mid \bm\kappa, \bm\lambda) p(\bm\kappa \mid \bm\beta^*, \bm\lambda) \rmd \bm\kappa \ge \delta(\bm\lambda) N_p(\bm\beta; \bm\mu^\star(\bm\lambda), \Sigma^\star)
\]
where $\Sigma^\star = (\frac{1}{2} X^\T X + \Sigma_\beta^{-1})^{-1}$, $\bm\mu^\star = \bm\mu^\star(\bm\lambda) = (\frac{1}{2} X^\T X + \Sigma_\beta^{-1})^{-1}\Sigma_\beta^{-1/2}\bfs(\bm\lambda)$, and  $\delta(\bm\lambda) = 2^{-\sum_{i}\lambda_i}e^{-\sum_{i=1}^n\lambda_i^2/4}|\Sigma^\star|^{1/2}|\Sigma_\beta|^{-1/2} \exp\left(-\frac{1}{2}\bfs^\T(I_p-(\Sigma_\beta^{1/2}X^\T X \Sigma_\beta^{1/2}/2+I_p)^{-1})\bfs\right)$.
\label{prop:s2}
\end{proposition}

\begin{proof}
First, we have $p(\bm\kappa\mid \bm\beta^*, \bm\lambda) = \prod_{i=1}^n \sinhc(\bmx_i^\T\bm\beta^*/2)^{\lambda_i} \exp(-(\bmx_i^\T\bm\beta^*)^2 \kappa_i/2) \pkg(\kappa_i;\lambda_i, 0)$. Then, with Proposition~\ref{prop:s1}, the integrand has a lower bound
 \begin{align}
        p(\bm\beta \mid \bm\kappa, \bm\lambda) p(\bm\kappa \mid \bm\beta^*, \bm\lambda) &\ge (2\pi)^{-p/2}|\Sigma_\beta|^{-1/2}\exp\left(-\frac{1}{2}\bm\beta^\T\Sigma_\beta^{-1}\bm\beta - \frac{1}{2}\bfs^\T \bfs + \bfs^\T\Sigma_{\beta}^{-1/2}\bm\beta\right)\nonumber \\
        &\times \prod_{i=1}^n\sinhc\left(\frac{\bmx_i^\T\bm\beta^*}{2}\right)^{\lambda_i} \exp\left(-\frac{\left\{(\bmx_i^\T\bm\beta^*)^2+(\bmx_i^\T\bm\beta)^2\right\}\kappa_i}{2}\right) \pkg(\kappa_i;\lambda_i, 0)\nonumber
    \end{align}
Analyzing the terms involving $\kappa_i$, after integration, 
\begin{align}
    \int_0^\infty \exp\left(-\frac{\left\{(\bmx_i^\T\bm\beta^*)^2+(\bmx_i^\T\bm\beta)^2\right\}\kappa_i}{2}\right) \pkg(\kappa_i;\lambda_i, 0) \rmd \kappa_i 
    &= \left\{\sinhc\left(\frac{\sqrt{|\bmx_i^\T\bm\beta^*|^2+|\bmx_i^\T\bm\beta|^2}}{2} \right)\right\}^{-\lambda_i}\nonumber \\
    &\ge 
    \left\{\sinhc\left(\frac{|\bmx_i^\T\bm\beta^*|}{2} + \frac{|\bmx_i^\T\bm\beta|}{2}\right)\right\}^{-\lambda_i}\nonumber\\
    &\ge  \left\{ 2 \sinhc\left(\frac{|\bmx_i^\T \bm\beta^*|}{2}\right)\cosh\left(\frac{|\bmx_i^\T\bm\beta|}{2}\right)\right\}^{-\lambda_i}\nonumber
\end{align}
where the first inequality is due to $\sqrt{a+b}\le \sqrt{a} +\sqrt{b}$ for $a,b\ge 0$, combining with $\sinhc(x)^{-\lambda_i}$ is a decreasing function for $x\ge 0$, and the second inequality is by Lemma~\ref{lemma:sinhcineq}. Thus, we have a lower bound of the integral that is independent of $\bm\beta^*$,
     \begin{align}
     &\int_{(0,\infty)^n}\prod_{i=1}^n\sinhc\left(\frac{\bmx_i^\T\bm\beta^*}{2}\right)^{\lambda_i} \exp\left(-\frac{\left\{(\bmx_i^\T\bm\beta^*)^2+(\bmx_i^\T\bm\beta)^2\right\}\kappa_i}{2}\right) \pkg(\kappa_i;\lambda_i, 0) \rmd \bm\kappa \nonumber\\
        \ge& 2^{-\sum_{i}\lambda_i}\left[\prod_{i=1}^n\cosh\left(\frac{|\bmx_i^\T\bm\beta|}{2}\right)^{-\lambda_i}\right]\nonumber\\
        \ge & 2^{-\sum_{i}\lambda_i}\prod_{i=1}^n \exp\left\{-\frac{1}{2} \frac{(\bmx_i^\T\bm\beta)^2 + \lambda_i^2}{2} \right\} = 2^{-\sum_{i}\lambda_i}e^{-\sum_{i=1}^n\lambda_i^2/4}\exp\left\{-\frac{1}{2}\left(\frac{\bm\beta^\T X^\T X \bm\beta}{2}\right)\right\}\nonumber
    \end{align}
where the last inequality follows from $\cosh(|x|)^{-l} \ge e^{-l|x|}\ge \exp(-x^2-l^2/4)$ for any $x\in\bbR$ and $l>0$. Combining together with the remaining parts, denoting $\Sigma^\star = (\frac{1}{2} X^\T X + \Sigma_\beta^{-1})^{-1}$, 
\begin{align*}
    \int_{(0,\infty)^n} p(\bm\beta \mid \bm\kappa, \bm\lambda) p(\bm\kappa \mid \bm\beta^*, \bm\lambda) &\ge (2\pi)^{-p/2}|\Sigma_\beta|^{-1/2}\exp\left(-\frac{1}{2}\bm\beta^\T\Sigma_\beta^{-1}\bm\beta - \frac{1}{2}\bfs^\T \bfs + \bfs^\T\Sigma_{\beta}^{-1/2}\bm\beta\right)\\
        &\quad \times 2^{-\sum_{i}\lambda_i}e^{-\sum_{i=1}^n\lambda_i^2/4}\exp\left\{-\frac{1}{2}\left(\frac{\bm\beta^\T X^\T X\bm\beta}{2}\right)\right\}\\
        &= 2^{-\sum_{i}\lambda_i}e^{-\sum_{i=1}^n\lambda_i^2/4}|\Sigma^\star|^{1/2}|\Sigma_\beta|^{-1/2}\\
        &\quad \times \exp\left(-\frac{1}{2}\bfs^\T(I_p-(\Sigma_\beta^{1/2} X^\T X\Sigma_\beta^{1/2}/2+I_p)^{-1})\bfs\right)\\
        &\quad \times (2\pi)^{-p/2}|\Sigma^\star|^{-1/2}\exp\left(-\frac{1}{2} (\bm\beta - \bm\mu^\star)^\T(\Sigma^\star)^{-1}(\bm\beta - \bm\mu^\star)\right)\\
       &=\delta(\bm\lambda)N_p(\bm\beta; \bm\mu^\star(\bm\lambda), \Sigma^\star)
\end{align*}
with $\delta(\bm\lambda) = 2^{-\sum_{i}\lambda_i}e^{-\sum_{i=1}^n\lambda_i^2/4}|\Sigma^\star|^{1/2}|\Sigma_\beta|^{-1/2} \exp\left(-\frac{1}{2}\bfs^\T(I_p-(\Sigma_\beta^{1/2} X^\T X\Sigma_\beta^{1/2}/2+I_p)^{-1})\bfs\right)$ and $\bm\mu^\star = \bm\mu^\star(\bm\lambda) = (\frac{1}{2} X^\T X + \Sigma_\beta^{-1})^{-1}\Sigma_\beta^{-1/2}\bfs$.
\end{proof}

Finally, to establish uniform ergodicity, we define $q(\bm\Theta_2) \coloneqq q(\bm\beta)q(\psi)$ in a product form, where each component is
\[
q(\bm\beta) =\frac{1}{Z_\beta}\min_{\bm\lambda \in \{1,\dots,L\}^n}N_p(\bm\beta; \bm\mu^\star(\bm\lambda), \Sigma^\star), \quad \bm\beta\in \bbR^p
\]
and 
\[
q(\psi) = \frac{1}{Z_\psi}\min_{\bm\lambda \in \{1,\dots,L\}^n}\mathrm{Beta}\Big(\psi; a_\psi + 2n, b_\psi - n + \sum_{i=1}^n\lambda_i\Big), \quad \psi \in(0,1)
\]
where $Z_\beta = \int_{\bbR^p} \min_{\bm\lambda \in \{1,\dots,L\}^n}N_p(\bm\beta; \bm\mu^\star(\bm\lambda), \Sigma^\star) \rmd\bm\beta < \infty$ and $Z_\psi = \int_0^1 \min_{\bm\lambda \in \{1,\dots,L\}^n}\mathrm{Beta}(\psi; a_\psi + 2n, b_\psi - n + \sum_{i=1}^n\lambda_i) \rmd \psi < \infty$ are normalizing constants. 
Then, for any $\bm\lambda$, since $N_p(\bm\beta; \bm\mu^\star(\bm\lambda), \Sigma^\star) \ge Z_\beta q(\bm\beta)$ for all $\bm\beta\in\bbR^p$ and $p(\psi \mid \bm\lambda) \ge Z_\psi q(\psi)$ for all $\psi\in(0,1)$ by definition, 
we have 
\[
k(\bm\Theta_2 \mid \bm\Theta_2^*) \ge E_{\bm\lambda \sim p(\bm\lambda\mid \bm\beta^*)}\{\delta(\bm\lambda)\} \times  Z_\beta q(\bm\beta)Z_\psi q(\psi)\ge \min_{\bm\lambda}\{\delta(\bm\lambda)\}Z_\beta q(\bm\beta)Z_\psi q(\psi) = \delta^\star q(\bm\Theta_2)
\]
which completes the proof of uniform ergodicity of Algorithm~\ref{alg:micobin}, since $\delta^\star = \min_{\bm\lambda}\{\delta(\bm\lambda)\}Z_\beta Z_\psi > 0$. 

To see $\delta^\star > 0$ is a constant, defining $\bfs^\star = \Sigma_\beta^{1/2} X^\T\tilde{\bmy}^\star$ where $\tilde{\bmy}^\star\in\bbR^n$ is a vector with $i$th element $L(y_i - 1/2)$, which does not depend on $\bm\lambda$, we have 
\[
\min_{\bm\lambda}\{\delta(\bm\lambda)\} = 2^{-nL}e^{-nL^2/4}|\Sigma^\star|^{1/2}|\Sigma_\beta|^{-1/2} \exp\left(-\frac{1}{2}(\bfs^\star)^\T(I_p-(\Sigma_\beta^{1/2}X^\T X \Sigma_\beta^{1/2}/2+I_p)^{-1})\bfs^\star\right)
\]
since $(I_p-(\Sigma_\beta^{1/2}X^\T X \Sigma_\beta^{1/2}/2+I_p)^{-1})$ is nonnegative definite by Lemma~\ref{lemma:eigen}.

The uniform ergodicity of Algorithm~\ref{alg:cobin} for cobin regression is based on the simpler Markov transition density with $\bm\Theta_2 = \bm\beta$,
\begin{align*}
    k(\bm\beta \mid \bm\beta^*) &= \sum_{\lambda=1}^L \left\{\int_{(0,\infty)^n} p(\bm\beta \mid \bm\kappa, \lambda) p(\bm\kappa \mid \bm\beta^*, \lambda)  \rmd \bm\kappa \right\}  p(\lambda\mid \bm\beta^*)\\
    &\ge \sum_{\lambda=1}^L \delta(\lambda)N_p(\bm\beta; \bm\mu^\star(\lambda), \Sigma^\star)p(\lambda\mid \bm\beta^*)\\
    &\ge \min_{\lambda}\{\delta(\lambda)\} Z_\beta q(\bm\beta) = \delta^\star q(\bm\beta)
\end{align*}
where previous vector $\bm\lambda$ inputs are now corresponding to $\lambda \bm{1}_n$, and previous minimum over $\{1,\dots,L\}^n$ is now corresponding to minimum over $\{1,\dots,L\}$. This completes the uniform ergodicity of the cobin regression blocked Gibbs sampler.

\subsection{Proof of Proposition~\ref{prop:kg10density}}
We recall from \cite{Feller1948-ui} and \cite{Devroye1986-gj} \S 5.6 that the Kolmogorov distribution $\calK$ admits two different density representations:
\begin{equation}
    p_{\textsc{k}}(x)=8\sum_{n=0}^\infty (-1)^n (n+1)^2 x \exp\{-2(n+1)^2x^2\}
\end{equation}
and
\begin{equation}
    p_{\textsc{k}}(x)=\frac{ (2 \pi)^{1/2}}{x} \sum_{n=0}^\infty \left\{ \frac{(2n+1)^2\pi^2}{4x^3} - \frac{1}{x} \right\} \exp\left\{-\frac{(2n+1)^2\pi^2}{8x^2}\right\}
\end{equation}
It is well known that the squared Kolmogorov random variable can be represented as an infinite convolution of exponential random variables, i.e. $\calK^2 \stackrel{d}{=} 0.5\sum_{k=1}^\infty \epsilon_k/k^2$, $\epsilon_k\iidsim \mathrm{Exp}(1)$ \citep[][\S 4]{Andrews1974-ms}. By definition of Kolmogorov-Gamma, we have $\mathrm{KG}(1,0)\stackrel{d}{=}\mathcal{K}^2/\pi^2$. Thus, applying change of variables with both density representations, 
\begin{equation}
    \pkg(x;1,0)=\sum_{n=0}^{\infty}(-1)^n4\pi^2(n+1)^2\exp\left\{-2\pi^2(n+1)^2x\right\}
\end{equation}
 and
\begin{equation}
    \pkg(x;1,0)=\frac{1}{(2\pi)^{1/2}}\sum_{n=0}^{\infty}\left\{\frac{(2n+1)^2}{4x^{5/2}}-\frac{1}{x^{3/2}}\right\}\exp\left\{-\frac{(2n+1)^2}{8x}\right\}.
\end{equation}
It is easy to see how the first sum terms simplify into the $a_n^R(x)$ terms given in the statement of Proposition \ref{prop:kg10density}. Obtaining the second representation requires re-indexing the sum so that the even parts correspond to the term with $4x^{5/2}$ in the denominator and the odd parts correspond to the term with $x^{3/2}$ in the denominator. After that, the form for the $a^L_n(x)$ terms given in Proposition \ref{prop:kg10density} emerges.

\subsection{Proof of Lemma~\ref{lemma:trange}}
It suffices to show that both $a_n^L(x)$ and $a_n^R(x)$ are decreasing in $n$ in the indicated interval. We first consider $a_n^R(x)$, which we note to be decreasing in $n$ if $(n+1)^2\exp(-2\pi^2(n+1)^2x)< n^2\exp(-2\pi^2n^2x)$. But this is equivalent to inequality $(\log(n+1)-\log(n))/((2n+1)\pi^2)<x$. Since the numerator of the left term decreases in $n$ and the denominator increases in $n$, the inequality is satisfied for all $n$ if it holds for $n=1$. But that reduces to the inequality $\log (2)/(3\pi^2)<x$, and we conclude that $a_n^R(x)$ is decreasing in $n$ within $(\log(2)/3\pi^2,\infty)$.

We now show that $a_n^L(x)$ decreases in $n$ for $x<1/4$. If $n$ is even, $a_n^L(x)-a_{n+1}^L(x)$ can be expressed as a positive number times $(n+1)^2-4x$, which is positive if $x<1/4$. On the other hand, if $n$ is odd, $\log a_n^L(x)-\log a_{n+1}^L(x)>0$ if $(n+1)/(2x) + \log x-2\log(n+2)+2\log2>0$. We note that the derivative of this expression with respect to $x$ is $1/x-(n+1)/(2x^2)$, which is negative for $x\in(0,(n+1)/2)$, and so the difference is decreasing on the interval $(0,(n+1)/2)$ that includes $(0,1/4)$. Finally, we observe that the inequality $(n+1)/(2x) + \log x-2\log(n+2)+2\log2>0$ holds for $x=1/4$, and therefore $a_n^L(x)$ is decreasing in $n$ for all $x\in(0,1/4)$. 

Therefore, so long as $t\in(\log(2)/(3\pi^2),1/4)$, the sequence $a_n(x;c,t)$ is monotonically decreasing in $n$ regardless of $c$.

\subsection{Proof of Proposition~\ref{prop:optimal-t}}
The proof of this proposition has three parts. First, we show that the optimal cutoff $t^*$ is independent of $c$. Next, we provide an upper bound of the expected number of outer loop iterations for any $c$, i.e. the uniform lower bound on the acceptance probability of a proposal. Finally, we provide the uniform upper bound of the expected number of inner loop iterations. 

\textbf{Part 1. } We begin by recalling the definition of $A^L(c,t) = \int_0^t a_0(x;c,t) dx$ and $A^R(c,t) = \int_t^\infty a_0(x;c,t) dx$, where the left and right envelopes are, from Proposition~\ref{prop:kg10density} and \eqref{eq:anseq},
\[
a_0(x;c,t) = \begin{cases} 
\frac{\sinh(c/2)}{c/2}\exp(-c^2x/2)\frac{\exp(-\frac{1}{8x})}{\pi^{1/2}(2x)^{5/2}}, & 0 < x < t,\\
\frac{\sinh(c/2)}{c/2}\exp(-c^2x/2)4\pi^2\exp(-2\pi^2 x) & t \le x
\end{cases}
\]
We denote the c.d.f. of $\mathrm{GIG}(-3/2,c^2, 1/4)$ as $\texttt{pgig}(t \mid p=-3/2, a=c^2, b=1/4)$ and the c.d.f. of inverse gamma distribution with parameter $3/2$ and $1/8$ (with density proportional to $x^{-5/2}e^{-1/(8x)}$) as $\texttt{pgig}(t \mid p=-3/2, a=0, b=1/4)$ to unify the notation. Then, 
\begin{align}
    A^L(c,t)&=\int_0^ta_0(x;c,t)\rmd x\nonumber\\
    &= \int_0^t \frac{\sinh(c/2)}{c/2}\exp(-c^2x/2)\frac{\exp(-\frac{1}{8x})}{\pi^{1/2}(2x)^{5/2}} \rmd x \label{eq:leftarea-intermediate}\\
    &=\frac{\sinh(c/2)}{c/2}\frac{Z(c)}{\pi^{1/2}2^{5/2}}\int_0^t\frac{1}{Z(c)}x^{-5/2}\exp\left(-\frac{1}{2}\left(c^2x+\frac{1}{4x}\right)\right)\rmd x,\nonumber\\
    &= \frac{\sinh(c/2)}{c/2} (|c|+2)\exp(-|c|/2)\times\texttt{pgig}(t \mid p=-3/2, a=c^2, b=1/4)
    \label{eq:left-envelope-int}
\end{align}
where $Z(c) = \frac{2K_{-3/2}(|c|/2)}{(2|c|)^{-3/2}}$ if $c\neq 0$ or $Z(c) = \pi^{1/2}2^{7/2}$ if $c=0$ is a normalizing constant, and $K_{-3/2}$ is a modified Bessel function of the second kind. For the case when $c\neq 0$, we used $K_{-3/2}(|c|/2)= (\pi/|c|)^{1/2}\exp(-|c|/2)(1+2/|c|)$ in simplification. Also,
\begin{align}
    A^R(c,t)&=\int_t^\infty a_0(x;c,t) \rmd x = \frac{\sinh(c/2)}{c/2}\int_t^\infty \exp(-c^2x/2) 4\pi^2 \exp(-2\pi^2 x) \rmd x \label{eq:rightarea-intermediate}\\
    &= \frac{\sinh(c/2)}{c/2}
    \frac{4\pi^2}{2\pi^2 + c^2/2}\exp\left\{-(2\pi^2 + c^2/2) t\right\}.
    \label{eq:ar2}
\end{align}

We now show that the optimal cutoff $t^*$ that minimizes $A^L(c,t)+A^R(c,t)$ does not depend on $c$. For a given $c$, $A^L(c,t)$ and $A^R(c,t)$ are both differentiable in $t$. 
From expressions \eqref{eq:leftarea-intermediate} and \eqref{eq:rightarea-intermediate}, 
the optimal cutoff $t^*$ minimizes
 \begin{equation}
 \label{eq:optimal_t_objft}
     \int_0^t\frac{1}{\pi^{1/2} (2x)^{5/2}}\exp\left(-\frac{1}{2}\left(c^2x+\frac{1}{4x}\right)\right) \rmd x+\int_t^{\infty}4\pi^2\exp\left\{-\left(\frac{c^2}{2}+2\pi^2\right)x\right\} \rmd x
 \end{equation}
as a function of $t$, since the common $\sinh(c/2)/(c/2)$ term is never zero and does not depend on $t$. We claim that there is a unique $t^*\approx 0.050239$ within the interval $(\log(2)/(3\pi^2),1/4)$ that minimizes \eqref{eq:optimal_t_objft}. To see this, differentiating \eqref{eq:optimal_t_objft} with respect to $t$ and equating to zero, $t^*$ is a solution of
\begin{equation}
\frac{1}{\pi^{1/2} (2t)^{5/2}}\exp\left(-\frac{1}{8t}\right)-4\pi^2\exp\left(-2\pi^2t\right) = 0
\end{equation}
after canceling the $\exp(-c^2t/2)$ term, so the minimizer does not depend on $c$. Rearranging terms,
\begin{equation}
\label{eq:optimal_t_objft_deriv_rearranged}
    2^{9/2}\pi^{5/2}\exp\left(-2\pi^2 t+\frac{1}{8t}+\frac{5}{2}\log t \right) = 1
\end{equation}
whose solution is $t^*\approx 0.050239$ which lies within the bounds specified by Lemma \ref{lemma:trange}. Since the derivative of the LHS of \eqref{eq:optimal_t_objft_deriv_rearranged} is always negative on $t>0$, $t^*$ is unique.

\textbf{Part 2}. It can be easily checked that for fixed $t$, $A^L(c,t)$ and $A^R(c,t)$ are both continuous in $c$. Also, we have $A^L(0,t^*)+A^R(0,t^*) \approx 1.089002$. We claim that for fixed $t$, $A^L(c,t)\to 1$ and $A^R(c,t)\to 0$ as $c\to\infty$ to ensure that $A^L(c,t^*)+A^R(c,t^*)$ converges to 1 as $c\to\infty$ and thus the sampler is not ill behaved in the large $c$ regime. 
To see $\lim_{c\to\infty}A^L(c,t) = 1$, denoting $G\sim \mathrm{GIG}(-3/2,c^2,1/4)$, we have $1-\texttt{pgig}(t \mid p=-3/2, a=c^2, b=1/4)\le E(G)/t = 1/\{2t(c+2)\}$ by the Markov inequality, thus 
\[
1-1/(2tc+4t)\le \texttt{pgig}(t \mid p=-3/2, a=c^2, b=1/4) \le 1.
\]
Combining with $\lim_{c\to \infty} \frac{\sinh(c/2)}{c/2} (|c|+2)e^{-|c|/2}= 1$ for the LHS of the inequality, we have $\lim_{c\to\infty}A^L(c,t) = 1$. The fact that $\lim_{c\to\infty}A^R(c,t) = 0$ is easily deduced from \eqref{eq:ar2}.

For the optimal choice of $t^*\approx 0.050239$, numerical investigation shows that $A^L(c,t^*)+A^R(c,t^*)$ attains a maximum value of approximately 1.145583 as a function of $c$ when $c=10.134$. Therefore, $M$ is bounded above by 1.1456 for any given $c$. Hence, the average probability of accepting a proposal is uniformly bounded below by 0.8729.

\textbf{Part 3}. We follow Proposition 3 of \cite{Polson2013-gb}, where the probability of deciding to accept or reject a proposal $X$ upon checking the $m$th partial sum $S_m(X)$ is given by
\begin{equation*}
    \frac{1}{A^L(c,t^*)+A^R(c,t^*)}\int_0^{\infty}\left\{a_{m-1}(x;c,t^*)-a_m(x;c,t^*)\right\} \rmd x
\end{equation*}
The first few of these probabilities for the worst possible envelope  (that is, when $c=10.34$) are presented below.
\begin{table}[h]
\centering
    \begin{tabular}{ccccccc}
       \multicolumn{7}{c}{Probability of deciding to accept or reject upon computing $m^{\text{th}}$ series term}\\
      $m$ & $1$ & $2$ & $3$ & $4$ & $5$ & $6$ \\ \midrule 
      Prob. &  $0.127$ & $2.068\times10^{-4}$ & $2.226\times10^{-7}$ & $3.124\times10^{-11}$ & $5.894\times10^{-16}$ & $1.513\times10^{-21}$ 
    \end{tabular}
\label{tab:accept-reject-n-prob}
\end{table}

We see that the probabilities are decaying quickly, guaranteeing that with very high probability, only a small handful of $S_m(x)$ terms will need to be computed. From the above, using the worst-case $c=10.134$, the expected number of series terms that will need to be computed to decide whether to accept or reject a proposal is 1.1274624.

\section{Details of cobin and micobin distributions}
\label{appendix:cobinmicobindetail}

\subsection{Cobin distribution as an exponential dispersion model}
Following \citet[][\S2]{Jorgensen1986-ay}, we describe in detail how the proposed continuous binomial distribution arises from an exponential dispersion model generated from the uniform distribution. Let $Q$ be a probability measure that corresponds to a uniform distribution.
Then $M(s) = \int \exp(sy)\rmd Q(x) =  (e^{s}-1)/s = \exp(B(s))$ for $s\in\bbR$. 

First, consider the set of $\lambda$ values such that $M(s)^\lambda = \{(e^s-1)/s)\}^\lambda$ is a moment generating function of some distribution $Q_\lambda$. This set corresponds to the non-negative integers $\{\lambda\in \bbR \backslash \{0\}: M^\lambda$ is the m.g.f. of some distribution $Q_\lambda\} = \{1,2,\dots,\}$. To see this, from Theorem 3.1 (vi) of \cite{Jorgensen1986-ay}, it suffices to show that the analytic continuation of $M$ to the complex plane has at least one simple zero, which is confirmed by checking $M(z) = 0$ has solutions $z = \pm 2\pi i, \pm 4\pi i, \dots$, where $i$ is an imaginary unit.

Then, the corresponding probability measure $Q_\lambda$ is the $\lambda$-fold i.i.d. convolution of uniform distribution, i.e., the Irwin-Hall distribution defined on the interval $(0,\lambda)$. Based on $Q_\lambda$ ($\lambda = 1,2,\dots$), we consider the probability measure $Q_{\lambda, \theta}$ that satisfies
\[
\frac{\rmd Q_{\lambda, \theta}}{\rmd Q_\lambda} = \exp\{x\theta - \lambda  B(\theta)\},
\]
where the left-hand side denotes the Radon–Nikodym derivative. Hence, $Q_{\lambda, \theta}$ is the distribution obtained from exponentially tilting $Q_{\lambda}$ by $e^{x\theta}$. Transforming to $y=x/\lambda$, we obtain the exponential dispersion model $P_{\lambda, \theta}$ that satisfies
$
\frac{\rmd P_{\lambda, \theta}}{\rmd P_\lambda} = \exp\{\lambda y\theta - \lambda  B(\theta)\}.
$ 
Since the density corresponding to $P_\lambda$ (Irwin-Hall scaled by $1/\lambda$) is 
$
h(y, \lambda) = \frac {\lambda}{(\lambda-1)!}\sum _{k=0}^{\lambda}(-1)^{k}{\lambda \choose k}\max\{(\lambda y-k),0\}^{\lambda-1},
$ 
the density corresponding to $P_{\lambda, \theta}$ coincides with \eqref{eq:cobin}.

Table~\ref{table:cobinderivation} presents a comparative summary with normal, gamma, and inverse Gaussian exponential dispersion families. 

\begin{table}
\caption{Illustration of the derivation of cobin distribution as a continuous exponential dispersion model, along with a comparison with normal, gamma, and inverse Gaussian distributions.}
\centering
\begin{tabular}{c c c}
\toprule
\multicolumn{3}{c}{}\\ 
\multicolumn{3}{c}{\qquad\qquad exponential tilting by $\theta$\quad \qquad $\lambda$-fold convolution and scale by $\lambda^{-1}$}\\ 
\multicolumn{3}{c}{  \quad $\xrightarrow{\qquad\quad}$\qquad \qquad\qquad\qquad\quad$\xrightarrow{\qquad\quad}$}\\ \midrule 
& &  \\ 
$H\sim \mathrm{Unif}(0, 1)$&
$Y\sim \mathrm{cobin}(\theta, 1)$ &
$Y\sim \mathrm{cobin}(\theta, \lambda^{-1})$ \\
$h(y) = 1$& $\exp(\theta y - B(\theta))h(y)$ & $\exp(\lambda\theta y - \lambda B(\theta))h(y,\lambda)$ \\
$0\le y\le 1$ & $\theta\in \bbR$ & $\lambda \in \bbN$  \\  
&  &   \\  \midrule
& &  \\ 
$H_1\sim N(0, 1)$&
$Y\sim N(\theta, 1)$ &
$Y\sim N(\theta, \lambda^{-1})$ \\
$h_1(y) = (2\pi)^{-1/2}\exp(-y^2/2)$& $\exp(\theta y - B_1(\theta))h_1(y)$ & $\exp(\lambda\theta y - \lambda B_1(\theta))h_1(y,\lambda)$ \\
$y\in\bbR$ & $\theta\in \bbR$ & $\lambda > 0$  \\  
&  &   \\  \midrule
& &  \\ 
$H_2\sim \mathrm{Exp}(1)$&
$Y\sim \mathrm{Exp}(1-\theta)$ &
$Y\sim \mathrm{Gamma}(\lambda, \lambda(1-\theta))$ \\
$h_2(y) = \exp(-y)$& $\exp(\theta y - B_2(\theta))h_2(y)$ & $\exp(\lambda\theta y - \lambda B_2(\theta))h_2(y,\lambda)$ \\
$y\ge 0$ & $\theta<1$ & $\lambda > 0$  \\  
&  &   \\  \midrule
& &  \\ 
$H_3\sim \mathrm{InvGamma}(1/2, 1/2)$&
$Y\sim \mathrm{InvGau}((-2\theta)^{-1/2}, 1)$ &
$Y\sim \mathrm{InvGau}((-2\theta)^{-1/2}, \lambda)$ \\
$h_3(y) = (2\pi y^3)^{-1/2}\exp(-\frac{1}{2y})$& $\exp(\theta y - B_3(\theta))h_3(y)$ & $\exp(\lambda\theta y - \lambda B_3(\theta))h_3(y,\lambda)$ \\
$y>0$ & $\theta<0$ & $\lambda > 0$  \\  
&  &   \\  \midrule\midrule
\multicolumn{3}{c}{
$
B(\theta) = \log E(e^{\theta H}) = \log\left(\frac{e^\theta-1}{\theta}\right), \quad h(y,\lambda) = \text{density of } \frac{1}{\lambda}\sum_{l=1}^\lambda H^{(l)} 
$}\\
\multicolumn{3}{c}{
$
B_1(\theta) = \log E(e^{\theta H_1}) = \frac{\theta^2}{2}, \quad h_1(y,\lambda) = \text{density of } \frac{1}{\lambda}\sum_{l=1}^\lambda H_1^{(l)} = \frac{\exp(-\lambda y^2/2)}{\sqrt{2\pi/\lambda}}
$
}\\
\multicolumn{3}{c}{
$
B_2(\theta) = \log E(e^{\theta H_2}) = -\log(1-\theta), \quad h_2(y,\lambda) = \text{density of } \frac{1}{\lambda}\sum_{l=1}^\lambda H_2^{(l)} = \frac{\lambda^{\lambda}y^{\lambda-1}\exp(-\lambda y)}{\Gamma(\lambda)} 
$}\\
\multicolumn{3}{c}{$
B_3(\theta) = \log E(e^{\theta H_3}) = -(-2\theta)^{1/2}, \quad h_3(y,\lambda) = \text{density of } \frac{1}{\lambda}\sum_{l=1}^\lambda H_3^{(l)} = \frac{\lambda^{1/2}\exp(-\lambda/(2y))}{\sqrt{2\pi y^3}} 
$} \\
\bottomrule
\end{tabular}
\label{table:cobinderivation}
\end{table}

\subsection{Further details of cobin and micobin distributions}

The density of the micobin distribution is 
\begin{equation}
    p_{\mathrm{micobin}}(y;\theta, \psi) = \sum_{\lambda=1}^\infty \lambda(1-\psi)^{\lambda-1}\psi^{2}h(y,\lambda)\frac{e^{\lambda \theta y}}{\{(e^{\theta}-1)/\theta\}^\lambda\}}, \quad 0\le y \le 1
\end{equation}
which follows from its definition as an hierarchical model $Y\mid \lambda \sim \mathrm{cobin}(\theta,\lambda^{-1})$, $(\lambda-1)\sim \mathrm{negbin}(2,\psi)$. $p_{\mathrm{micobin}}$ is continuous in $y\in[0,1]$, since mixture component densities $p_{\mathrm{cobin}}(y; \theta, \lambda^{-1})$ are all continuous.  
When $\psi \to 1$, it reduces to $\mathrm{cobin}(\theta,1)$. It has nonzero density at boundary values, 
\[
p_{\mathrm{micobin}}(0;\theta, \psi) = \psi^2p_{\mathrm{cobin}}(0;\theta, 1) = \psi^2 \frac{\theta}{e^\theta-1}
\]
\[
p_{\mathrm{micobin}}(1;\theta, \psi) = \psi^2p_{\mathrm{cobin}}(1;\theta, 1) = \psi^2 \frac{\theta e^{\theta}}{e^\theta-1}
\]
Unlike cobin, which belongs to the exponential dispersion family and its likelihood is guaranteed to be log-concave in terms of $\theta$, there is no guarantee that the likelihood of micobin distribution is log-concave. 

The c.d.f. of $\mathrm{cobin}(\theta,\lambda^{-1})$ can be obtained by integrating $p_{\mathrm{cobin}}$ term-by-term. Its form when $\lambda = 1$ is available from \citet{Loaiza-Ganem2019-dl}, which is $F_{\mathrm{cobin}}(z;\theta,1) = (e^{\theta z}-1)/(e^\theta-1)$. Assuming $\lambda \ge 2$,
\begin{align*}
F_{\mathrm{cobin}}(z; \theta, \lambda^{-1}) &= \frac {\lambda}{(\lambda-1)!}\sum_{k=0}^{\lambda}(-1)^{k}{\lambda \choose k} \exp(-\lambda B(\theta))\int_0^z \max\{(\lambda y-k),0\}^{\lambda-1} \exp(\lambda y \theta) \rmd y\\
&=\frac {\lambda}{(\lambda-1)!}\sum_{k=0}^{\lambda}(-1)^{k}{\lambda \choose k} \exp(-\lambda B(\theta))\int_{k/\lambda}^z (\lambda y-k)^{\lambda-1} \exp(\lambda y \theta) \rmd y
\end{align*}
where we used $\max\{(\lambda y-k),0\} = 0$ if $y\le k/\lambda$. The integral term becomes
\[
\int_{k/\lambda}^z (\lambda y-k)^{\lambda-1} \exp(\lambda y \theta) \rmd y = \begin{cases}
\lambda^{-2} (\lambda z - k)^\lambda  & \theta = 0\\
\frac{e^{\theta k}}{\lambda (-\theta)^\lambda}\gamma(\lambda, -\theta(\lambda z-k)) & \theta \neq 0
\end{cases}
\]
where $\gamma(\lambda,x) = \int_0^x t^{\lambda-1}e^{-t} \rmd x$ is a lower incomplete gamma function. In practice, many existing software programs only support the calculation of the lower incomplete gamma function with positive $x$, i.e., when $\theta$ is negative. To resolve this, one can use symmetry of cobin distributions between $\mathrm{cobin}(\theta,\lambda)$ and $\mathrm{cobin}(-\theta,\lambda)$, which yields $F_{\mathrm{cobin}}(z; \theta, \lambda^{-1})  = 1- F_{\mathrm{cobin}}(1-z; -\theta, \lambda^{-1})$. 

The c.d.f. of $\mathrm{micobin}(\theta,\psi)$ is simply a weighted sum of cobin c.d.f.s,
\[
F_{\mathrm{micobin}}(z; \theta, \psi) = \sum_{\lambda=1}^\infty  \lambda(1-\psi)^{\lambda-1}\psi^{2}F_{\mathrm{cobin}}(z; \theta, \lambda^{-1}).
\]

The random variate generation of $\mathrm{cobin}(\theta, \lambda^{-1})$ can be easily done by taking an average of $\lambda$ i.i.d. $\mathrm{cobin}(\theta, 1)$ variables. The sampling from $\mathrm{cobin}(\theta, 1)$ can be done by $F_{\mathrm{cobin}}^{-1}(U;\theta,1)$ with $U\sim \mathrm{Unif}(0,1)$, where $F_{\mathrm{cobin}}^{-1}(u;\theta,1) = \theta^{-1}\log(ue^\theta-u+1)$ \citep{Loaiza-Ganem2019-dl}. Note that if $\theta = 0$ then $\mathrm{cobin}(0, 1)$ is a uniform distribution. Random variate generation of micobin directly follows from its definition.

\subsection{Cobit link function and variance function}

The cobit link function $g_{\mathrm{cobit}}:(0,1)\to \bbR$ has an inverse $g^{-1}_{\mathrm{cobit}}(x) = B'(x) = e^\theta/(e^\theta-1) - \theta^{-1}$. Although the expression for $g_{\mathrm{cobit}}$ is not available analytically, the numerical inversion of $B'(x)$ can be easily done with the Newton--Raphson algorithm. Under cobin regression with canonical link, the sufficient statistic of $\bm\beta$ is $X^\T\bmy$, and the log-likelihood is guaranteed to be concave \citep{Agresti2015-xu}. 
The cobit link function satisfies $\lim_{x\to -\infty}g^{-1}_{\mathrm{cobit}}(\pi x)/g_{\mathrm{cauchit}}^{-1}(x) = \lim_{x\to +\infty}\{1-g^{-1}_{\mathrm{cobit}}(\pi x)\}/\{1-g_{\mathrm{cauchit}}^{-1}(x)\} = 1$, where $g_{\mathrm{cauchit}}^{-1} = \arctan(x)/\pi+0.5$ \citep{Koenker2009-up} is an inverse of the cauchit link function. Thus, up to scale difference in the linear predictor, the cobit link maps a large linear predictor to the mean around 0 or 1, asymptotically at the same rate as the cauchit link. Compared to the logit link, this significantly reduces the influence of large outlying predictors, also contributing to robustness; see \citet[][\S 6.6]{Gelman2007-hl} for a discussion. See Figure~\ref{fig:linkft} for a comparison between cauchit and logit link functions.

\begin{figure}
    \centering
    \includegraphics[width=\linewidth]{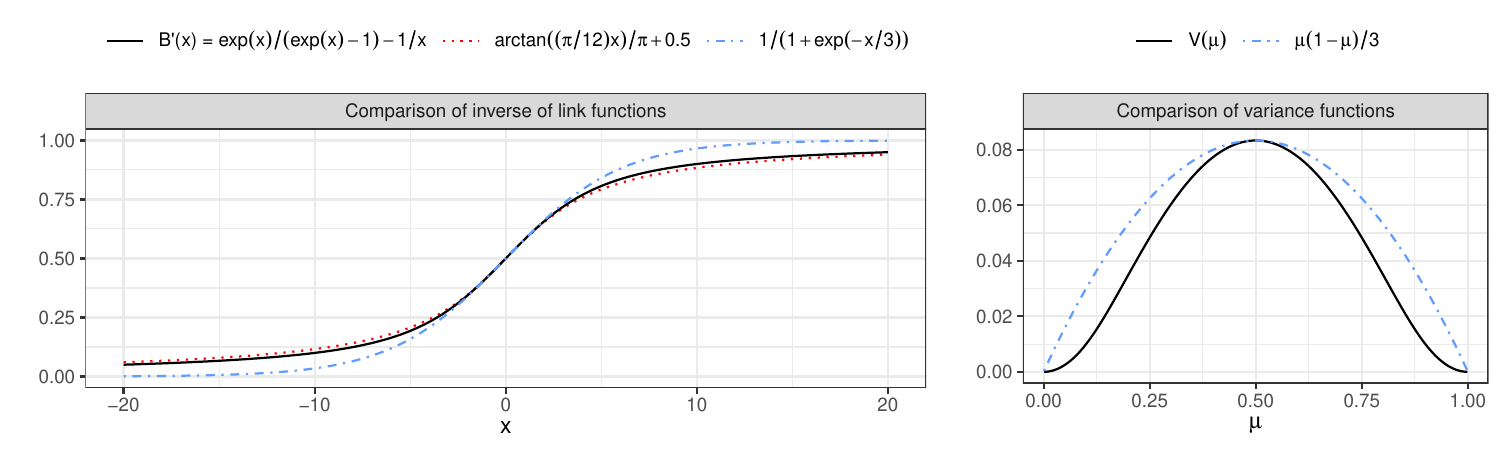}
    \caption{Comparison of link and variance functions. (Left) Inverse of cobit link $g^{-1}_{\mathrm{cobit}}(x) = B'(x)$, Cauchit link and logit link; cauchit and logit are scaled such that derivatives at zero are the same. (Right) Variance function $V(\mu) = B''\{(B')^{-1}(\mu)\}$ with variance function of beta regression model scaled by $1/3$.}
    \label{fig:linkft}
\end{figure}

The variance function $V(\mu) = B''\{(B')^{-1}(\mu)\}$ has a maximum of $1/12$ at $\mu = 1/2$. This is in constrast to the ``variance function'' of beta regression $\mu(1-\mu)$, which has a maximum of $1/4$ at $\mu = 1/2$ and satisfies $\var(Y) = \mu(1-\mu)/(1+\phi)$ for $Y\sim \mathrm{Beta}(\mu,\phi)$ (mean, precision parametrization). The range of the variance function of cobin is less than that of the beta, as the beta density spikes at the boundaries for small $\phi$, leading to a higher variance. Another notable difference is the behavior of the variance function when $\mu$ is close to zero or one, where $\mu(1-\mu)$ approaches to 0 rapidly whereas $V(\mu) = B''\{(B')^{-1}(\mu)\}$ approaches to 0 smoothly. See the right panel of Figure~\ref{fig:linkft} for a comparison.

\section{Detailed derivation of Gibbs samplers and EM algorithm}
\label{appendix:algdetails}

\subsection{Fixed effects models}
We first add details on the derivations of Algorithm~\ref{alg:cobin} and Algorithm~\ref{alg:micobin}, where we suppressed the notations conditioning on data $\bmy$ for conciseness. The update for $\lambda$ (or $\lambda_i$ for micobin) is based on the posterior under the non-augmented model,
\begin{align*}
p(\bm\beta, \lambda \mid \bmy) &\propto p(\bm\beta) p_\lambda(\lambda) \prod_{i=1}^n p_{\mathrm{cobin}}(y_i; \bmx_i^\T\bm\beta, \lambda^{-1}) & \text{(cobin)}\\
p(\bm\beta, \bm\lambda \mid \psi, \bmy) &\propto p(\bm\beta) \left\{\prod_{i=1}^n \lambda_i(1-\psi)^{\lambda_i-1}\psi^2\right\} \left\{\prod_{i=1}^n p_{\mathrm{cobin}}(y_i; \bmx_i^\T\bm\beta, \lambda_i^{-1})\right\}& \text{(micobin)}
\end{align*}
which yields step 1 of both algorithms by conditioning on $\bm\beta$. Given $\lambda$ (or $\bm\lambda$ for micobin), the conditional posterior under the augmented model is, combining with \eqref{eq:auglik},
\begin{align*}
p(\bm\beta, \bm\kappa \mid \lambda, \bmy) &\propto p(\bm\beta) \prod_{i=1}^n \exp\{{\lambda(y_i - 0.5)\bm{x}_i^\T\bm\beta}- \kappa_i (\bm{x}_i^\T\bm\beta)^2/2\}\pkg (\kappa_i; \lambda, 0), & \text{(cobin)}  \\
p(\bm\beta, \bm\kappa \mid \bm\lambda, \bmy) &\propto p(\bm\beta) \prod_{i=1}^n \exp\{{\lambda_i(y_i - 0.5)\bm{x}_i^\T\bm\beta}- \kappa_i (\bm{x}_i^\T\bm\beta)^2/2\}\pkg (\kappa_i; \lambda_i, 0).   &\text{(micobin)}
\end{align*}
From Theorem~\ref{thm:kgexptilting}, both algorithms' step 2 of updating KG random variables follows by conditioning on $\bm\beta$. Conditioning on $\bm\kappa$, with normal prior $\bm\beta\sim N(\bm{0}, \Sigma_\beta)$, step 3 follows by observing that $\exp\{{\lambda(y_i - 0.5)\bm{x}_i^\T\bm\beta}- \kappa_i (\bm{x}_i^\T\bm\beta)^2/2\}$  is proportional to $N(\lambda (y_i-0.5)/\kappa_i; \bmx_i^\T\bm\beta, \kappa_i^{-1})$ in $\bm\beta$, where $N(y;\mu, \sigma^2)$ denotes the $N(\mu, \sigma^2)$ density evaluated at $y$.  Finally, step 4 of Algorithm~\ref{alg:micobin} comes from the fact that the contribution from the latent negative binomial model to the likelihood of $\psi$ is $\prod_{i=1}^n \{\lambda_i (1-\psi)^{\lambda_i-1}\psi^2\}$.


\subsection{Mixed effect models}

We first derive a generic partially collapsed Gibbs sampler \citep{Van-Dyk2008-dr} for cobin and micobin mixed effect models, where fixed effect and random effect coefficients are jointly updated. This is made possible due to KG augmentation, leveraging conditional normal likelihood and normal-normal conjugacy. Then we specialize in random intercept and spatial regression models. Consider the following general mixed model setting with random effect $\bmu\in\bbR^q$,  
\begin{equation}
g_{\mathrm{cobit}}(E\{y_i\mid u_i, \bmx, \bmz\}) = \bmx_i^\T\bm\beta + \bmz_i^\T \bm{u}, \quad \bmu\sim N_q(\bm{0}, \Sigma_u(\vartheta)), \quad i=1,\dots,n,
\end{equation}
along the cobin or micobin response distribution with parameter $\theta_i = \eta_i = \bmx_i^\T\bm\beta + \bmz_i^\T \bm{u}$. Here $\vartheta$ is random effect covariance parameter(s), $\bmz_i\in\bbR^q$ corresponds to the random effect covariate, and $Z = (\bmz_1^\T,\dots,\bmz_n^\T)^\T\in\bbR^{n\times q}$ is a random effect design matrix. 

Given the linear predictor $\bmx_i^\T\bm\beta + \bmz_i^\T \bm{u}$, the joint update of parameters related to dispersion ($\lambda$ or $(\bm\lambda, \psi)$ and KG random variables is the same as fixed-effect models. 
Now, let the KG variables and dispersion parameters be fixed, then the likelihood (in terms of $\bm\beta$ and $\bmu$) is proportional to $\prod_{i=1}^n N(\tilde{y}_i; \eta_i, \kappa_i^{-1})$, where $\tilde{y}_i$ is defined as $\lambda (y_i - 0.5)/\kappa_i$ for cobin and $\lambda_i(y_i-0.5)/\kappa_i$ for micobin. In matrix form, denoting $K = \diag(\kappa_1,\dots,\kappa_n)$, the likelihood is proportional to $N_n(\tilde{\bmy}; X\bm\beta + Z\bmu, K^{-1})$.

\begin{algorithm}[t]
\small
\captionsetup{font=small} 
\caption{One cycle of a partially collapsed Gibbs sampler}
\label{alg:cobinmixed}
\begin{algorithmic}[1]
\State Sample $\lambda$ from $\pr(\lambda = l \mid \bm\beta, \bmu) \propto p_{\lambda}(l)\prod_{i=1}^n p_{\mathrm{cobin}}(y_i; \eta_i, l^{-1})$ among $\{1,\dots,L\}$ 
\State Sample $\kappa_i$ from $(\kappa_i \mid \lambda, \bm\beta, \bmu) \indsim \mathrm{KG}(\lambda, \eta_i)$, $i=1,\dots,n$ \Comment{steps 1,2 jointly updates $(\lambda,\bm\kappa)$}
\State Sample $\bm\beta$ from $(\bm\beta \mid \lambda, \bm\kappa, \vartheta)\sim N_p(\bm{m}_\beta, V_\beta)$, where \Comment{$\bmu$ integrated out but conditioned on $\vartheta$} \vspace{-1mm}
\[
V_\beta^{-1} = X^\T(K^{-1} + Z\Sigma_u(\vartheta)Z^\T)^{-1}X + \Sigma_\beta^{-1}, \quad \bm{m}_\beta =  V_\beta X^\T(K^{-1} + Z\Sigma_u(\vartheta)Z^\T)^{-1}\tilde{\bmy} 
\]
\State Sample $\vartheta$ from $p(\vartheta \mid \bm\beta, \lambda, \bm\kappa)$ using Metropolis-Hastings with some proposal  $\vartheta^\star\sim q(\vartheta^\star\mid\vartheta)$, \vspace{-1mm}
\[
\text{accept }\vartheta^* \text{ with probability }\min \left\{1, \frac{\calL(\vartheta^\star)p(\vartheta^\star)}{\calL(\vartheta)p(\vartheta)}\frac{q(\vartheta \mid \vartheta^\star)}{q(\vartheta^\star \mid \vartheta)}\right\}, \text{ where }\calL(\vartheta) \text{ is in \eqref{eq:fullcondvartheta}}
\]
\State Sample $\bmu$ from $(\bmu \mid \vartheta, \bm\beta, \lambda, \bm\kappa)\sim N_q(\bm{m}_u, V_u)$, where \Comment{Steps 4 and 5 jointly updates $\vartheta$ and $\bmu$} \vspace{-1mm}
\[
V_u^{-1} = Z^\T KZ + \Sigma(\vartheta)^{-1}, \quad \bm{m}_u =  V_u Z^\T K(\tilde{\bmy} - X\bm\beta)
\]
\end{algorithmic}
\end{algorithm}

The sampler first updates $\bm\beta$ from the partially collapsed posterior $p(\bm\beta \mid \lambda, \bm\kappa, \vartheta)$ where $\bmu$ is marginalized out but the random effect covariance parameter $\vartheta$ is not. By normal-normal conjugacy, integrating out $\bmu\sim N_n(\bm{0}, \Sigma_u(\vartheta) )$, we have 
\[
p(\bm\beta \mid \lambda, \bm\kappa, \vartheta) \propto N_n(\tilde{\bmy}; X\bm\beta, K^{-1} + Z\Sigma_u(\vartheta)Z^\T)N_p(\bm\beta; \bm{0}, \Sigma_\beta)
\]
which yields step 3 of Algorithm~\ref{alg:cobinmixed}. 
Next, we update $(\vartheta, \bmu)$ jointly by first sampling $\vartheta\sim p(\vartheta \mid \bm\beta, \lambda, \bm\kappa)$ and then sampling $\bmu\sim p(\bmu \mid \vartheta, \bm\beta, \lambda, \bm\kappa)$. This is based on
\[
p(\bmu, \vartheta \mid \bm\beta, \lambda, \bm\kappa) \propto N_q(\bmu; (Z^\T K^{-1}Z)^{-1}Z^\T K(\tilde{\bmy} - X\bm\beta), (Z^\T K^{-1}Z)^{-1})N_q(\bmu; \bm{0}, \Sigma(\vartheta))p(\vartheta)
\]
where we used $N_n(Z\bmu; \bm\mu, K) \propto N_q\left(\bmu; (Z^\T K^{-1}Z)^{-1}Z^\T K^{-1}\bm\mu, (Z^\T K^{-1}Z)^{-1}\right)$ in terms of $\bmu$ for a full column rank $Z\in\bbR^{n\times q}$. Conditioning on $\vartheta$, it yields step 5 of Algorithm~\ref{alg:cobinmixed}.  Marginalizing out $\bmu$ using normal-normal conjugacy, we have 
\begin{equation}
\label{eq:fullcondvartheta}
p(\vartheta \mid \bm\beta, \lambda, \bm\kappa) \propto \underbrace{N_q\left( (Z^\T K^{-1}Z)^{-1}Z^\T K(\tilde{\bmy} - X\bm\beta) ; \bm{0}, (Z^\T K^{-1}Z)^{-1} + \Sigma(\vartheta) \right)}_{\calL(\vartheta)}p(\vartheta),
\end{equation}
which gives step 4 of Algorithm~\ref{alg:cobinmixed}. While it is possible to condition on $\bmu$ when sampling $\vartheta$, which may provide a direct sampler (such as when $\vartheta$ is a marginal variance and with an inverse gamma prior), marginalizing out $\bmu$ significantly improves mixing.

Now we specialize this generic sampler into three cases. 

\textbf{Random intercept model}. In this case, the dimension of the random effect coefficient $q$ ($q<n$) corresponds to the number of groups. Let $g_i\in\{1,\dots,q\}$ be a group label for $i=1,\dots,n$. Then the random effect design matrix $Z$ has elements $Z_{i,j} = 1$ if $g_i = j$ and $Z_{i,j} = 0$ otherwise. Importantly, $Z^\T K^{-1}Z = Z^\T \diag(1/\kappa_1,\dots,1/\kappa_n)Z$ becomes a diagonal matrix.

\textbf{Spatial mixed effects model}. Spatial regression model \eqref{eq:spreg} corresponds to the case when $q=n$ and $Z = I_n$. Here $Z^\T K^{-1}Z = K^{-1}$ is a diagonal matrix.

\textbf{Spatial mixed effects model with sparse $\Sigma(\vartheta)^{-1}$}. When the precision matrix $\Sigma(\vartheta)^{-1}$ is known to be sparse, such as under the NNGP model, matrix operations that involve inversion of $n\times n$ matrices can utilize sparse matrix algorithms. First, in step 2, the calculation of $V_\beta$ involves $(K^{-1} + \Sigma_u(\vartheta))^{-1} = K - K( \Sigma(\vartheta)^{-1} + K)^{-1}K$, where the equation follows from the Woodbury identity. Here $\Sigma(\vartheta)^{-1} + K$ is sparse since $K$ is diagonal, thus its inversion can utilize a sparse matrix algorithm. Next, in step 4, it involves the inverse and determinant calculation of the covariance matrix of the $\calL(\vartheta)$ term, i.e. $(K+\Sigma(\vartheta))^{-1} = K^{-1} - K^{-1}(\Sigma(\vartheta)^{-1} + K^{-1})^{-1}K^{-1}$ and $|K+\Sigma(\vartheta)| = |\Sigma(\vartheta)^{-1}+K^{-1}||K||\Sigma(\vartheta)|$, which are all based on sparse matrices. Finally, in step 5, $K + \Sigma(\vartheta)^{-1}$ is already sparse. Leveraging such sparsity significantly improves  computation, where we used \texttt{spam} \texttt{R} package \citep{Furrer2010-ky}. Note that this corresponds to a block update of $\bmu$ that is different from a sequential update of $\bmu$, a default sampler of \texttt{spNNGP} \texttt{R} package \citep{Finley2022-pg}.  

For a low-rank structure on a covariance matrix $\Sigma(\vartheta)$, a similar strategy involving Woodbury matrix identity can be employed; see \citet{Lee2024-cq} for details. 

The derivations for micobin mixed effects models are essentially the same except for step 1 and an additional step of sampling $\psi$. Derivations for those additional steps are detailed in the previous subsection and thus omitted.

\subsection{Varying dispersion micobin model}
\label{subsec:vardispmicobin}
Next, we describe the posterior inference procedure for variable dispersion micobin regression \eqref{eq:vardispmicobin}. Given $\bm\lambda$, the dispersion submodel is a negative binomial regression model $(\lambda_i-1)\sim \mathrm{negbin}(2, \psi_i)$, $\mathrm{logit}(\psi_i) = \bmd_i^\T\bm\gamma$ . Therefore, under the normal prior on $\bm\gamma\sim N(\bm\mu_\gamma, \Sigma_\gamma)$, we can combine P\'olya-Gamma (PG) data augmentation \citep{Pillow2012-qq,Polson2013-gb} that leads to conditionally conjugate sampling of $\bm\gamma$. The resulting Gibbs sampler has 5 steps, where steps 1 to 3 are equal to Algorithm~\ref{alg:micobin}. Denote $\bmd_i\in\bbR^{d}$ as a covariate for the dispersion model, and $D\in\bbR^{n\times d}$ as a design matrix. The new step 4 corresponds to  sampling $\omega_i \indsim \mathrm{PG}(1+\lambda_i, \bmd_i^\T\bm\gamma)$ for $i=1,\dots,n$. Step 5 corresponds to sampling $(\bm\gamma \mid \bm\lambda, \bm\omega)\sim N_q(\bm{m}_\gamma, V_\gamma)$, where $V_\gamma^{-1} = D^\T \diag(\omega_1,\dots,\omega_n) D + \Sigma_\gamma^{-1}$ and $\bm{m}_\gamma = V_\gamma \{D^\T(1.5-0.5\lambda_1,\dots,1.5 -0.5\lambda_n)^\T + \Sigma_\gamma^{-1} \bm\mu_\gamma\}$.

\subsection{EM algorithm}     

The Kolmogorov-Gamma augmentation also leads to an EM algorithm that can be used to find the MLE or posterior mode under a cobin regression model with cobit link. First, when $\kappa \sim \mathrm{KG}(b,c)$, we have $E(\kappa) =b c^{-2}\{(c/2)\coth(c/2) - 1\}$ if $c\neq 0$ or $E(\kappa) = b/12$ if $c=0$, which is easy to check from the definition of the KG as an infinite convolution of gammas. 

Based on the prior $\bm\beta\sim N_p(\bm{0},\Sigma_\beta)$ and the augmented model \eqref{eq:auglik} with fixed $\lambda$, we are interested in finding the posterior mode
\[
\hat{\bm\beta} = \argmax_{\bm\beta} \int_{\bbR_+^n}  \exp(-\bm\beta^\T\Sigma_\beta^{-1}\bm\beta/2)\prod_{i=1}^n\left\{\exp( \lambda(y_i - 0.5)\bmx_i^\T\bm\beta - \kappa_i(\bmx_i^\T\bm\beta)^2/2)\pkg(\kappa_i; \lambda, 0)\right\} \rmd \bm\kappa
\]
When finding MLE, the prior term $\exp(-\bm\beta^\T\Sigma_\beta^{-1}\bm\beta/2)$ is ignored. By Theorem~\ref{thm:kgexptilting}, the conditional distribution of $\bm\kappa$ given $\bm\beta^{(t)}$ is $\kappa_i\indsim \mathrm{KG}(\lambda, \bmx_i^\T\bm\beta^{(t)})$ for $i=1,\dots,n$. Therefore, E step is
\[
Q(\bm\beta \mid \bm\beta^{(t)}) = \mathrm{constant}  -\frac{1}{2}\sum_{i=1}^n E(\kappa_i) (\bmx_i^\T\bm\beta)^2 + \sum_{i=1}^n \lambda(y_i - 0.5) \bmx_i^\T\bm\beta - \frac{1}{2}\bm\beta^\T\Sigma_\beta^{-1}\bm\beta
\]
Plugging in $\hat\kappa_i = \lambda (\bmx_i^\T\bm\beta^{(t)})^{-2}\{(\bmx_i^\T\bm\beta^{(t)}/2)\coth(\bmx_i^\T\bm\beta^{(t)}/2) - 1\}$ in place of $E(\kappa_i)$ (if $\bmx_i^\T\bm\beta^{(t)} = 0$, $\hat\kappa_i = \lambda/12$), the M step is
\[
\bm\beta^{(t+1)} = (X^\T\diag(\hat{\kappa}_1,\dots,\hat{\kappa}_n) X +  \Sigma_\beta^{-1})^{-1} \{X^\T( \lambda \bmy-0.5  \lambda \bm{1}_n)\} 
\]
Iterating the E step and M step (until $\bm\beta^{(t)}$ is stabilized given a tolerance level) gives the posterior mode (or MLE) of $\bm\beta$ under fixed $\lambda$. In fact, without the prior contribution $\exp(-\bm\beta^\T\Sigma_\beta^{-1}\bm\beta/2)$, the MLE  does not depend on $\lambda$ due to orthogonality between $\bm\beta$ and $\lambda$. 

For unknown $\lambda$ with candidates $\lambda \in\{1,\dots,L\}$ for some large $L$, one can run EM algorithm $L$ times, calculate the unnormalized posterior $p(\bm\beta)\prod_{i=1}^n p_{\mathrm{cobin}}(y_i; \bmx_i^\T\bm\beta, \lambda)$, and choose $(\hat{\bm\beta}(\lambda),\lambda)$ that maximizes it.

\section{Further details of simulation studies}
\label{appendix:simul}
\subsection{Robustness of MLE to distributional misspecification and beyond}
\label{appendix:simul_consistency}

\begin{figure}
    \centering
    \includegraphics[width=\linewidth]{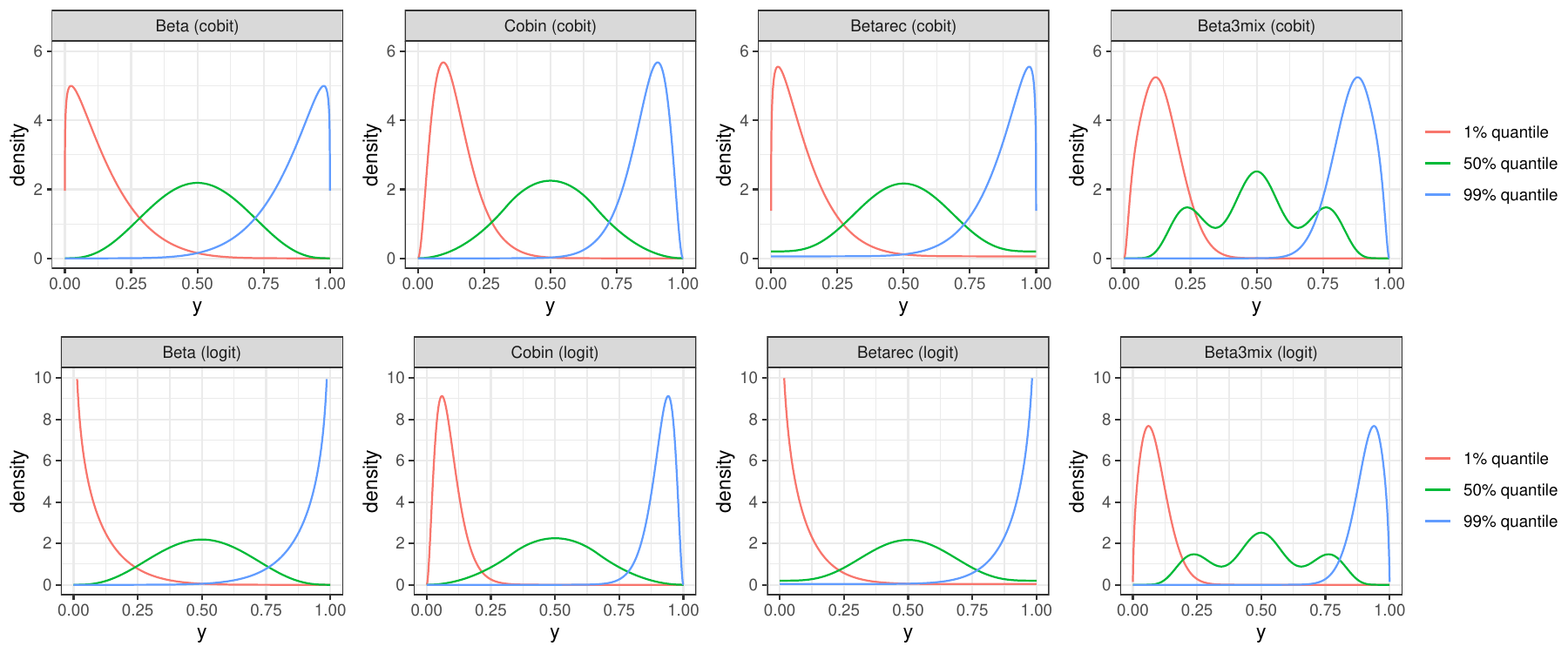}
    \caption{Data-generating distributions used in the MLE robustness simulation study. Density plots are shown for the 1\% and 99\% quantiles of the covariate distribution, illustrating variation of the response distribution across the covariate range.}
    \label{fig:simul1_data}
\end{figure}
Figure~\ref{fig:simul1_data} depicts four different data-generating distributions (beta, cobin, beta rectangular, beta mixture), where the true mean $\mu_i = g^{-1}(\bmx_i^\T\bm\beta) = g^{-1}(x_{i1})$ with a choice of link $g$ (cobit, logit) varies from the 1\%  quantile of $x_{i1}$ to the 99\% quantile of $x_{i1}$. 
In addition to bias and RMSE, Table~\ref{table:simul_consistency_supp1} summarizes the results in terms of length of the 95\% Wald confidence interval, its empirical coverage, and negative test log-likelihood (NTLL) across 1000 replicates, where NTLL is calculated based on separate held-out data of size $n/4$. The results show that, except when the true data-generating distribution is beta, cobin regression outperforms beta regression in both empirical coverage and NTLL. These differences become more prominent as $n$ increases, since the bias in the beta regression estimates has a larger impact on inference and prediction.

\begin{table}
\scriptsize
\caption{Analysis of beta and cobin regression MLE under four different data generating distributions: beta, cobin, beta rectangular (betarec), and mixture of beta (beta3mix). Results are based on 1000 replicates under the correct mean structure but possibly misspecified response distributions. CI: 95\% Wald confidence interval, cover: empirical coverage probability, NTLL: negative test log likelihood}
    \label{table:simul_consistency_supp1}
    \centering
    \begin{tabular}{c c p{0.6cm}p{0.6cm}p{0.8cm}p{0.6cm}p{0.6cm}p{0.8cm}p{0.6cm}p{0.6cm}p{0.8cm}p{0.6cm}p{0.6cm}p{0.8cm}}
    \toprule
 \multirow{3}{*}{\makecell{Method,\\link\\(data\&fit)}}  &  & \multicolumn{3}{c}{Beta data}  & \multicolumn{3}{c}{Cobin data} & \multicolumn{3}{c}{Betarec data}  & \multicolumn{3}{c}{Beta3mix data}\\
     \cmidrule(lr){3-5}\cmidrule(lr){6-8} \cmidrule(lr){9-11}\cmidrule(lr){12-14}
  &  & \multirow{2}{*}{\makecell{CI$(\beta_1)$\\length}} & \multirow{2}{*}{\makecell{CI$(\beta_1)$\\cover}} & NTLL & \multirow{2}{*}{\makecell{CI$(\beta_1)$\\length}} & \multirow{2}{*}{\makecell{CI$(\beta_1)$\\cover}} & NTLL & \multirow{2}{*}{\makecell{CI$(\beta_1)$\\length}} & \multirow{2}{*}{\makecell{CI$(\beta_1)$\\cover}} & NTLL & \multirow{2}{*}{\makecell{CI$(\beta_1)$\\length}} & \multirow{2}{*}{\makecell{CI$(\beta_1)$\\cover}} & NTLL\\ 
  & $n$ & \\
     \midrule
 \multirow{3}{*}{\makecell{beta\\regression,\\cobit}} & $100$ &0.430 & 95.4\% &  -0.501 & 0.394 & 94.1\%& -0.551 & 0.472 & 85.6\% & -0.363 & 0.384 & 95.6\% & -0.578\\
   & $400$ & 0.214 & 94.3\% & -0.515 & 0.197 & 91.3\% & -0.559 & 0.237 & 75.2\%& -0.367 & 0.192 & 89.5\% & -0.581 \\
   & $1600$ &  0.108 & 95.0\% & -0.517 & 0.099 & 79.0\% & -0.566 & 0.118 & 32.5\% & -0.378 & 0.096 & 70.7\% & -0.586\\
    \cmidrule(lr){1-14}
  \multirow{3}{*}{\makecell{cobin\\regression,\\cobit}} & $001$ & 0.422 & 94.1\% & -0.473 &0.384 & 94.0\% & -0.563 & 0.484 & 93.0\% & -0.386 & 0.370 & 95.8\% & -0.594 \\
    & $400$ & 0.209 & 93.8\% & -0.481 &0.191 & 94.6\% &-0.574 & 0.243 & 91.5\% & -0.393 & 0.183 & 96.5\% & -0.598\\
    & $1600$ & 0.105 & 94.3\% & -0.483 & 0.095 & 95.6\%& -0.580 &0.122 & 92.0\% & -0.402 & 0.092 & 96.3\% & -0.603\\
    \midrule
 \multirow{3}{*}{\makecell{beta\\regression,\\logit}} & $100$ & 0.322 & 95.0\% & -0.555 & 0.297 & 94.2\% & -0.599 & 0.360 & 87.9\% & -0.415 & 0.291 & 94.3\% & -0.617\\
    & $400$ & 0.161 & 94.7\% & -0.568 & 0.148 & 78.9\% & -0.606 & 0.180 & 72.8\% & -0.427 & 0.145 & 80.3\% & -0.623\\
    & $1600$ & 0.080 & 94.0\% & -0.570 &0.074 & 32.0\% & -0.607 &0.090 & 28.2\% & -0.431 & 0.073 & 26.9\% & -0.628\\
    \cmidrule(lr){1-14}
 \multirow{3}{*}{\makecell{cobin\\regression,\\logit}} & $100$ & 0.298 & 85.6\% & -0.483 & 0.263 & 96.1\% & -0.625 & 0.344 & 85.3\% & -0.397 & 0.254 & 92.0\% & -0.635\\
    & $400$ & 0.149 & 86.7\% & -0.498 & 0.130 & 92.8\% & -0.635 & 0.174 & 83.9\% & -0.414 & 0.126 & 94.1\% & -0.645\\
    & $1600$ & 0.074 & 84.7\% & -0.505 & 0.065 & 95.7\% & -0.635 & 0.088 & 82.5\% & -0.418 & 0.063 & 94.2\% & -0.648\\
    \bottomrule
    \end{tabular}
    \begin{flushleft} 
{\footnotesize Monte Carlo standard errors are all less than 0.002 for CI length, 1.5\% for cover, and 0.008 for NTLL across all settings.}  
\end{flushleft}
\end{table}

We conduct an additional simulation study to evaluate the predictive performance of cobin regression when even the mean structure is misspecified. We fitted two different logit link models (beta, cobin) to four different datasets (beta, cobin, beta rectangular, beta mixture) generated from a cobit link model. We also performed the reverse, by fitting two different cobit link models (beta, cobin) to four different datasets (beta, cobin, beta rectangular, beta mixture) generated from a logit link model. 

Furthermore, we compare with the data transformation approach by fitting a robust linear model to the transformed data on the real line and then back-transforming the prediction, which also represents the case when the mean structure is misspecified. Specifically, we applied the cobit transformation $g_{\mathrm{cobit}}:(0,1)\to \mathbb{R}$ to the data generated under a logit link (across four different response distributions) and fitted a robust linear model, and also applied the logit transformation to data generated under a cobit link (also across four different response distributions). On the transformed responses $\tilde{y}_i$, robust linear model is fit using the \texttt{GAMLSS} R package \citep{Stasinopoulos2007-th} assuming Student-t errors (with 3 degrees of freedom, fixed), to obtain $\hat{\bm\beta}_{\mathrm{trans}}$ and $\hat{\sigma}_{\mathrm{trans}}$ from the model $\tilde{y}_i\mid \bmx_i \indsim t_3(\bmx_i^\T \bm\beta_{\mathrm{trans}}, \sigma_{\mathrm{trans}})$ for $i=1,\dots,n$, where $t_3(\mu,\sigma)$ denotes a location-scale parameterization of t distribution. Predictive performance was evaluated on the original response scale. Specifically, negative test log likelihood (NTLL) was computed using the inverse transformed $t_3$ density supported on $(0,1)$ evaluated at the test data, and the calculation of MSPE was also based on the inverse-transformed prediction. For example, when the logit transformation is used for data transformation, MSPE was computed using $\hat{y}_{\mathrm{test}} = 1/(1 + \exp(-\bmx_{\mathrm{test}}^\T \hat{\bm\beta}_{\mathrm{trans}})) \in (0,1)$, following standard practice in the transformation-based approach.

Table~\ref{table:simul_consistency_supp2} summarizes the results in terms of negative test log-likelihood (NTLL) and mean squared prediction error (MSPE) on held-out data of size $n/4$ across 500 replicates. Overall, except when the true data-generating distribution is beta, cobin regression exhibits comparable predictive performance relative to beta regression, even under link misspecification. In particular, cobin regression provides better predictive accuracy than beta regression in terms of NTLL when the data are generated from the three-point beta mixture model, which represents a more complex and multimodal scenario.

From Table~\ref{table:simul_consistency_supp2}, the transformation-based approach generally performs worse than both beta and cobin regression. In particular, predictive accuracy notably deteriorates in two settings: (1) data from a logit link with a beta response and (2) data from a logit link with a beta rectangular response. These scenarios correspond to situations where boundary-proximate responses are more common (see Figure~\ref{fig:simul1_data}), highlighting the well-known sensitivity of the transformation-based approach to boundary-proximate observations, a limitation that persists even when a heavy-tailed error distribution is employed.

\begin{table}
\scriptsize
\caption{Predictive performance of beta regression, cobin regression, and a data transformation approach under misspecified mean settings. Results are based on 500 simulation replicates.}
    \label{table:simul_consistency_supp2}
    \centering
    \begin{tabular}{c c c c c c c c c c c}
    \toprule
 &  &  & \multicolumn{2}{c}{Beta data}  & \multicolumn{2}{c}{Cobin data} & \multicolumn{2}{c}{Betarec data}  & \multicolumn{2}{c}{Beta3mix data}\\
     \cmidrule(lr){4-5}\cmidrule(lr){6-7} \cmidrule(lr){8-9}\cmidrule(lr){10-11}
    & Method & $n$ & NTLL & MSPE$^\dagger$ & NTLL & MSPE$^\dagger$ & NTLL & MSPE$^\dagger$ & NTLL & MSPE$^\dagger$\\ 
     \midrule
 \multirow{9}{*}{\makecell{Data from\\cobit link}}   & \multirow{3}{*}{\makecell{beta regression,\\logit link}} & $100$ & -0.492 & 0.053& -0.538 & 0.055 & -0.360 & 0.093 & -0.570 & 0.051\\
   & & $400$ & -0.509 & 0.024 & -0.555 & 0.029 & -0.360 & 0.046 & -0.576 & 0.029\\
   & & $1600$ & -0.513 & 0.016 & -0.557 & 0.021 & -0.372 & 0.036& -0.579 & 0.022\\
    \cmidrule(lr){2-11}
   & \multirow{3}{*}{\makecell{cobin regression,\\logit link}} & $100$ & -0.462 & 0.064  & -0.550 & 0.055 & -0.377 & 0.076 & -0.585 & 0.051\\
   & & $400$ & -0.471 & 0.031 & -0.568 & 0.029 & -0.383 & 0.037 & -0.592 & 0.027\\
   & & $1600$ & -0.476 & 0.022 & -0.572 & 0.021 & -0.392 & 0.024 & -0.596 & 0.020\\
       \cmidrule(lr){2-11}
   & \multirow{3}{*}{\makecell{logit transform,\\linear regression \\ with $t_3$ error}} & $100$ & -0.451 & 0.087 & -0.522 & 0.062 &  -0.426 & 0.118 & -0.544 & 0.057\\
   & & $400$ & -0.465 & 0.040 & -0.522 & 0.062 &-0.428&  0.076& -0.550 & 0.023\\
   & & $1600$ & -0.471 & 0.029 & -0.538 & 0.025 & -0.436 & 0.066 & -0.551 & 0.015 \\
    \midrule
 \multirow{9}{*}{\makecell{Data from\\logit link}}  &\multirow{3}{*}{\makecell{beta regression,\\cobit}} & $100$ & -0.552 & 0.059 & -0.599 & 0.054 & -0.409 & 0.080 & -0.609 & 0.052\\
   & & $400$ & -0.557 & 0.028 & -0.608 & 0.026 & -0.422 & 0.035 & -0.620 & 0.025\\
   & & $1600$ & -0.561 & 0.020 & -0.608 & 0.018 & -0.426 & 0.023 & -0.628 & 0.018\\
    \cmidrule(lr){2-11}
   & \multirow{3}{*}{\makecell{cobin regression,\\cobit}} & $100$ & -0.485 & 0.065 & -0.612 & 0.058 & -0.397 & 0.081 & -0.618 & 0.056\\
   & & $400$ &  -0.488 & 0.031 & -0.622 & 0.029 & -0.408 & 0.037 & -0.632 & 0.028\\
   & & $1600$ & -0.497 & 0.023 & -0.622 & 0.023 & -0.409 & 0.024 & -0.636 & 0.022\\
    \cmidrule(lr){2-11}
      & \multirow{3}{*}{\makecell{cobit transform,\\linear regression \\ with $t_3$ error}} & $100$ & -0.233 & 0.125 & -0.488 & 0.124 &-0.186 & 0.216 & -0.491 & 0.099\\
   & & $400$ & -0.236 & 0.069&-0.502 & 0.084 & -0.191 & 0.152 & -0.503 & 0.055\\
   & & $1600$ & -0.247 & 0.059 & -0.502 & 0.078& -0.196 & 0.137 & -0.506 & 0.047\\
    \bottomrule
    \end{tabular}
    \begin{flushleft} 
{\footnotesize $^\dagger$MSPE are scaled by 100. Monte Carlo standard errors of NTLL are all less than 0.012 for $n=100$, 0.006 for $n=400$, and 0.003 for $n=1600$. Monte Carlo standard errors of MSPE scaled by 100 are all less than 0.007.}  
\end{flushleft}
\end{table}

\subsection{Robustness and scalability under spatial regression models}
\label{appendix:simul_spatial}

We first describe the prior specification of dispersion parameters under three different regression models. 
First, for micobin regression, we consider the $\psi\sim \mathrm{Beta}(2,2)$ prior, reflecting the prior belief that $E\{\var(Y)\} = 0.5V(\mu)$ for a given mean $\mu = E(Y)$. For the prior of $\lambda$ for cobin regression, we consider $p(\lambda) = 36\lambda\Gamma(\lambda+1)/\Gamma(\lambda+5)$ for $\lambda = 1,2,\dots$, which is derived from the marginal distribution of $\lambda$ when $(\lambda-1)\mid \psi \sim \mathrm{Negbin}(2,\psi)$ and $\psi\sim \mathrm{Beta}(2,2)$. In practice, we truncate the support of $\lambda$ at some large upper bound $L$, where we choose $L=70$ throughout the paper. 

The precision parameter $\phi$ in the beta regression model (the sum of two beta shape parameters) satisfies $\var(Y) = \mu(1-\mu)/(1+\phi)$. While there are several different choices of $p(\phi)$ available, we choose squared uniform distribution $\phi\stackrel{d}{=} U^2, U\sim \mathrm{Unif}(0,A)$ for some $A>0$, following the suggestion of \citet{Figueroa-Zuniga2013-pd}. To match the prior belief of $E(\psi) = 1/2$ so that $E(\var(Y)) = 0.5V(\mu)$, since $\mu(1-\mu)/3$ and $V(\mu)$ operates on a similar scale (see Figure~\ref{fig:linkft}), we choose $A = 8.74$ so that the prior of $\phi$ satisfies $E(3/(1+\phi)) \approx 1/2$. The beta regression model is fitted with \texttt{Stan} version 2.32.2 along with rstan version 2.32.6, with all default options using No-U-Turn sampler. 

Next, we describe how negative test log-likelihood (negtestLL) and mean squared prediction error (MSPE) are calculated. Let $\{\bm\theta^{(m)}\}_{m=1}^M$ be a set of parameter samples from MCMC output. For each parameter sample $\bm\theta^{(m)}$, posterior predictive samples at new locations $\{s_i^*\}$ can be generated as $\eta(s_i^*)^{(m)} = \bmx(s_i^*)^\T\bm\beta^{(m)} + u(s_i^*)^{(m)}$ for $i=1,\dots,n_{\mathrm{test}}$, where $(u(s_1^*), \dots, u(s_{n_{\mathrm{test}}}^*))^{(m)}$ is drawn from a multivariate normal conditioned on the $m$th sample of the spatial random effects at $n_{\mathrm{train}}$ locations. Based on the true mean $\mu(s^*_i)$ and held-out realizations $y(s_i^*)$ for $i=1,\dots,n_{\mathrm{test}}$, the negative test log-likelihood based on the assumed model $p(y(s_i^*)\mid \bm\theta^{(m)})$ (beta, cobin, micobin) is calculated as
\[
\text{negtestLL} = -\frac{1}{n_{\mathrm{test}}}\sum_{i=1}^{n_{\mathrm{test}}} \log \left\{\frac{1}{M}\sum_{m=1}^Mp(y(s_i^*)\mid \bm\theta^{(m)})\right\}.
\]
The MSPE is calculated as
\[
\text{MSPE} = \frac{1}{n_{\mathrm{test}}}\sum_{i=1}^{n_{\mathrm{test}}}\{\mu(s_i^*) - \hat\mu(s_i^*)\}^2
\]
where $\hat\mu(s_i^*)$ is a posterior mean estimate of the conditional mean at a new location $s_i^*$, i.e. $M^{-1}\sum_{m=1}^M B'(\eta(s_i^*)^{(m)})$. The reported quantities in Table~\ref{table:sim_spatial} are averaged over 100 replicates. Computations are carried out under the Intel(R) Xeon(R) Gold 6336Y 2.40GHz CPU environment.

\begin{table}
\footnotesize
\caption{Bayesian spatial regression simulation results based on 100 replicates, under the correct mean structure but misspecified distributions.}
    \label{table:sim_spatial_supp}
    \centering
    \begin{tabular}{c c c c c c}
    \toprule
  & & \multicolumn{2}{c}{$\rho = 0.1$ (moderate)} & \multicolumn{2}{c}{$\rho = 0.2$ (strong)}   \\
    \cmidrule(lr){3-4}\cmidrule(lr){5-6} 
    Method  & $(n_{\mathrm{train}},n_{\mathrm{test}})$ & CI$_{.95}(\beta_1)$ length  & CI$_{.95}(\beta_1)$ coverage &   CI$_{.95}(\beta_1)$ length  & CI$_{.95}(\beta_1)$ coverage\\ 
     \midrule
     \multirow{2}{*}{\makecell{beta\\
    regression}} & $(200,50)$  & 0.388 & 89.0\% & 0.366 & 85.0\% \\
     & $(400,100)$ & 0.264 & 85.0\% & 0.252 & 80.0\% \\
    \cmidrule(lr){1-6}
    \multirow{2}{*}{\makecell{cobin\\
    regression}} & $(200,50)$  & 0.353 & 93.0\% & 0.348 & 95.0\%\\
     & $(400,100)$ & 0.246 & 94.0\% & 0.242 & 90.0\% \\
     \cmidrule(lr){1-6}
     \multirow{2}{*}{\makecell{micobin\\
    regression}} & $(200,50)$  & 0.352 & 91.0\% & 0.348 & 97.0\% \\
     & $(400,100)$ & 0.244 & 89.0\% & 0.240 & 81.0\%\\
    \bottomrule
    \end{tabular}
\begin{flushleft} 
{\footnotesize Monte Carlo standard errors are all less than 0.004 for 95\% CI length and 4\% for coverage.}  
\end{flushleft}
\end{table}

Table~\ref{table:sim_spatial_supp} provides additional summaries of the spatial regression simulation results in terms of 95\% credible interval lengths and empirical coverage. Among the beta, cobin, and micobin models, cobin regression achieves empirical coverage closest to the nominal 0.95 level, also suggesting that it is better suited for parameter inference under potential misspecification.

\subsection{Robustness to the presence of outliers}
\label{appendix:simul_outlier}

Figure~\ref{fig:simul_outlier_data} depicts two different data-generating distributions (beta, cobin), where the true mean $\mu_i = g^{-1}(\bmx_i^\T\bm\beta) = g_{\mathrm{cobit}}^{-1}(x_{i1})$ varies from the 1\%  quantile of $x_{i1}$ to the 99\% quantile of $x_{i1}$. It also describes an outlier $y^\circ = 10^{-3}$ and the conditional distribution of the response based on the  value of $\bmx^\circ = (6,6)^\T$, which corresponds to the conditional mean of 0.5. It also describes representative quantile residual plots \citep{Dunn1996-gb} for four different models fitted to the beta data with $y^\circ = 10^{-3}$, illustrating that an outlier exhibits poor goodness of fit under the beta and cobin regression models. 

For prior specification of parameters not related to the mean, we followed the same settings as in the spatial regression simulation. Specifically, we used $p(\lambda) \propto \lambda \Gamma(\lambda + 1)/\Gamma(\lambda + 5)$ for $\lambda = 1, 2, \dots, L = 70$ in the cobin model; a $\mathrm{Beta}(2, 2)$ prior for $\psi$ in the micobin model; $\phi \stackrel{d}{=} U^2$ with $U \sim \mathrm{Unif}(0, 8.74)$ for the beta model; and for the beta rectangular model, $\phi \stackrel{d}{=} U^2$ with $U \sim \mathrm{Unif}(0, 8.74)$ and $\alpha \sim \mathrm{Unif}(0,1)$ independently.  Computations are also carried out under the Intel(R) Xeon(R) Gold 6336Y 2.40GHz CPU environment.

\begin{figure}
    \centering
    \includegraphics[width=.9\linewidth]{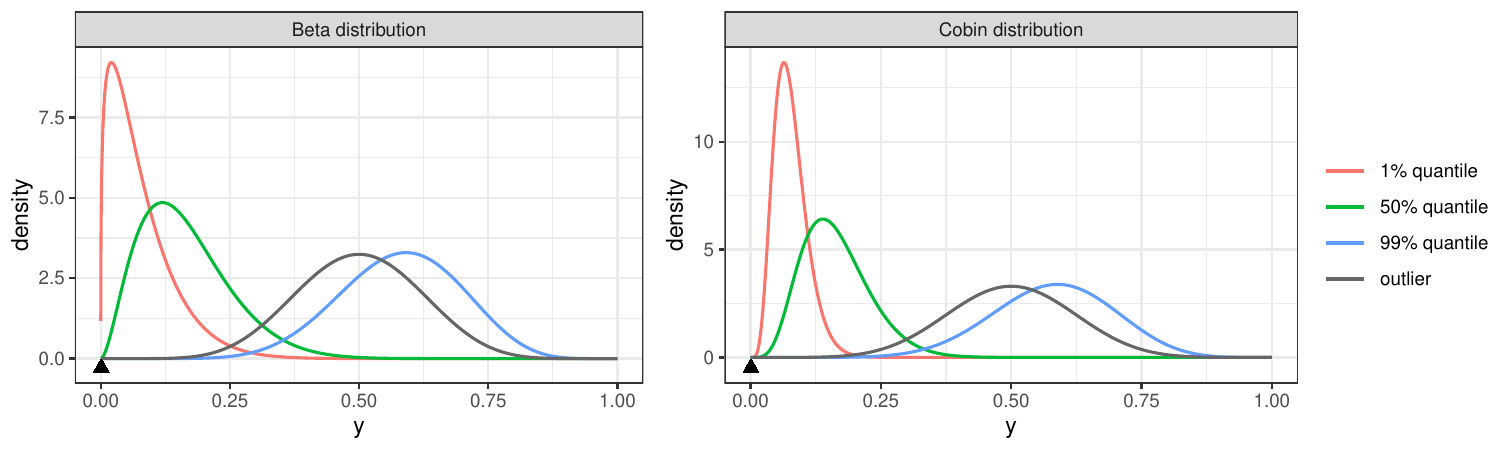}
    \includegraphics[width=.9\linewidth]{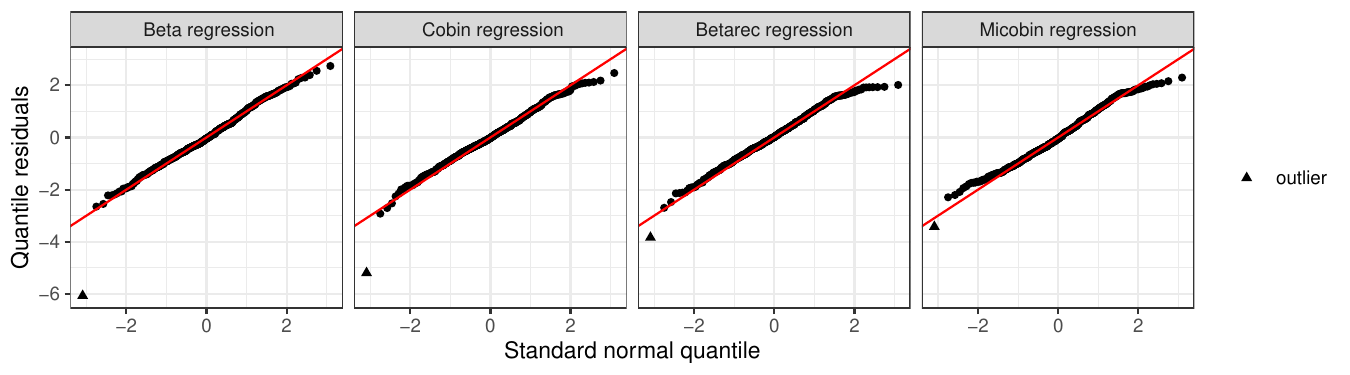}
    \caption{(Top) Data-generating distributions used in the outlier robustness simulation study. Density plots are shown for the 1\% and 99\% quantiles of the covariate distribution. Outlier density corresponds to the conditional distribution of response given $\bmx^\circ=(6,6)^\T$, and a triangle represents $y^\circ = 10^{-3}$.   
    (Bottom) Quantile residual plots from four different models fitted on beta data with an outlier $y^\circ = 10^{-3}$, the quantile residual corresponding to the outlier is marked as a triangle. }
    \label{fig:simul_outlier_data}
\end{figure}

\begin{table}
\footnotesize
\caption{Outlier robustness simulation results averaged over 100 replicates, summarized in terms of posterior mean estimates and sampling efficiency.}
    \label{tab:outliersim_supp}
    \centering
    \begin{tabular}{c c c c c c c c c c}
    \toprule
\multirow{2}{*}{Method} &  \multirow{2}{*}{Setting}&  \multicolumn{2}{c}{$\beta_0^{\mathrm{true}}=-6$} & \multicolumn{2}{c}{$\beta_1^{\mathrm{true}}=1$}   & \multicolumn{2}{c}{$\beta_2^{\mathrm{true}}=0$} & \multicolumn{2}{c}{ Sampling $(\bm\beta)$} \\
    \cmidrule(lr){3-4} \cmidrule(lr){5-6}\cmidrule(lr){7-8}\cmidrule(lr){9-10}
     & & $\hat\beta_0^\circ$ & $\hat\beta_0$ & $\hat\beta_1^\circ$ & $\hat\beta_1$  & $\hat\beta_2^\circ$ & $\hat\beta_2$ & mESS & time (sec)  \\  
 \midrule
\multirow{2}{*}{\makecell{beta\\regression}} 
  & A1 & -5.899 & \multirow{2}{*}{-5.995} & 0.938 & \multirow{2}{*}{0.999} & -0.061 & \multirow{2}{*}{-0.004} & 4629.6 & 9.5 \\
  & A2 & -5.850 &  & 0.908 &  & -0.087 &  & 4650.4 & 9.1 \\
\cmidrule(lr){1-10}
\multirow{2}{*}{\makecell{cobin\\regression}} 
  & A1 & -5.993 & \multirow{2}{*}{-5.993} & 0.982 & \multirow{2}{*}{0.999} & -0.025 & \multirow{2}{*}{-0.004} & 1847.2 & 25.0 \\
  & A2 & -5.994 &  & 0.982 &  & -0.025 &  & 1846.3 & 23.8 \\
\cmidrule(lr){1-10}
\multirow{2}{*}{\makecell{betarec\\regression}} 
  & A1 & -5.916 & \multirow{2}{*}{-5.955} & 0.986 & \multirow{2}{*}{0.992} & -0.004 & \multirow{2}{*}{-0.004} & 4727.1 & 22.7 \\
  & A2 & -5.916 &  & 0.986 &  & -0.004 &  & 4732.8 & 21.8 \\
\cmidrule(lr){1-10}
\multirow{2}{*}{\makecell{micobin\\regression}} 
  & A1 & -5.929 & \multirow{2}{*}{-5.928} & 0.981 & \multirow{2}{*}{0.986} & -0.011 & \multirow{2}{*}{-0.004} & 1282.4 & 73.9 \\
  & A2 & -5.926 &  & 0.981 &  & -0.010 &  & 1302.9 & 72.8 \\
\midrule  
\multirow{2}{*}{\makecell{beta\\regression}} 
  & B1 & -5.835 & \multirow{2}{*}{-5.934} & 0.915 & \multirow{2}{*}{0.977} & -0.056 & \multirow{2}{*}{0.002} & 4656.7 & 9.6 \\
  & B2 & -5.785 &  & 0.884 &  & -0.083 &  & 4747.7 & 9.3 \\
\cmidrule(lr){1-10}
\multirow{2}{*}{\makecell{cobin\\regression}} 
  & B1 & -6.010 & \multirow{2}{*}{-6.010} & 0.990 & \multirow{2}{*}{1.006} & -0.018 & \multirow{2}{*}{0.003} & 1849.2 & 35.4 \\
  & B2 & -6.010 &  & 0.989 &  & -0.018 &  & 1849.4 & 33.9 \\
\cmidrule(lr){1-10}
\multirow{2}{*}{\makecell{betarec\\regression}} 
  & B1 & -5.822 & \multirow{2}{*}{-5.866} & 0.958 & \multirow{2}{*}{0.965} & 0.001 & \multirow{2}{*}{0.001} & 4408.8 & 24.3 \\
  & B2 & -5.821 &  & 0.958 &  & 0.001 &  & 4346.7 & 23.0 \\
\cmidrule(lr){1-10}
\multirow{2}{*}{\makecell{micobin\\regression}} 
  & B1 & -6.003 & \multirow{2}{*}{-6.003} & 1.002 & \multirow{2}{*}{1.005} & -0.002 & \multirow{2}{*}{0.002} & 1139.2 & 89.4 \\
  & B2 & -6.003 &  & 1.002 &  & -0.002 &  & 1154.0 & 86.9 \\
\bottomrule
    \end{tabular}
\begin{flushleft} 
{\footnotesize Monte Carlo standard errors are all less than 0.019 for  $\hat{\beta}^\circ_0$ and $\hat{\beta}_0$, 0.005 for $\hat{\beta}^\circ_0$ and $\hat{\beta}_0$, 0.004 for $\hat{\beta}^\circ_0$ and $\hat{\beta}_0$, 43.5 for mESS, and 1.9 sec for runtime.}  
\end{flushleft}
\end{table}

In addition to parameter stability and empirical coverage, Table~\ref{tab:outliersim_supp} summarizes the results in terms of individual parameter estimates (with and without $y^\circ$) and sampling efficiency of $\bm\beta$. These results further support the robustness of cobin and micobin regression methods in the presence of outliers, since beta regression exhibits substantial bias. 
Notably, the results suggest that the poor goodness of fit of the model to the outlying observations, which is evident in both beta and cobin regression, does not necessarily lead to degraded parameter stability, as illustrated by the stable cobin regression estimates across settings. 
Although the gains in sampling efficiency through Kolmogorov-Gamma augmentation are less pronounced in simpler fixed-effects models, where the runtime of beta and beta rectangular regression models implemented in \texttt{Stan} (excluding compilation time) is relatively short, they become more substantial in hierarchical settings, as demonstrated in the previous spatial regression simulation study. 

\section{Additional information for MMI data analysis}
\label{appendix:mmi}
\subsection{Data processing and detailed settings}
\label{appendix:mmi_data}
The benthic macroinvertebrate multivariate index (MMI) data from the 2017 National Lakes Assessment Survey \citep{US-Environmental-Protection-Agency2022-dp} and lake watershed covariates from LakeCat data \citep{Hill2018-eq} are joined based on the unique identifier of the lake (COMID). Table~\ref{table:mmicovariate} outlines the details of the 9 selected LakeCat covariates.

\begin{table}
    \caption{Description of covariates from LakeCat data \citep{Hill2018-eq}.  NLCD refers to the National Land Cover Dataset. Watershed refers to a set of hydrologically aggregated catchments (including lake waterbody surface area) that represent the full contributing landscape area to a downslope lake. All variables are $\log_2(x+1)$ transformed. }
\footnotesize
    \centering
    \begin{tabular}{l l l}
    \toprule
    Variable (unit) & Brief and detailed description\\
    \midrule
    \makecell[tl]{agkffact \\ (unitless)} & \makecell[tl]{\textit{Ag soil erodibility Kf factor}\\ Mean of state soil geographic Kf factor raster on agricultural land (NLCD 2006) within \\ the
    watershed. Kf factor is a soil erodibility factor which quantifies the susceptibility \\
    of soil particles to detachment and movement by water. This factor is used in the \\ universal soil loss equation to calculate soil loss by water}\\
    \midrule
    \makecell[tl]{bfi\\(\%)} & \makecell[tl]{\textit{Base flow index} \\ The component of streamflow that can be attributed to ground-water discharge into \\
    streams. bfi is the ratio of base flow to total flow within the watershed}\\
    \midrule
    \makecell[tl]{cbnf\\(kgN$\cdot$ha$^{-1}\cdot$yr$^{-1}$)} & \makecell[tl]{\textit{Cultivated Biological Nitrogen Fixation Mean Rate}\\ 
   Mean rate of biological nitrogen fixation from the cultivation of crops \\
   within the watershed} \\
   \midrule
    \makecell[tl]{conif\\ (\%)} & \makecell[tl]{\textit{Evergreen forest percentage} \\
    \% of watershed area classified as evergreen forest land cover (NLCD class 42) in 2016}\\
    \midrule
    \makecell[tl]{crophay\\ (\%)} & \makecell[tl]{\textit{Row crop and pasture/hay percentage} \\ \% of watershed area classified as crop and hay land use (NLCD classes 81, 82) in 2016}\\
    \midrule
    \makecell[tl]{fert \\ (kgN$\cdot$ha$^{-1}\cdot$yr$^{-1}$)} & \makecell[tl]{\textit{Synthetic nitrogen fertilizer application mean rate}\\
    Mean rate of synthetic nitrogen fertilizer application to agricultural land \\ within the watershed}\\
    \midrule
    \makecell[tl]{manure \\ (kgN$\cdot$ha$^{-1}\cdot$yr$^{-1}$) } & \makecell[tl]{ \textit{Mean manure application rate} \\Mean rate of manure application to agricultural land from confined animal feeding \\operations within the watershed}\\
    \midrule
    \makecell[tl]{pestic97 \\ (kg/km$^2$)} & \makecell[tl]{\textit{Mean pesticide use} \\ Mean pesticide use in year 1997 within the watershed}\\
    \midrule
    \makecell[tl]{urbmdhi\\(\%)}& \makecell[tl]{\textit{Developed, medium and high intensity land use percentage} \\\% of watershed area classified as developed, medium and high intensity land use \\(NLCD classes 23, 24) in 2016}\\
    \bottomrule
    \end{tabular}
    \label{table:mmicovariate}
\end{table}

We standardized all (log-transformed) covariates to have a mean of 0 and a variance of 1 prior to running MCMC, and transformed back to the original scale after sampling. 
Based on these standardized covariates, we used the same prior as in the simulation study. That is, $\bm\beta\sim N_p(\bm{0}, 100^2I_p)$ for regression coefficients, a standard half-Cauchy prior on the spatial random effect standard deviation $\sigma_u$, beta dispersion $\phi\stackrel{d}{=} U^2, U\sim \mathrm{Unif}(0,8.74)$, cobin dispersion $p(\lambda) \propto \lambda\Gamma(\lambda+1)/\Gamma(\lambda+5)$ for $\lambda = 1,2,\dots$, and micobin dispersion $\psi\sim \mathrm{Beta}(2,2)$. We set the upper bound of $\lambda$ to be $L=70$. All computations for the MMI data analysis were performed on an Apple M1 CPU.

The MMI data analysis involves an NNGP prior for the spatial random effects. To enable nearest neighbor calculation using the \texttt{spNNGP} \texttt{R} package \citep{Finley2022-pg}, we used the WGS 84 / UTM zone 15N (EPSG:32615) coordinate system. The NNGP prior requires the specification of an ordering and a number of neighborhoods. Following the default setting of \texttt{spNNGP} \texttt{R} package, we use coordinate-based ordering and use 15 nearest neighbors. We use an exponential covariance kernel $\cov\{u(s), u(s')\} = \sigma_u^2\exp(-\|s-s'\|/\rho)$ in the construction of the NNGP prior. We fixed the spatial range parameter at $\rho = 200$km, leading to an effective range of approximately 600km, the distance where the spatial correlation is below 0.05 (before NNGP approximation). This specification is consistent with the analysis of \citet{Fox2020-xt}, where the empirical spatial autocorrelation of stream benthic macroinvertebrate MMI in the US was found to be close to zero for distances beyond 580km.

\subsection{Quantile residual analysis for identifying potential outliers}
\label{appendix:qresidual}

\begin{figure}
    \centering
    \includegraphics[width=\linewidth]{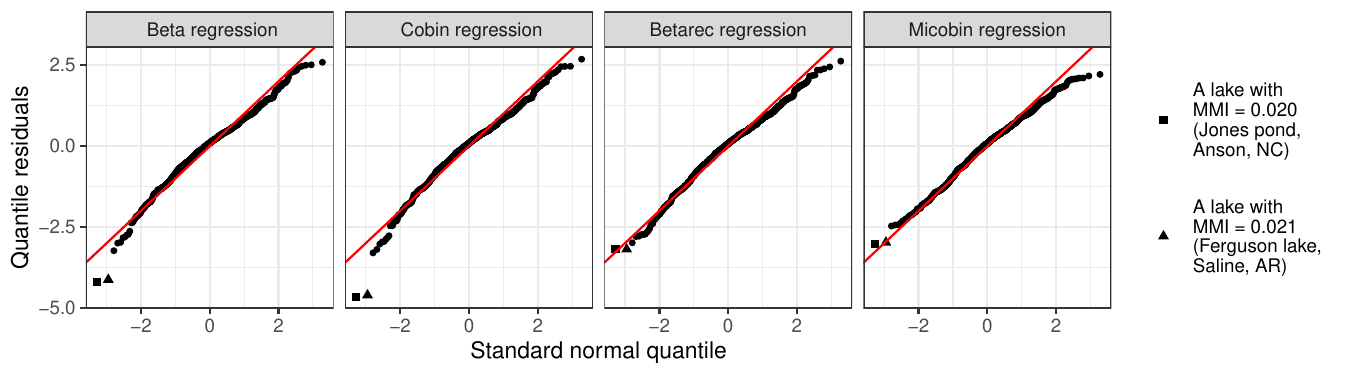}
    \caption{Comparison of quantile residuals for goodness-of-fit assessment, along with two observations corresponding to the lowest quantiles. The red line corresponds to the $y=x$ line.}
    \label{fig:qresidplot}
\end{figure}

We first present the quantile residual \citep{Dunn1996-gb} analysis results from four different models based on the dataset $\calD^\circ$, where quantile residuals are computed conditional on the spatial random effects, with model parameters fixed at the posterior mean (posterior median of $\lambda$ for cobin). 
Figure~\ref{fig:qresidplot} shows that beta and cobin regression exhibit a lack of fit in the left tail, whereas beta rectangular and micobin regression capture the left tail accurately but show slightly degraded fit in the right tail.

Further inspection reveals that the two observations with the lowest quantile residuals are identical for all four models, with MMI values of 0.02 and 0.021. This suggests that these lakes are not well explained by the beta and cobin models, as their low MMI values fall outside the expected range of residual variation under those models. In contrast, beta rectangular and micobin regression accommodate these low values more naturally, but at the same time imply greater dispersion in the right tail than what is observed in the data. This reflects the fundamental characteristic of the beta rectangular and micobin models based on mixture formulation, which allows for improved robustness to extreme values, but at the cost of less sharp probabilistic predictions. See also Figure~\ref{fig:simul_outlier_data} for similar behavior of quantile residuals. 

In summary, the quantile residual analysis identifies two low MMI lakes as potential outliers under the beta and cobin regression models. One of the key goals of the MMI data analysis is to evaluate and compare the robustness of the models in the presence of such outliers.

\subsection{Additional summary and figures}
\label{appendix:mmi_results}
\begin{table}[t]
\scriptsize
    \caption{Estimated fixed-effect coefficients $(\hat{\bm\beta}^\star)$ from two spatial regression models (beta rectangular, micobin) fitted on $\calD^\star (n = 950)$ and 95\% credible intervals. Bold entries indicate non-intercept coefficients whose 95\% credible intervals exclude zero.}
    \label{table:mmiresult_n950}
    \centering
\begin{tabular}{l c c c c  }
\toprule
 & \multicolumn{2}{c}{Betarec regression ($\|\hat{\bm\beta}^\star - \hat{\bm\beta}\|_2 = 0.103$ )}  & \multicolumn{2}{c}{Micobin regression ($\|\hat{\bm\beta}^\star - \hat{\bm\beta}\|_2 = 0.048$ )} \\
\cmidrule(lr){2-3} \cmidrule(lr){4-5}
Variable & $\hat\beta_j^\star$ ($n=950$) & $\hat\beta_j$  ($n=947$) & $\hat\beta_j^\star$ ($n=950$) & $\hat\beta_j$  ($n=947$)\\
\midrule
(Intercept) & -1.596 (-3.385, 0.156) & -1.696 (-3.543, 0.092) & -1.758 (-3.517, -0.040) & -1.741 (-3.520, 0.018)\\
agkffact & \textbf{-3.072} (-5.916, -0.288) & \textbf{-3.082} (-5.961, -0.200) &\textbf{-3.456} (-6.175, -0.794) & \textbf{-3.494} (-6.220, -0.822) \\
bfi & 0.202 (-0.113, 0.530) & 0.220 (-0.101, 0.551) & 0.219 (-0.100, 0.537) & 0.219 (-0.097, 0.540) \\
cbnf & 0.200 (-0.033, 0.437) & 0.200 (-0.036, 0.437) & 0.187 (-0.040, 0.415) & 0.197 (-0.034, 0.424) \\
conif & \textbf{0.107} (0.026, 0.186) & \textbf{0.103} (0.023, 0.184) & \textbf{0.128} (0.048, 0.208)& \textbf{0.125} (0.045, 0.204) \\
crophay & -0.036 (-0.202, 0.132) & -0.042 (-0.209, 0.124) & -0.060 (-0.222, 0.101) & -0.050 (-0.210, 0.110) \\
fert & -0.127 (-0.357, 0.100) & -0.120 (-0.348, 0.105) & -0.071 (-0.296, 0.130) & -0.089 (-0.316, 0.135) \\
manure & -0.018 (-0.166, 0.130) & -0.013 (-0.165, 0.137) & -0.031 (-0.178, 0.115) & -0.022 (-0.167, 0.122) \\
pestic97 & -0.034 (-0.121, 0.053) & -0.034 (-0.121, 0.054) & -0.023 (-0.106, 0.059) & -0.027 (-0.110, 0.055) \\
urbmdhi & \textbf{-0.162} (-0.266, -0.058) & \textbf{-0.172} (-0.278, -0.067) & \textbf{-0.141} (-0.243, -0.038) & \textbf{-0.143} (-0.242, -0.043) \\
\midrule
mESS$(\bm\beta)$/t & 3711.9/117 mins & 3211.1/96 mins & 2995.2/5 mins & 3012.6/5 mins\\
PSIS-LOO& -1101.6$(\calD^\star)$ &-1124.8$(\calD)$ & -1107.2$(\calD^\star)$ &-1126.9$(\calD)$\\
WAIC &-1106.4$(\calD^\star)$ & -1131.3$(\calD)$ & -1111.0$(\calD^\star)$ & -1131.1$(\calD)$\\
\bottomrule
\end{tabular}
    \vspace{-3mm}
\begin{flushleft} 
{\scriptsize 
agkffact, soil erodibility factor; bfi, base flow index; cbnf, cultivated biological N fixation; conif, coniferous forest cover; crophay, crop/hay land cover; fert, synthetic N fertilizer use; manure, manure application; pestic97, 1997 pesticide use; urbmdhi, medium/high-density urban land cover; mESS$(\bm\beta)$, t, multivariate effective sample size and wall-clock fitting time (in mins), averaged over 3 chains; PSIS-LOO, leave-one-out cross-validation estimate using Pareto smoothed importance sampling; WAIC, widely applicable information criterion. All variables are $\log_2(x+1)$ transformed. 
}
\end{flushleft}
\end{table}

\begin{table}[t]
\scriptsize
\caption{Estimated fixed-effect coefficients $(\hat{\bm\beta})$ from four spatial regression models fitted on $\calD (n = 947)$ and 95\% credible intervals. Bold entries indicate non-intercept coefficients whose 95\% credible intervals exclude zero.}
\label{table:mmiresult_n947}
\centering
\begin{tabular}{lcccc}
\toprule
 & Beta & Cobin & Beta rectangular & Micobin \\
\midrule
(Intercept) & -1.829 (-3.621, -0.070) & -1.756 (-3.531, 0.006) & -1.696 (-3.543, 0.092) & -1.741 (-3.520, 0.018) \\
agkffact    & \textbf{-3.088} (-5.982, -0.206)& \textbf{-3.150} (-6.019, 0.258) & \textbf{-3.082} (-5.961, -0.200) & \textbf{-3.494} (-6.220, -0.822) \\
bfi         & 0.244 (-0.075, 0.568) & 0.228 (-0.088, 0.552) & 0.220 (-0.101, 0.551) & 0.219 (-0.097, 0.540) \\
cbnf        & 0.191 (-0.044, 0.430) & 0.200 (-0.035, 0.437) & 0.200 (-0.036, 0.437) & 0.197 (-0.034, 0.424) \\
conif       & \textbf{0.096} (0.014, 0.175)& \textbf{0.103} (0.021, 0.183) & \textbf{0.103} (0.023, 0.184) & \textbf{0.125} (0.045, 0.204)\\
crophay     & -0.057 (-0.223, 0.110) & -0.053 (-0.218, 0.114) & -0.042 (-0.209, 0.124) & -0.050 (-0.210, 0.110) \\
fert        & -0.096 (-0.327, 0.135) & -0.104 (-0.329, 0.122) & -0.120 (-0.348, 0.105) & -0.089 (-0.316, 0.135) \\
manure      & -0.001 (-0.148, 0.148) & -0.009 (-0.157, 0.138) & -0.013 (-0.165, 0.137) & -0.022 (-0.167, 0.122) \\
pestic97    & -0.031 (-0.118, 0.057) & -0.030 (-0.119, 0.057) & -0.034 (-0.121, 0.054) & -0.027 (-0.110, 0.055) \\
urbmdhi     & \textbf{-0.180} (-0.283, -0.076) & \textbf{-0.169} (-0.275, 0.064)& \textbf{-0.172} (-0.278, -0.067)& \textbf{-0.143} (-0.242, -0.043) \\
\midrule
mESS$(\bm\beta)$/t & 3392.4/103 mins & 4748.7/4 mins & 3211.1/96 mins & 3012.6/5 mins \\
PSIS-LOO & -1125.2$(\calD)$ & -1134.7$(\calD)$ & -1124.8$(\calD)$ & -1126.9$(\calD)$ \\
WAIC     & -1134.7$(\calD)$ & -1142.4$(\calD)$ & -1131.3$(\calD)$ & -1131.1$(\calD)$ \\
\bottomrule
\end{tabular}
\vspace{-2mm}
\begin{flushleft}
{\scriptsize
agkffact, soil erodibility factor; bfi, base flow index; cbnf, cultivated biological N fixation; conif, coniferous forest cover; crophay, crop/hay land cover; fert, synthetic N fertilizer use; manure, manure application; pestic97, 1997 pesticide use; urbmdhi, medium/high-density urban land cover; mESS$(\bm\beta)$, t, multivariate effective sample size and wall-clock fitting time (in mins), averaged over 3 chains; PSIS-LOO, leave-one-out cross-validation estimate using Pareto smoothed importance sampling; WAIC, widely applicable information criterion. 
All variables are $\log_2(x+1)$ transformed. 
}
\end{flushleft}
\end{table}

\begin{table}[t]
\scriptsize
    \caption{Estimated average slopes (average marginal effects) from four spatial regression models fitted on $\calD^\circ (n = 949)$ and 95\% credible intervals. Bold entries indicate the average slopes whose 95\% credible intervals exclude zero.}
    \label{table:mmiresult_avgslope}
    \centering
\begin{tabular}{l c c c c  }
\toprule
 & Beta regression  & Cobin regression & Betarec regression & Micobin regression\\
\midrule
agkffact & -0.196 (-0.422, 0.025) & \textbf{-0.219} (-0.433, -0.000) & \textbf{-0.231} (-0.449, -0.015) & \textbf{-0.262} (-0.463, -0.061)  \\
bfi & \textbf{0.026} (0.001, 0.051) & 0.022 (-0.002, 0.046) & 0.016 (-0.008, 0.040) & 0.017 (-0.006, 0.041) \\
cbnf & 0.012 (-0.006, 0.031) & 0.014 (-0.004, 0.032) & 0.015 (-0.003, 0.033) & 0.014 (-0.003, 0.032)  \\
conif & 0.006 (-0.000, 0.012) & \textbf{0.007} (0.001, 0.013) & \textbf{0.008} (0.002, 0.014)& \textbf{0.009} (0.003, 0.015) \\
crophay & -0.006 (-0.019, 0.007) & -0.005 (-0.018, 0.008) & -0.003 (-0.016, 0.010) & -0.004 (-0.016, 0.008) \\
fert & -0.006 (-0.023, 0.012) & -0.007 (-0.024, 0.010) & -0.010 (-0.027, 0.007) & -0.006 (-0.023, 0.010)  \\
manure & -0.004 (-0.015, 0.008) & -0.003 (-0.014, 0.009) & -0.001 (-0.013, 0.010) & -0.002 (-0.013, 0.009) \\
pestic97 & -0.001 (-0.008, 0.006) & -0.002 (-0.008, 0.005) & -0.003 (-0.009, 0.004) & -0.002 (-0.008, 0.004)  \\
urbmdhi & \textbf{-0.014} (-0.022, -0.006) & \textbf{-0.013} (-0.021, -0.005) & \textbf{-0.013} (-0.021, -0.005) & \textbf{-0.011} (-0.018, -0.003) \\
\bottomrule
\end{tabular}
    \vspace{-3mm}
\begin{flushleft} 
{\scriptsize 
agkffact, soil erodibility factor; bfi, base flow index; cbnf, cultivated biological N fixation; conif, coniferous forest cover; crophay, crop/hay land cover; fert, synthetic N fertilizer use; manure, manure application; pestic97, 1997 pesticide use; urbmdhi, medium/high-density urban land cover.
}
\end{flushleft}
\end{table}

Tables~\ref{table:mmiresult_n950} and~\ref{table:mmiresult_n947} summarize results based on the full dataset $\calD^\star$ ($n = 950$) and the reduced dataset $\calD$ ($n = 947$), respectively. Based on the full dataset results in Table~\ref{table:mmiresult_n950}, compared to the beta rectangular regression model, the micobin regression model exhibits greater stability of regression coefficients, improved predictive performance in terms of PSIS-LOO and WAIC on full dataset, and substantial computational advantages. 

Table~\ref{table:mmiresult_avgslope} reports results based on the dataset $\calD^\circ$ ($n = 949$) in terms of average slopes, also known as average marginal effects \citep{Williams2012-jz, Arel-Bundock2024-ee}. These quantities are computed as the empirical average of the partial derivative $\frac{\partial E(y \mid \bm{x})}{\partial x_j} \big|_{\bm{x} = \bm{x}_i}$ across the observed data points $i = 1, \dots, n$ for a given covariate $x_j$. Average slopes are often interpreted as an approximation of an average change of $E(y\mid \bmx)$ associated with a unit change of $x_j$, holding other covariates constant \citep{Norton2019-ur}. As an example, the estimated average marginal effect of urbmdhi from the micobin regression is -0.011. This can be interpreted as doubling the percentage of medium/high-density urban land cover is associated with a decrease of MMI by 1.1 on average (on a 0-100 scale), holding the other covariates constant.

Figure~\ref{fig:mmipredict_comparison} compares spatial MMI predictions at 55,215 unsampled lakes from each model against those from the spatial micobin regression. The results show that the predictions of the spatial beta regression model are more influenced by the two lakes with low MMI, leading to a substantial underestimation of MMI of almost 14 units (on the original scale of 0-100), compared to the beta-rectangular and micobin. In contrast, the spatial cobin regression model is noticeably less influenced by these two lakes. Figure~\ref{fig:trace} shows trace plots of log-likelihood conditional on random effects from four models based on $\calD^\circ$, which form the basis for computing PSIS-LOO and WAIC.

\begin{figure}
    \centering
    \includegraphics[width=\linewidth]{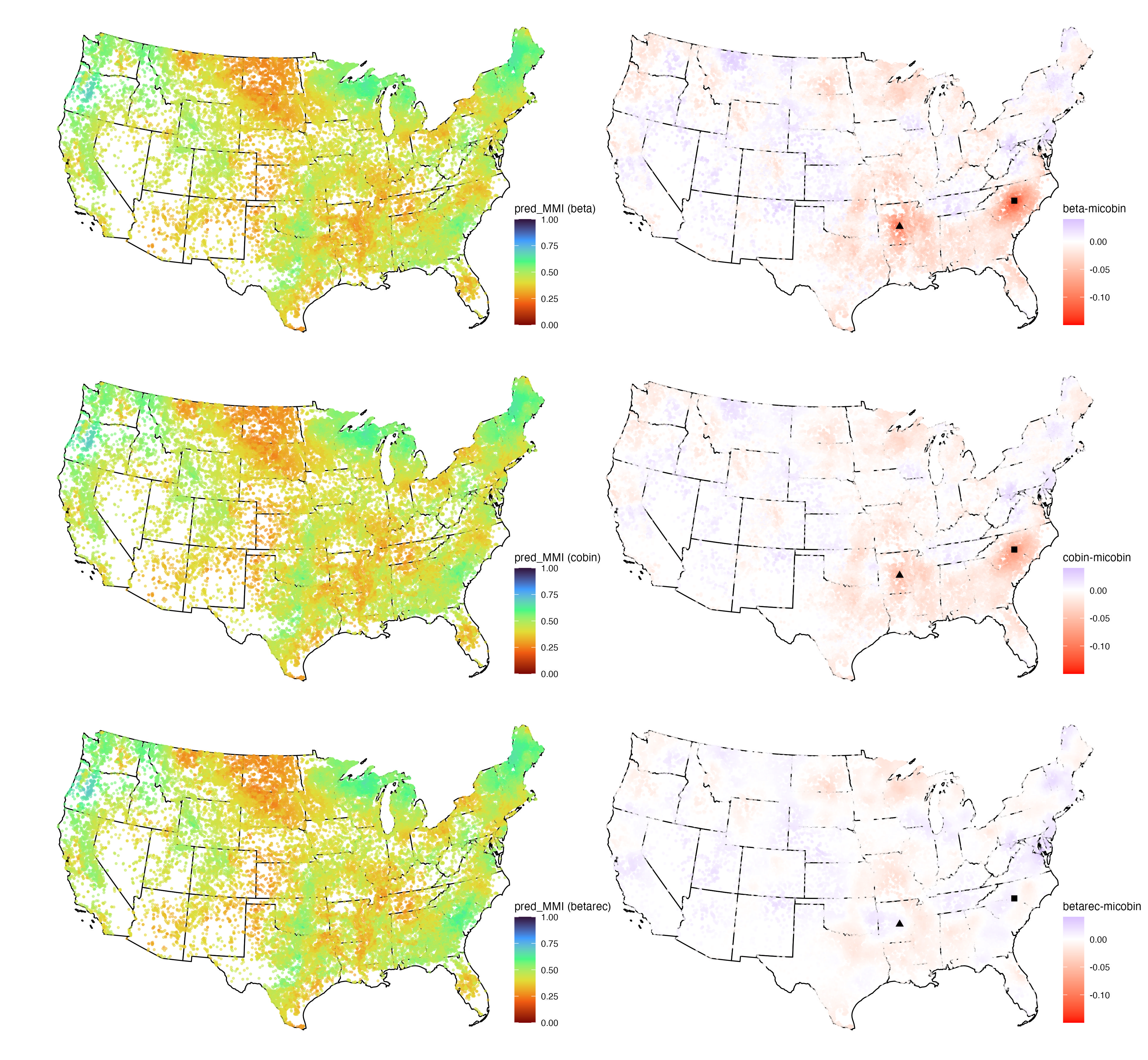}
    \caption{(Left) Predicted MMI at 55,215 lakes from the spatial beta (left top), spatial cobin (left middle), and spatial beta rectangular (left bottom) regression model based on $\calD^\circ$ in terms of posterior predictive mean of $E\{Y(s^*) \mid X(s^*)\}$ at unsampled location $s^*$. (Right) Differences in predicted MMI between spatial micobin regression and each of the three models on the left, respectively. See Figure~\ref{fig:qresidplot} for the legend corresponding to triangle and square markers, indicating the location of two low MMI lakes.}
    \label{fig:mmipredict_comparison}
\end{figure}

\begin{figure}
    \centering
    \includegraphics[width=\linewidth]{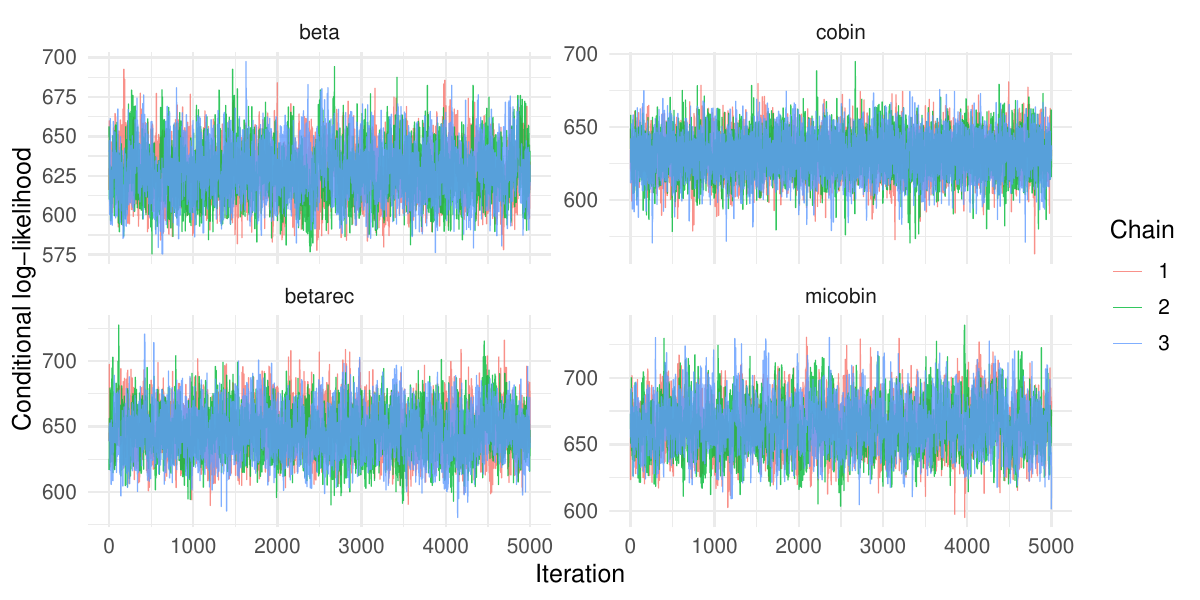}
\caption{Trace plots of the log-likelihood (conditional on spatial random effects) from the spatial beta, cobin, beta rectangular, and micobin regression models, based on dataset $\calD^\circ$. Each panel shows results from three MCMC chains.}
    \label{fig:trace}
\end{figure}

\section{Data augmentation and boundary-concentrated data}
\label{appendix:imbalanced}

\subsection{Bayesian cobin regression under boundary-concentrated settings}
The Kolmogorov-Gamma augmentation scheme proposed in this article shares many similarities with the P\'olya-Gamma augmentation scheme \citep{Polson2013-gb}, including its limitations.  
One important limitation of the P\'olya-Gamma augmentation for logistic models arises in highly unbalanced binary data settings, where the corresponding Markov chain often exhibits poor mixing. This issue is formally investigated in \citet{Johndrow2019-os}, where an ``infinitely imbalanced'' data setting \citep{Owen2007-lk} is considered, defined as $(y_1,y_2,\dots,y_n)\in\{0,1\}^n$ such that $\sum_{i=1}^ny_i = 1$ (1 success, $n-1$ failure) with growing $n$. 
In this setting, \citet{Johndrow2019-os} considered an intercept-only logistic model $\mathrm{logit}\{\pr(y_i=1)\} = \beta$ for $i=1,\dots,n$, and investigated the Markov chain arising from Polya-Gamma augmentation, targeting the posterior of $\beta$.

We illustrate through an example that the Markov chain based on the Kolmogorov-Gamma augmentation may exhibit poor mixing behavior when most of the data are concentrated on either 0 or 1. Without loss of generality, analogous to ``infinitely imbalanced'' binary data setting, we consider ``infinitely boundary-concentrated'' setting, defined as $(y_1,y_2,\dots,y_n)\in[0,1]^n$ satisfying $\sum_{i=1}^ny_i = 1$, so that the empirical average of response is $1/n$ with $n$ datapoints. 
In this setting, we consider the intercept-only cobin regression model $Y_i\iidsim \mathrm{cobin}(\beta, \lambda^{-1})$ with cobit link and fixed nuisance parameter $\lambda$. Under the flat prior $p(\beta)\propto 1$, the posterior of $\beta$ given $y_{1:n}$ (such that $\sum_{i=1}^ny_i = 1$, where $\sum_{i=1}^ny_i$ is a sufficient statistic of $\beta$) is
\begin{equation}
\label{eq:post_boundary}
    p(\beta\mid y_{1:n}) = C\left[\frac{e^\beta}{\{(e^\beta-1)/\beta\}^n}\right]^\lambda, \quad \beta\in \bbR
\end{equation}
with $C^{-1} = \int_{-\infty}^\infty\left[\frac{e^\beta}{\{(e^\beta-1)/\beta\}^n}\right]^\lambda \rmd\beta$, which is a proper posterior as long as $n\ge 2$. Note that for fixed $n$, the posterior distribution becomes more concentrated as $\lambda$ increases, approximately by a factor of $\lambda^{-0.5}$, since increasing $\lambda$ by a factor of 2 is equivalent to observing two independent replicates of $y_{1:n}$, based on the form of the likelihood.

\begin{figure}[t]
    \centering
    \includegraphics[width=\linewidth]{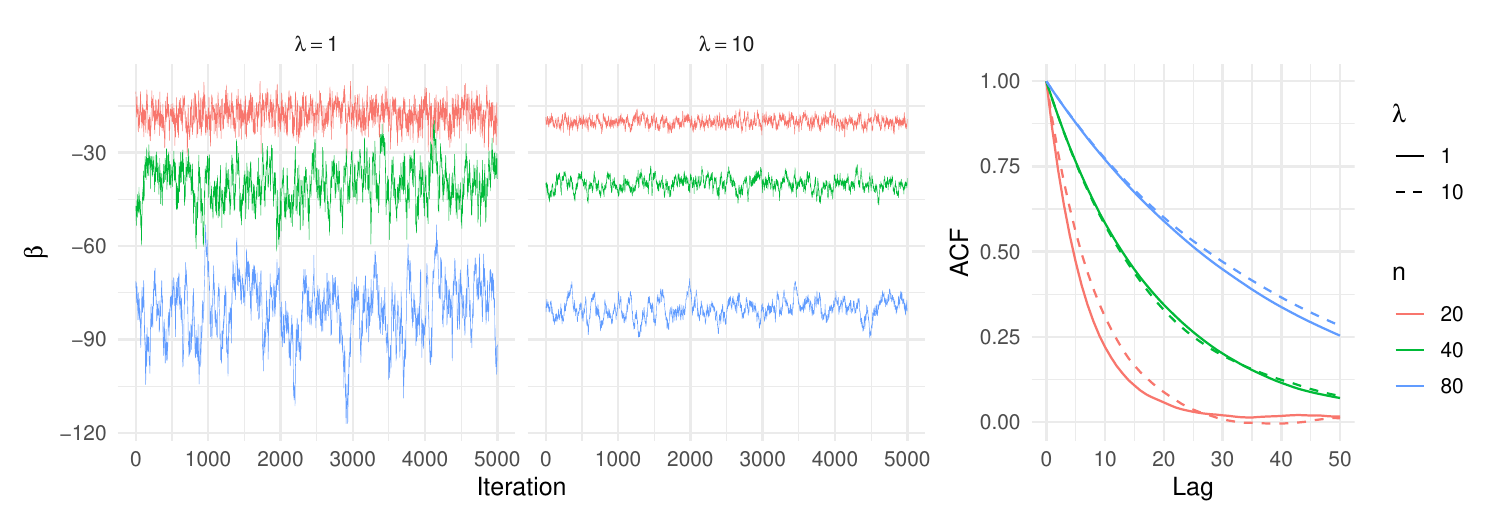}
    \caption{(Left) Examples of MCMC trace plots (only first 5,000 iterations shown) of $\beta$ for $p(\beta \mid y_{1:n})$ in \eqref{eq:post_boundary} with Kolmogorov-Gamma augmentation (Right) Autocorrelation function plots  with different $n$ and $\lambda$ settings.}
    \label{fig:boundaryproximate_trace}
\end{figure}

To sample $\beta$ from $p(\beta\mid y_{1:n})$ in \eqref{eq:post_boundary}, we apply Algorithm~\ref{alg:cobin} with Kolmogorov-Gamma augmentation, where $\lambda$ is fixed. One cycle of the MCMC algorithm from the state $(\beta^{(t)},\{\kappa_i^{(t)}\})$, targeting the posterior \eqref{eq:post_boundary}, consists of 
\begin{enumerate}
    \item $\kappa_i^{(t+1)}\mid \beta^{(t)}\iidsim \mathrm{KG}(\lambda, \beta^{(t)})$ for $i=1,\dots,n$, 
    \item $\beta^{(t+1)}\mid \{\kappa_i^{(t+1)}\}_{i=1}^n \sim N(\lambda(1-0.5n)/K, 1/K)$ where  $K = \sum_{i=1}^n\kappa_i^{(t+1)}$.
\end{enumerate}
For each setting of $n \in \{20, 40, 80\}$ and $\lambda \in \{1, 10\}$, we generate $y_{1:n}$ from $Y_i\iidsim \mathrm{cobin}(\beta, \lambda^{-1})$ and run the MCMC algorithm with a total of 101,000 iterations, discarding the first 1,000 samples as burn-in. 
Figure~\ref{fig:boundaryproximate_trace} presents trace plots of $\beta$ (showing only the first 5,000 iterations after burn-in for clearer visualization of mixing behavior) and the corresponding autocorrelation functions (ACF). Both the trace plots and ACF indicate that the mixing of $\beta$ deteriorates as $n$ increases. For a fixed $n$, increasing $\lambda$ results in a more concentrated posterior; however, it has a negligible effect on the mixing behavior of $\beta$ as shown in the ACF.

These results suggest that in intercept-only models, when the empirical average of data $n^{-1}\sum_{i=1}^n y_i$ becomes increasingly concentrated near the boundary (which leads to more extreme intercept estimates), the Kolmogorov-Gamma augmentation scheme leads to poor mixing of $\beta$. This behavior appears to be primarily driven by the proximity of the empirical average of the data to the boundary, regardless of the empirical variance of the data that affects the nuisance parameter $\lambda$.

\subsection{Key insights on poor mixing behavior and open problems}
We illustrate that the reason for poor mixing under the Kolmogorov-Gamma augmentation with highly boundary-concentrated data is analogous to that of the P\'olya-Gamma augmentation with highly imbalanced data. For the P\'olya-Gamma case, the main reason of the poor mixing is that the posterior distribution of $\beta$ covers a relatively wide region, while the typical step sizes of $\beta$ induced by the P\'olya-Gamma augmentation become increasingly narrower as $n$ grows \citep{Johndrow2019-os}. This scale mismatch makes the typical MCMC move size of $\beta$ too small relative to the bulk of the posterior, leading to slow mixing. We show that a similar phenomenon happens for the Kolmogorov-Gamma augmentation. First, we have the following results on the posterior distribution $p(\beta \mid y_{1:n})$ \eqref{eq:post_boundary} under an infinitely boundary-concentrated setting. We defer technical lemmas to the end of this section for better reading.

\begin{theorem}
\label{thm:mledensheight}
    Denoting $\hat\beta = \argmax_\beta p(\beta\mid y_{1:n})$, we have $p(\hat\beta\mid y_{1:n}) = \calO(1/\log n)$.   
\end{theorem}
\begin{proof}
    Lemma~\ref{lemma:mledensheight1} and 
    Lemma~\ref{lemma:mledensheight2} shows that posterior density $p(\beta\mid y_{1:n})$ is essentially flat around the interval of length $\Theta(\log n)$ that contains $\hat\beta$. This implies $p(\hat\beta \mid y_{1:n}) = \calO(1/\log n)$ since the posterior density must integrate to 1.
\end{proof}
\begin{figure}[t]
    \centering
    \includegraphics[width=\linewidth]{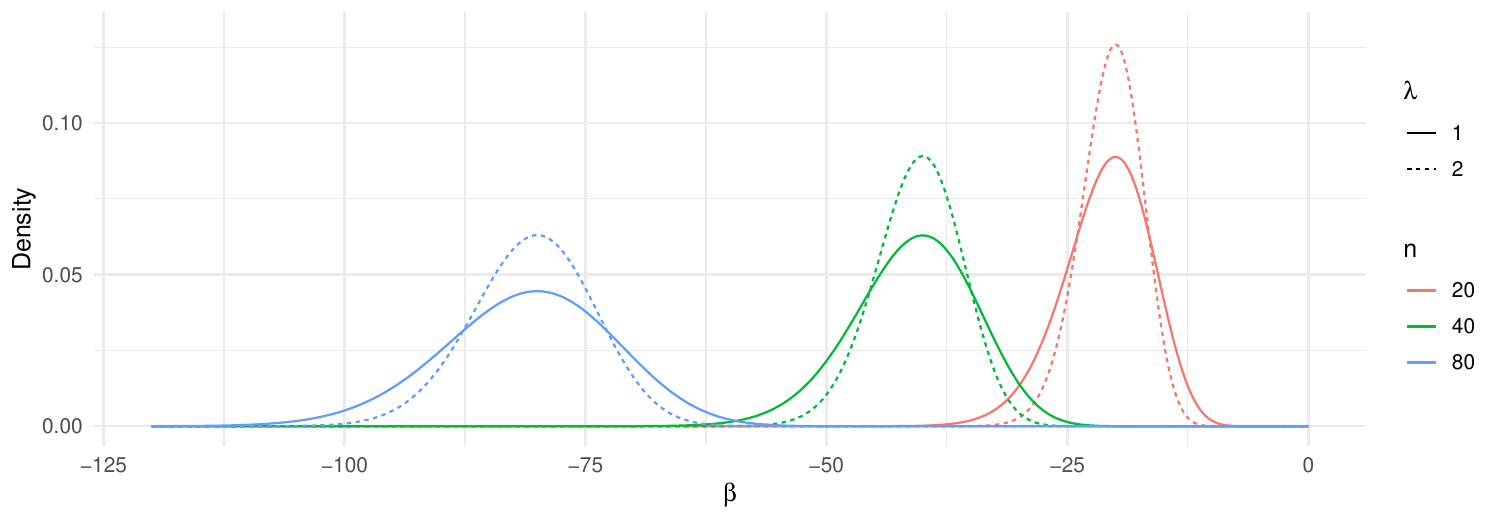}
    \caption{Posterior densities $p(\beta \mid y_{1:n})$ in \eqref{eq:post_boundary} under infinitely boundary-concentrated scenarios ($\sum_{i=1}^ny_i = 1$) for $n\in\{20,40,80$\} and $\lambda \in \{1,2\}$. Normalizing constants are calculated numerically.}
    \label{fig:boundaryproximate_target}
\end{figure}

Theorem~\ref{thm:mledensheight} implies that the bulk of the posterior covers the interval of length $\Omega(\log n)$. See Figure~\ref{fig:boundaryproximate_target} for visualization of $p(\beta\mid y_{1:n})$ with $\sum_{i=1}^n y_i = 1$ when $n\in\{20,40,80\}$ and $\lambda = 1$ or $\lambda = 2$.

Recall that one cycle of the MCMC algorithm for the model $Y_i\iidsim \mathrm{cobin}(\beta, \lambda^{-1})$ with $\sum_{i=1}^n y_i = 1$ targeting the posterior \eqref{eq:post_boundary}, consists of steps 
\begin{enumerate}
    \item $\kappa_i^{(t+1)}\mid \beta^{(t)}\iidsim \mathrm{KG}(\lambda, \beta^{(t)})$ for $i=1,\dots,n$, 
    \item $\beta^{(t+1)}\mid \{\kappa_i^{(t+1)}\}_{i=1}^n \sim N(\lambda(1-0.5n)/K, 1/K)$ where  $K = \sum_{i=1}^n\kappa_i^{(t+1)}$.
\end{enumerate}

When $\beta^{(t)}$ is close to the posterior mode (e.g. $\beta^{(t)} = -n$), the mean and variance of $K = \sum_{i=1}^n\kappa_i^{(t+1)}$ satisfy $
E(K\mid \beta^{(t)} = -n) = \lambda(0.5 - 1/n) + \calO(ne^{-n})\approx \lambda(0.5-1/n)$ and $\var(K\mid \beta^{(t)} = -n) = \lambda/(2n^2) - 2\lambda/n^3 + \calO(n^{-1}e^{-n})$; see Lemma~\ref{lemma:mledensheight3}. 
Note that $\var(K\mid \beta^{(t)} = -n)$ quickly shrinks to 0 as $n$ grows. Conditioning on $K=\lambda(0.5-1/n)$, we have 
\[
\beta^{(t+1)}\mid K=\lambda(0.5-1/n)\sim N\left(-n, \frac{2}{\lambda}\frac{n}{n-2} \right).
\]
This means that, approximately for large $n$, the marginal transition kernel $k(\beta^{(t+1)}\mid\beta^{(t)})$ around the posterior mode ($\beta^{(t)}=-n$) involves step sizes proportional to $\lambda^{-0.5}$. This is in contrast to the fact that, as shown in Theorem~\ref{thm:mledensheight}, the bulk of the posterior $p(\beta \mid y_{1:n})$ has increasing width with at least a $\log n$ factor.
This creates a mismatch between the typical step size of the Markov chain of $\beta$ and the target distribution’s asymptotic spread. For fixed $n$, this mismatch does not appear to vanish even when the posterior becomes more concentrated with increasing $\lambda$, since larger $\lambda$ simultaneously reduces the step size at approximately the same rate, i.e., by a factor of $\lambda^{-0.5}$, as the posterior contracts. In fact, although the autocorrelation of $\beta$ remain similar, it further deteriorates the efficiency of the Markov chain, as the time required to sample a $\mathrm{KG}(\lambda, c)$ random variable based on the sum of $\lambda$ i.i.d. $\mathrm{KG}(1, c)$ variables increases linearly with $\lambda$. 
This phenomenon suggests several imporant and promising directions for future research, analogous to recent advances in improving data augmentation MCMC for imbalanced binary and categorical data \citep{Duan2018-zl, Johndrow2019-os, Zens2024-nb}.

\subsection{Technical lemmas}

\begin{lemma} 
\label{lemma:mledensheight1}
For $n\ge 2$, $p(\beta \mid y_{1:n})$ is unimodal and attains maximum at $\hat{\beta} = - n + \calO(1)$.
\end{lemma}

\begin{proof} First note that $p(\beta \mid y_{1:n})>0$ for any $\beta$, i.e. the posterior is supported on $\bbR$.  
    By direct calculation, the derivative of $p(\beta \mid y_{1:n})$ with respect to $\beta$ satisfies $p'(\beta \mid y_{1:n}) = \lambda f(\beta)p(\beta \mid y_{1:n})$ with $f(\beta) = 1-n\{1-\beta^{-1}+1/(e^\beta-1)\}=1-ng_{\mathrm{cobit}}^{-1}(\beta)$. Since $f'(\beta) = -n/\beta^2+ (ne^\beta)/(e^\beta-1)^2 <0$ for all $\beta$, $f(\beta) = 0$ has the unique solution $\hat{\beta}$, which satisfies $g_{\mathrm{cobit}}^{-1}(\hat{\beta})=1/n$. Since $f(-n-1)>0$ and $f(-n+2.1)<0$ for $n\ge 2$, $p(\beta \mid y_{1:n})$ is unimodal and $\hat{\beta} = -n + \calO(1)$.   
\end{proof}

\begin{lemma} 
\label{lemma:mledensheight2}
Let $\beta_1 = -n + r \log(n)$ and $\beta_2 = -n - r\log(n)$ where $r$ is some arbitrary constant satisfying $\hat{\beta}\in (\min(\beta_1,\beta_2), \max(\beta_1,\beta_2))$. Then $p(\beta_1\mid y_{1:n})/p(\beta_2 \mid y_{1:n}) \to 1$ as $n\to\infty$. 
\end{lemma}
\begin{proof} We have 
    \begin{align}
        \frac{p(\beta_1\mid  y_{1:n})}{p(\beta_2\mid  y_{1:n})} &= \left\{e^{\lambda(\beta_1-\beta_2)}\left(\frac{\beta_1}{\beta_2}\right)^{\lambda n} \right\}\left\{\left(\frac{e^{\beta_2}-1}{e^{\beta_1}-1}\right)^{\lambda n}\right\}
    \end{align}
We prove the statement by showing that the two curly bracket terms above both converge to 1. For the first term, using $\beta_1-\beta_2 = 2r \log(n)$, and  
\[
\log(\beta_1/\beta_2) = \log(1-rn^{-1}\log n) - \log ( 1+ rn^{-1}\log n) = -\frac{2r\log n}{n} - \calO\left(\frac{\log^3 n}{n^3}\right)
\]
for sufficiently large $n$ such that $|rn^{-1}\log n|< 1$ using Taylor expansion $\log(1+x) = x-x^2/2+x^3/3-\dots$ for $|x|<1$. Thus
\begin{align*}
\left\{e^{\beta_1-\beta_2}\left(\frac{\beta_1}{\beta_2}\right)^n\right\}^\lambda = \left\{n^{2r}\times \exp\left(n \log \frac{\beta_1}{\beta_2}\right)\right\}^\lambda =n^{2\lambda r}n^{-2\lambda r} \exp\left( \calO(\frac{\log^3 n }{n^2})\right)\to 1 
\end{align*}

For the second term, using $e^{\beta_1} = n^re^{-n}$ and $e^{\beta_2} = n^{-r}e^{-n}$ so that $|(e^{\beta_1}-e^{\beta_2})/(1-e^{\beta_1})| = e^{-n}|n^r - n^{-r}| / |1-e^{\beta_1}|= \calO(e^{-n}n^{|r|})$. Since 
\[
\left(\frac{e^{\beta_2}-1}{e^{\beta_1}-1}\right)^n = \exp\left( n \log \left(\frac{e^{\beta_2}-1}{e^{\beta_1}-1}\right) \right) = \exp\left( n \log\left(1+\frac{e^{\beta_1}-e^{\beta_2}}{1-e^{\beta_1}}\right) \right),
\]
using $|\log(1+x)|\le 2|x|$ for $|x|<1/2$, we have $n |\log\left(1+\frac{e^{\beta_1}-e^{\beta_2}}{1-e^{\beta_1}}\right)|\le 2n |\frac{e^{\beta_1}-e^{\beta_2}}{1-e^{\beta_1}}|\to 0$, which implies $\left(\frac{e^{\beta_2}-1}{e^{\beta_1}-1}\right)^n \to 1$, thus $\left(\frac{e^{\beta_2}-1}{e^{\beta_1}-1}\right)^{\lambda n} \to 1$. 
\end{proof}

\begin{lemma}
\label{lemma:mledensheight3}
    If $\kappa \sim \mathrm{KG}(b,c)$ and $c\neq 0$, we have 
\begin{align*}
 E(\kappa) &=b c^{-2}\{(c/2)\coth(c/2) - 1\}\\
\var(\kappa) &= \frac{b}{4c^4}\mathrm{csch}^2(c/2) (c^2 + c \sinh(c) - 4 \cosh(c) + 4)
\end{align*}
If $c=0$, $E(\kappa) = b/12$ and $\var(\kappa) = b/360$. For large $c>0$, recalling $\mathrm{KG}(b,c)\stackrel{d}{=}\mathrm{KG}(b,-c)$,
\begin{align*}
 E(\kappa) &=b (0.5c^{-1}-c^{-2}) + \calO(c^{-1}e^{-c})\\
\var(\kappa) &= b(0.5c^{-3}-2c^{-4})+\calO(c^{-2}e^{-c})
\end{align*}
\end{lemma}
\begin{proof}
These are derived straightforwardly from the definition of the Kolmogorov-Gamma as an infinite convolution of gammas, using $\sum_{k=1}^\infty 1/\{k^2+c^2/(4\pi^2)\} = (\pi^2/c)\coth(c/2) - 2\pi^2/c^2$ \citep[][\S 1.421.4]{Gradshteyn2014-ym}. Asymptotic derivations are straightforwardly follow from $(c/2)\coth(c/2) = c/2 + ce^{-c}/(1-e^{-c})$, $\mathrm{csch}^2(c/2) = 4e^c/(e^c-1)^2 = 4e^{-c}(1+2e^{-c}+3e^{-2c}+\cdots)$.
\end{proof}

\end{document}